\documentclass[pre,aps,showpacs,showkeys,twocolumn]{revtex4-1}

\usepackage{amsmath}
\usepackage{amsthm}
\usepackage{graphicx}
\usepackage[colorlinks=true,linkcolor=blue,citecolor=blue]{hyperref}
\usepackage{enumerate}
\usepackage{amstext}
\usepackage{amsfonts}
\usepackage{amssymb,amscd}
\usepackage{amsopn}
\usepackage{bm,bbm}
\usepackage{epsfig,epic}
\usepackage{epstopdf}
\graphicspath{{./Images/}}

\newcommand{\newc}{\newcommand}
\newc{\beq}{\begin{equation}}
\newc{\eeq}{\end{equation}}
\newc{\kt}{\rangle}
\newc{\br}{\langle}
\newc{\beqa}{\begin{eqnarray}}
\newc{\eeqa}{\end{eqnarray}}
\newc{\pr}{\prime}
\newc{\longra}{\longrightarrow}
\newc{\ot}{\otimes}
\newc{\rarrow}{\rightarrow}
\newc{\h}{\hat}
\newc{\bom}{\boldmath}
\newc{\btd}{\bigtriangledown}
\newc{\al}{\alpha}
\newc{\be}{\beta}
\newc{\ld}{\lambda}
\newc{\sg}{\sigma}
\newc{\p}{\psi}
\newc{\eps}{\epsilon}
\newc{\om}{\omega}
\newc{\mb}{\mbox}
\newc{\tm}{\times}
\newc{\hu}{\hat{u}}
\newc{\hv}{\hat{v}}
\newc{\hk}{\hat{K}}
\newc{\ra}{\rightarrow}
\newc{\non}{\nonumber}
\newc{\ul}{\underline}
\newc{\hs}{\hspace}
\newc{\longla}{\longleftarrow}
\newc{\ts}{\textstyle}
\newc{\f}{\frac}
\newc{\df}{\dfrac}
\newc{\ovl}{\overline}
\newc{\bc}{\begin{center}}
\newc{\ec}{\end{center}}
\newc{\dg}{\dagger}
\newc{\tr}{\mbox{Tr}}

\newtheorem{proposition}{Proposition}[]
\newtheorem{lemma}{Lemma}[]

\newcommand\norm[1]{\left\lVert#1\right\rVert}
\newcommand\ip[1]{\left\langle#1\right\rangle}
\newcommand\ket[1]{\left|#1\right\rangle}

\newcommand\op[2]{\left|#1\right\rangle\left\langle#2\right|}

\theoremstyle{definition}
\theoremstyle{remark}

\begin{document}

\title{Tripartite mutual information, entanglement, and scrambling in permutation symmetric systems with an application to quantum chaos}

\author{Akshay Seshadri}
\author{Vaibhav Madhok}
\author{Arul Lakshminarayan}
\email[]{arul@physics.iitm.ac.in}
\affiliation{Department of Physics, Indian Institute of Technology Madras, Chennai 600036, India}
\date{}

\begin{abstract}

Many-body states that are invariant under particle relabelling, the permutation symmetric states, occur naturally when the system dynamics is described by symmetric processes or collective spin operators. We derive expressions for the reduced density matrix for arbitrary subsystem decomposition for these states and study properties of permutation symmetric states and their subsystems when the joint system is picked randomly and uniformly. Thus defining a new random matrix ensemble, we find the average linear entropy and von Neumann entropy which implies that random permutation symmetric states are marginally entangled and as a consequence the tripartite mutual information (TMI) is typically positive, preventing information from being shared globally. Applying these results to the quantum kicked top viewed as a multi-qubit system we find that entanglement, mutual information and TMI all increase for large subsystems across the Ehrenfest or log-time and saturate at the random state values if there is global chaos. During this time the out-of-time order correlators (OTOC) evolve exponentially implying scrambling in phase space. We discuss how positive TMI may coexist with such scrambling.

\end{abstract}

\maketitle

\section{Introduction}

Classical dynamical systems have a hierarchy of complexity from ergodic to mixing and $K-$ systems \cite{arnol1968ergodic}. Classical Hamiltonian systems that are nonintegrable are capable of displaying the highest amount of deterministic randomness dubbed chaos. Quantum chaos aims to address and extend these questions to the quantum domain \cite{Gutzwiller90, Haake}. Seminal works in this regard range from semiclassical methods utilizing classical unstable periodic orbits \cite{gutzwiller1971periodic}, energy level fluctuations and the use of random matrix theory \cite{Berry375, bohigas1984, robnik1985classical}, characteristics of semiclassical Wigner functions of chaotic eigenfunctions \cite{Berry77a, voros1976semi}. Study of model systems from two-dimensional billiards \cite{McDonald, heller1984bound} to quantum maps \cite{BERRY197926} was crucial in this development and  brought up phenomena including scarring of eigenfunctions and dynamical localization. Most of these works have a strong focus on the time-independent Schr{\"o}dinger's equation and the properties of the stationary states of a single particle. 

Few and many-particle quantum chaos has received attention, largely from a spectral statistics point of view, but also including dynamics and entanglement generation studies \cite{aberg1990onset, montambaux1993quantum, hsu1993level, van1994two, kudo2003unexpected, rabson2004crossover,jacquod1997emergence, georgeot1998integrability, LakSub2003, avishai2002level, santos2004integrability, kudo2004level, LakSub2005, KarthikSharmaLak2007}. The connections with topics ranging from thermalization in closed systems to information scrambling are currently being vigorously explored \cite{Rigol2016, MaldacenaSYK15}. Although nuclear physicists have long developed techniques such as the two-body-random matrix and embedded ensembles \cite{bohigas1971spacing,Kota} to deal with the spectral statistics of many-body systems. Ironically the nuclear many-body physics which motivated the canonical random matrix ensembles \cite{mehta2004random} is most relevant for single particle chaotic systems. With the advent of controlling quantum systems, evolving and measuring them, dynamical aspects of few and many-body physics and chaos have taken the center stage. Two recent experiments that preserve the purity of complex time-evolving states illustrate the richness of this domain \cite{Neill16, kaufman2016quantum}.

Quantum information science has added a new perspective and a new set of questions to the study of quantum chaos. Here, one is naturally led to consider how the dynamical generation of entanglement, quantum discord and other information theoretic quantities between quantum subsystems is connected with the chaotic dynamics of coupled classical degrees of freedom \cite{miller1999signatures, lakshminarayan2001entangling, bandyopadhyay2002testing, ScottCaves2003, jacquod2004semiclassical, Jacq2004Erratum, PetitJacq2006, Ghose, Wang2004, trail2008entanglement, madhok2015signatures, madhok2014information}. In this regard, the tensor product structure of quantum mechanics, essential for understanding systems with multiple degrees of freedom is crucial \cite{miller1999signatures, trail2008entanglement, lakshminarayan2001entangling, bandyopadhyay2002testing}. Attempts to address these questions has resulted in a better understanding of quantum phenomena like entanglement and decoherence by connecting time evolved states under quantum chaotic Hamiltonians to properties of random states.

Such studies address fundamental issues of complexity in quantum systems and are potentially applicable in quantum information processing, where quantum correlations like entanglement and quantum discord are considered to be a crucial resource. More recently, out-of-time-order correlators (OTOCs), which are linked to how large the commutator of observables can grow with time, have been studied extensively for their connection with chaos as well as ``information scrambling", and in particular, their growth till the Ehrenfest time has been investigated \cite{braunstein2013better, hartman2013time, shenker2014black, sekino2008fast, Lashkari2013, Maldacena2016, hosur2016chaos}. On similar lines, tripartite mutual information (TMI) has been associated with delocalization of information, or scrambling \cite{IyodaSagawa}, and connections between TMI and OTOC have been explored in this context \cite{hosur2016chaos}. It is conjectured that black holes are the fastest scramblers, and perhaps therefore, the most quantum chaotic of systems. Some studies argue that a negative TMI implies the scrambling of quantum information \cite{hosur2016chaos, IyodaSagawa}.

Many-body systems have an intimate connection with chaos, in as much as it arises due to non-integrability resulting from having fewer constants of motion than the total number of degrees of freedom. In the quantum domain, this coupling also results in quantum correlations between the local subsystems. Many quantum systems that are nonintegrable and display chaos, such as the currently intensely investigated Sachdev-Ye-Kitaev (SYK) model do not have apparent classical limits. However there are many-body systems whose collective observables have a classical limit that has a conventional integrability-chaos transition.  

In our study, we employ one such model where a multi-qubit system is collectively modeled as a kicked top  \cite{Haake, kus1987symmetry, kus88, Zyczkowski90} that transitions from regular to chaotic behaviour with suitable choice of parameters. Such Floquet, periodically forced systems which have natural realizations in quantum circuits \cite{emerson2003pseudo}, can give rise to a variety of dynamical features allowing us to explore connections between complexity of quantum chaos, properties of random states and dynamical generation of tripartite mutual information, OTOCs and entanglement. Like the SYK model, the kicked-top, considered as a many-body system, also involves all-to-all interaction of qubits, but the crucial difference is that the kicked top is not disordered, the source of the chaos in the system is from both an external uniform magnetic field and from periodic driving. 

The lack of disorder is central for enabling collective variables and for restricting dynamics to a subspace that is permutation symmetric.   
Thus in this paper, we consider in detail an ensemble of pure states that are uniformly selected from the permutation symmetric $N + 1$ dimensional subspace, that is {\it random permutation symmetric states} of $N$ qubits. We point out how to formulate the reduced density matrix for an arbitrary number of qubits in such states. A $Q$ qubit reduced density matrix can be written as a $(Q+1) \times (Q+1)$ dimensional matrix rather than a $2^Q$ dimensional one, and therefore we can scale up to large number of qubits easily. More crucially it implies that $Q$ qubit subsystems have an entropy, and hence entanglement, that cannot be larger than $\ln(Q+1)$ and in particular we show that typical states have an entropy that differs from this by a small number. Thus the states have an entanglement that is not a ``volume law" (proportional to $Q$), but more akin to critical spin chains \cite{Vidal2003} which follow a so-called area law \cite{latorre2005entanglement, popkov2005entangling, popkov2005logarithmic, barthel2006entanglement, vidal2007entanglement}. Thus there is qualitatively much less entanglement in random $N$~qubit permutation symmetric states than in generic states.

We study the eigenvalue distributions of the reduced density matrices of random permutation symmetric states and compare with the Marchenko-Pastur distribution which is valid for generic states \cite{marvcenko1967distribution}. In particular, for arbitrary sized subsystems we give analytical results for the average linear entropy of entanglement and provide estimates for the mutual informations based on linear entropy, as well as for the von Neumann entropy. While time-evolving states in a chaotic situation tends to these random states, we also study dynamical generation of quantum correlations as we evolve coherent states through repeated applications of the kicked top unitary. Dynamical behaviour of correlations like entanglement as a function of chaos for the kicked top viewed as a systems of qubits, have been undertaken in the past \cite{Ghose, Wang2004, madhok2015signatures, Madhok2018_corr}. However, such studies have largely focussed on single qubit and two qubit subsystems of the joint permutationally invariant state under evolution. In contrast, we consider subsystems containing arbitrary number of qubits and study their relevant correlations under chaotic dynamics.

Interestingly, we find that the dynamical behaviour of TMI is very similar to the behaviour of bipartite mutual information, which we refer to as mutual information (MI), and entanglement between subsystems under consideration. We also find that despite chaos in the system, TMI is positive for most states in the permutation symmetric subspace. Indeed, by applying L{\'e}vy's lemma to permutation symmetric systems, we show that most states will have a positive TMI when there are a large number of qubits in the system. This confirms the previous finding that TMI can be positive or negative for both integrable as well as non-integrable systems \cite{IyodaSagawa}, where non-integrable spin chains have been shown to result in positive TMI for a class of states that are prepared in the all-up state, which incidentally is obviously permutation symmetric. Thus it also appreciated that the sign of the TMI is dependent on the type of states that are evolved.

The positive TMI in the present work is a reflection of the permutation symmetry of the states and is connected, as we will show below, to the area-law scaling of the entanglement. Although the TMI is positive, the OTOCs as defined with symmetric collective operators can grow exponentially as we demonstrate. We compare the behaviour of TMI with that of OTOCs to further underline that TMI, sometimes defined as a metric of ``scrambling of quantum information", captures different aspects of quantum dynamics than the OTOCs. Another salient feature that comes from the study of permutation symmetric states is that for large enough subsystems the value of many correlations, such as the MI, TMI, and entanglement saturate to the permutation symmetric random state value after the Ehrenfest time that scales as $\ln(N)$. Thus while the sign of the TMI is not a quantum signature of classical chaos, the time it takes to saturate could be considered as one.

For completeness we define the TMI between 3 subsystems $A,B$ and $C$ as
\beq 
\label{eqn:tmi_from_mi}
I_3(A:B:C) = I_2(A:B) + I_2(A:C) - I_2(A:BC), 
\eeq
where $I_2(X:Y) = H(X) + H(Y) - H(XY) = H(X) - H(X|Y)$ is the MI between $X$ and $Y$, and $H(\cdot)$ is the Shannon entropy classically and von Neumann entropy quantum mechanically. A negative TMI implies the joint system $BC$ contains more information about the ``input" system $A$ than the  subsystems $B$ and $C$ individually. This is the classic case of when the whole is more than the parts. A related interpretation is that a negative TMI implies that the mutual information is monogamous, while a positive TMI implies that the same information is being shared by other parties. See \cite{IyodaSagawa, rangamani2015entanglement, rota2016tripartite} for further elucidations and insights into this quantity.

In the case of kicked tops, since there is a mapping to qubits, we can talk about the TMI between the qubits and find $I_3(n_1, n_2, n_3)$ where $n_i$ are the number of qubits in three different non-overlapping subsystems. Due to permutation symmetry it does not matter which qubits belong to the partitions. For instance $I_3(1, 1, 1)$ is the TMI between any 3 qubits. While the TMI is sensitive to the nature of classical dynamics, chaos does not imply a negative TMI due to permutation symmetry. In contrast, in the absence of this symmetry it is easy to see that typical states have a negative TMI \cite{rangamani2015entanglement}.

The remainder of this paper is organized as follows. In sections II and III  we study the properties of permutation symmetric states and their subsystems, and the properties of random permutation symmetric states including the eigenvalue distributions of the reduced density matrices. Reviewing the essential ideas, we derive analytic expressions for the typical entanglement and information theoretic quantities of a random state when we are restricted to the permutationally invariant subspace of the full tensor product space. This is of relevance here given the symmetry of the system. In section IV, we discuss the dynamical behaviour of tripartite mutual information and OTOCs, and explore the behaviour of ensemble average of TMI. Our results are discussed and summarized in Sec V.

\section{Permutation symmetric states}
As has been noted in the introduction, previous studies on qubit entanglement in the kicked top has been restricted to one or two qubits at the most. On the other hand the most well-studied case of ``block entanglement" concerning entanglement of a large number of spins (typically one-half) with others, especially in the context of random permutation symmetric states, is largely unexplored in this model. Recent works like \cite{moreno2018all} consider the Schmidt decomposition of various bipartitions of Dicke states but do not deal with typicality, randomness and quantum chaos in this context. Also previous works such as \cite{popkov2012reduced,markham2011entanglement, devi2012majorana, bohnet2016partial, wang2002pairwise, baguette2014multiqubit} have studied the reduced density matrices of permutationally invariant systems and their entanglement. Our approach here is to focus on generic pure permutation symmetric states with a view of defining an ensemble of them that would be useful in studies of quantum chaos as for example in the kicked top which we subsequently analyze in detail.

This presents an opportunity to study random states restricted to permutation symmetric subspaces. While random states on the whole Hilbert space is well-studied essentially using methods of random matrix theory, ensembles within such  restrictions remain by and large open. Recently an experiment used three qubits ($j = 3/2$) to explore the kicked top Floquet unitary operator and claimed evidence for thermalization in the chaotic regime \cite{Neill16}. Thus a critical examination of the permutation symmetric subspace may also be warranted from such viewpoints. As the stationary and time evolved states of the kicked top are permutation symmetric states of $2j$ qubits, we consider in this section properties of such states that will interest us.

Consider a system of $N$ qubits. If this system has permutation symmetry, then the effective dimension of the system is $N + 1$ instead of $2^N$, an exponential difference. The ``standard" basis vectors for such a permutation symmetric $N$-qubit system is given by $N + 1$ orthonormal states known as Dicke states \cite{Dicke}. Say $\{\ket{i} = \ket{\text{binary expansion of } i} |\ 0 \leq i \leq 2^N - 1\}$ is the computational basis for the $N$-qubit system. The most natural basis for the permutation symmetric case is then obtained by taking appropriate linear combinations of the computational basis vectors. The states involved in any particular (permutation symmetric) basis vector must have the same number of zeroes and ones, so that the linear combination is invariant under an arbitrary permutation of the qubits.
These are given as follows.
\beq
\label{defn:stdpermsymcts} 
    \ket{m_N} = \frac{1}{c_N(m)} \sum_{\substack{0 \le i \le 2^N-1\\ w(i)=m}} \ket{i},\ 0 \leq m \leq N
\eeq
where,  $w(i)$ is the Hamming weight of $i$, which is the number of    1 in the binary expansion of $ i$. The normalization constant  is
\beq
 c_N(m) = \sqrt{N\choose m}=\sqrt{\frac{N!}{m! (N-m)!}}.
\eeq
It can be easily verified that the $N+1$ Dicke states are orthonormal and indeed permutation symmetric. An arbitrary $N$-qubit permutation symmetric pure state can be written as
\begin{equation}
    \ket{\psi} = \sum_{m=0}^{N} a_m \ket{m_N}, \;\;\sum_{m=0}^{N} |a_m|^2=1.
\label{eq:permpsi1}
\end{equation}
All states of the kicked top, eigenstates or time-evolving ones, viewed as a multi-qubit one are of this form, with $N = 2j + 1$.

In order to compute the $Q$-qubit ($Q < N$) reduced density matrix  we wish to write this state in terms of tensor products of Dicke states corresponding to the $Q$-qubit and the $(N - Q)$-qubit subsystems. That is, we want a $(Q + 1) \times (N - Q + 1)$ dimensional ``coefficient matrix" $A$ such that
\begin{equation}
    \label{eq:permpsi2}
    \ket{\psi} = \sum_{m= 0}^{Q} \sum_{n = 0}^{N - Q}A_{mn} \ket{m_Q} \ket{n_{N - Q}}
\end{equation}
 
Such an expansion is well defined as the state has to be permutation symmetric in the first $Q$ block of spins as well in its complement. It is useful to note that, every state of the form of Eq.~(\ref{eq:permpsi1}) can be written as in Eq.~(\ref{eq:permpsi2}), but not vice-versa.
Thus the matrix elements $A_{mn}$ are correlated in a specific way, which is also seen as they will involve the $a_k$ which are only $N+1$ in number. However, the advantage of writing $\ket{\psi}$ in this way is that the $Q$-qubit reduced density matrix (in the standard $Q$-qubit permutation symmetric basis) is simply $\rho_Q=A A^{\dagger}$ as the sub-block Dicke states are also orthonormal.

Now the only remaining task is to determine the matrix elements of $A$.  Observe that $\ket{m_Q} \ket{n_{N - Q}}$ will contribute a sum of $N$-qubit computational basis states corresponding to the Dicke state $\ket{(m + n)_N}$. It is important to note that these need not be equal; rather, the computational basis vectors obtained from the former Dicke states (tensor product) is contained in the latter Dicke state. Appropriately handling the normalization factors involved with the Dicke states, we see that $A_{mn}/(c_Q(m) c_{N-Q}(n)) = a_{m+n}/c_N(m+n)$ or 
\begin{equation}
    \label{eq:Amnarray}
    A_{mn} = \frac{c_Q(m) c_{N - Q}(n)}{c_N(m + n)} a_{m + n}.
\end{equation}

As an example if $N = 4$, an arbitrary permutation symmetric pure state is given in the Dicke basis as 
\beq
|\psi\kt = a_0|0_4\kt +a_1|1_4\kt+a_2|2_4\kt+a_3 |3_4\kt+a_4|4_4\kt.
\eeq
If 
\beq
A' = \left( \begin{array}{cccc} a_0          & a_1/2        &a_1/2        &a_2/\sqrt{6}\\
                              a_1/2        & a_2/\sqrt{6} &a_2/\sqrt{6} &a_3/2\\
                              a_1/2        & a_2/\sqrt{6} &a_2/\sqrt{6} &a_3/2\\
                              a_2/\sqrt{6} & a_3/2        &a_3/2        &a_4
             \end{array} \right).
\eeq
then
\beq
|\psi\kt = \sum_{i,j=0}^3  A'_{ij}|ij\kt,
\eeq
where the binary representation of $i,j$ represent the qubit states. Thus in the standard basis the reduced density matrix of two qubits $\rho_2$ is $A'A'^{\dagger}$. However as is clear the rank of $A'$ and consequently $\rho_2$ is only 3, and this therefore calls for a reduction of the matrix to a typically full-rank $3\times 3$ matrix. This is essentially the coefficient matrix in the permutation symmetric basis of $\{ |00\kt$, $(|01\kt+|10\kt)/\sqrt{2}$, $|11\kt \}$ and is given by 
\beq
A = \left( \begin{array}{ccc} a_0           & \sqrt{2} a_1/2 & a_2/\sqrt{6}\\
                            \sqrt{2}a_1/2 & 2 a_2/\sqrt{6} & \sqrt{2}a_3/2\\
                            a_2/\sqrt{6}  & \sqrt{2} a_3/2 & a_4
           \end{array} \right).
\eeq

In general a $Q$ qubit reduced density matrix is derived from a $2^Q \times 2^{N-Q}$ dimensional coefficient matrix $A'_{ij}=a_{w(i)+w(j)}/c_N(w(i)+w(j))$, with $0 \le i \le 2^Q-1$ and $0 \le j \le 2^{N-Q}-1$. However this is largely rank-deficient and it suffices to consider the typically full rank $(Q+1) \times (N-Q+1)$ dimensional array in Eq.~(\ref{eq:Amnarray}).

With this, the $Q$-qubit reduced density matrix, as mentioned before, is $\rho_Q = A A^{\dagger}$. Note that this matrix is $(Q + 1) \times (Q + 1)$ dimensional as we are expanding it in terms of the $Q$-qubit permutation symmetric standard basis. The problem of entanglement of $Q$ qubits is reduced to a linear problem in the number of qubits rather than the usual exponential one, a consequence of the permutation symmetry. Thus a two qubit reduced density matrix is a $3 \times 3$ one and hence at most rank-3, and the maximum entropy of such a state is $\log_23$, rather than 2 ebits of a general state of 2 qubits. In general a reduced density matrix of $Q$ qubits can have at most the entropy of $\log_2(Q+1)$ ebits as opposed to $Q+1$ ebits. Permutation symmetric states are far less entangled than generic ones.
Note that not all permutation symmetric random mixed states are of the form considered above. For example the two qubit state
\beq
a\,  |00\kt \br 00| + b \, (|01\kt \br 01| +|10\kt \br 10|)+c\, |11\kt \br 11|
\eeq 
is manifestly permutation symmetric but is not a mixture of Dicke states. Thus the reduced density matrices we study are a subset of permutation symmetric states, specifically those that can be derived as reduced states of larger permutation symmetric pure states. For a more formal treatment of the reduced density matrix of a permutation symmetric state see Appendix \ref{apdx:permsymmembedding}.

\section{Random permutation symmetric states and a new ensemble of random matrices}

In order to the study the evolution of permutation symmetric states under chaotic but permutation symmetric evolutions such as in the kicked top, we have to study the properties of random permutation symmetric states. For our purposes it suffices to define an ensemble of permutation symmetric pure states as random permutation symmetric states if the coefficients $a_m$ in Eq.~(\ref{eq:permpsi1}) are drawn from the uniform (Haar) measure on the $N + 1$ dimensional space. In other words their joint probability distribution is given by 
\beq
\label{eq:jpdfHaar}
P(\{a_m\})=\df{N!}{\pi^{N + 1}} \delta\left(1 - \sum_{m=0}^N |a_m|^2 \right).
\eeq

While properties of random states including participation ratio, Shannon entropy \cite{Haake} and extreme value statistics \cite{lakshminarayan2008} have been studied previously, we are interested in the properties of subsystems and hence in reduced density matrices as in Eq.~(\ref{eq:Amnarray}). More precisely, we are interested in the properties of $A_Q A_Q^{\dg}$, where $A_Q$ is constructed from the $N + 1$ complex random numbers $a_m$ as
\beq
\label{eq:PermSymmRandDenMats}
(A_Q)_{mn}=\sqrt{\df{\binom{Q}{m} \binom{N-Q}{n}}{\binom{N}{m+n}}}\; a_{m+n}, \;\; \rho_Q^{PS}= A_Q A_Q^{\dagger},
\eeq
with $0\le m \le Q, \, 0\le n \le N-Q$.
The normalization of $a_m$ guarantees that $\tr(A_Q A_Q^{\dg}) = 1$, as required for density matrices. 
The generally rectangular matrix $A_Q$ has strongly correlated elements as there are only $N$ independent (complex) random numbers while there are $(Q+1)(N-Q+1)$ entries. Thus $\rho^{PS}_{Q}$ represents a new ensemble of positive random matrices that is of relevance to the study of random or typical permutation symmetric states, modeling their reduced density matrices of subsystems having $Q$ of the total $N$ qubits. 

This ensemble has very different properties than the well-studied normalized Wishart or trace-constrained ensemble. 
If $G$ is a $ N_1 \times N_2$ dimensional matrix with complex entries whose real and imaginary parts are independently normally (zero centered) distributed, the ensemble
\beq
\label{eq:Wishart}
\rho^W_{N_1}=\frac{G G^{\dagger}}{\tr(G G^{\dagger})}
\eeq
is a normalized Wishart ensemble \cite{wishart1928generalised}. Its eigenvalues are distributed according to the Marchenko-Pastur law \cite{marvcenko1967distribution}. This is the ensemble of reduced density matrices of $N_1$ dimensional subsystems of pure states that are uniformly sampled from the Hilbert space of dimension $N_1 N_2$.

It is useful to compare $\rho^{PS}_Q$ with two Wishart ensembles that naturally present themselves, one that is relevant to a $Q$ qubit subsystem of random $N$ qubit states, so that $N_1=2^Q$ and $N_2=2^{N-Q}$, which we will denote as $W, 2^Q$ as the subsystem dimensionality is $2^Q$. The other is that of a subsystem of dimension $Q + 1$ in a randomly chosen bipartite pure state of dimension $(Q + 1)\times (N - Q + 1)$, that is $N_1=Q+1$ and $N_2=N-Q+1$, which we will denote as $W, Q + 1$.

For $Q = N/2$, without the combinatorial factors in Eq.~(\ref{eq:PermSymmRandDenMats}) the matrix $A_Q$ is a square Hankel matrix, and random Hankel matrices have been considered before in the literature \cite{bryc2006}, and the existence of the limiting spectral distribution has been established, although its explicit form remains unknown. It seems natural in this context to study the ensemble with the combinatorial factors that also ensure normalization of the corresponding density matrix. As we are interested in the spectrum $\{ \lambda_i, \; 1\le i \le Q+1\}$ of $\rho^{PS}_Q=A_Q A_Q^{\dagger}$, we will be interested in the square of the singular values of $A_Q$.

Figure \ref{fig:eigval_density_single} shows the distribution of the $\{ \lambda_i \}$  of the $12$-qubit random permutation symmetric ensemble. Shown are the eigenvalues for $Q = 1,\,2,\,5,$ and $6$, scaled by a factor $Q + 1$, and one sees $Q + 1$ peaks that are merging into a smooth distribution. Except for the case $Q = N/2$, the density vanishes at the origin. When $Q = N/2$ the density looks very close to the Marchenko-Pastur law, which after scaling by the factor $Q+1$ (i.e. $x = \lambda (Q + 1)$) is 
\beq
\label{eq:MP}
P_{MP}(x) = \frac{1}{2 \pi} \sqrt{\frac{4 - x}{x}}. 
\eeq

This is seen in  Fig.~ \ref{fig:eigval_density_multiple}  which has three different values of $N = 50$, $100$ and $200$ qubits. The eigenvalues have been obtained from $N/2$-qubit subsystems of $2500$ randomly generated $N$-qubit permutation symmetric states, and have been scaled by a factor of $N/2 + 1$. It can be seen that the scaling causes all the distributions to collapse onto each other, indicating the existence of a limiting distribution. This limiting distribution, although similar to Marchenko-Pastur, does differ from it slightly in the bulk. In addition, numerical results indicate that the distribution does not diverge at the origin and has an exponentially decaying tail unlike the Marchenko-Pastur distribution that diverges at the origin and has a finite support in $[0, 4]$.

\begin{figure}[!htbp]
    \includegraphics[width=0.48\textwidth]{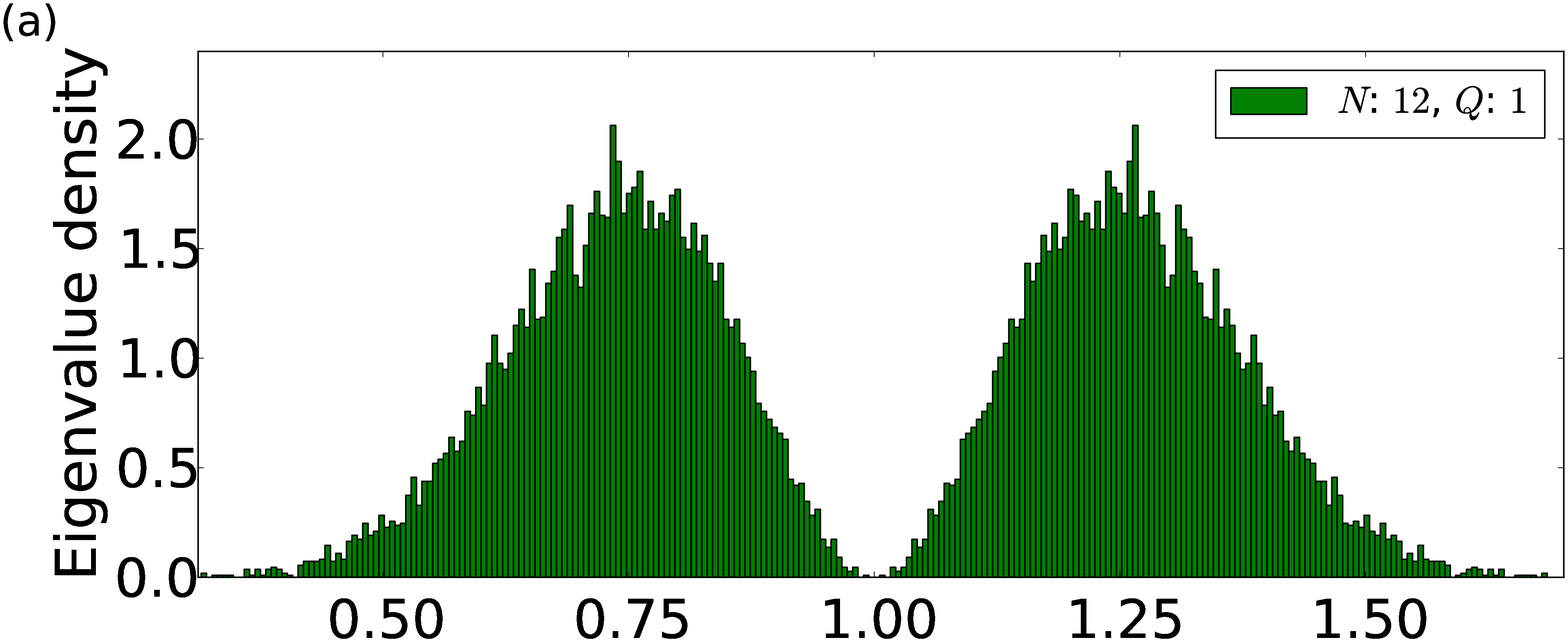}
    \\
    \includegraphics[width=0.48\textwidth]{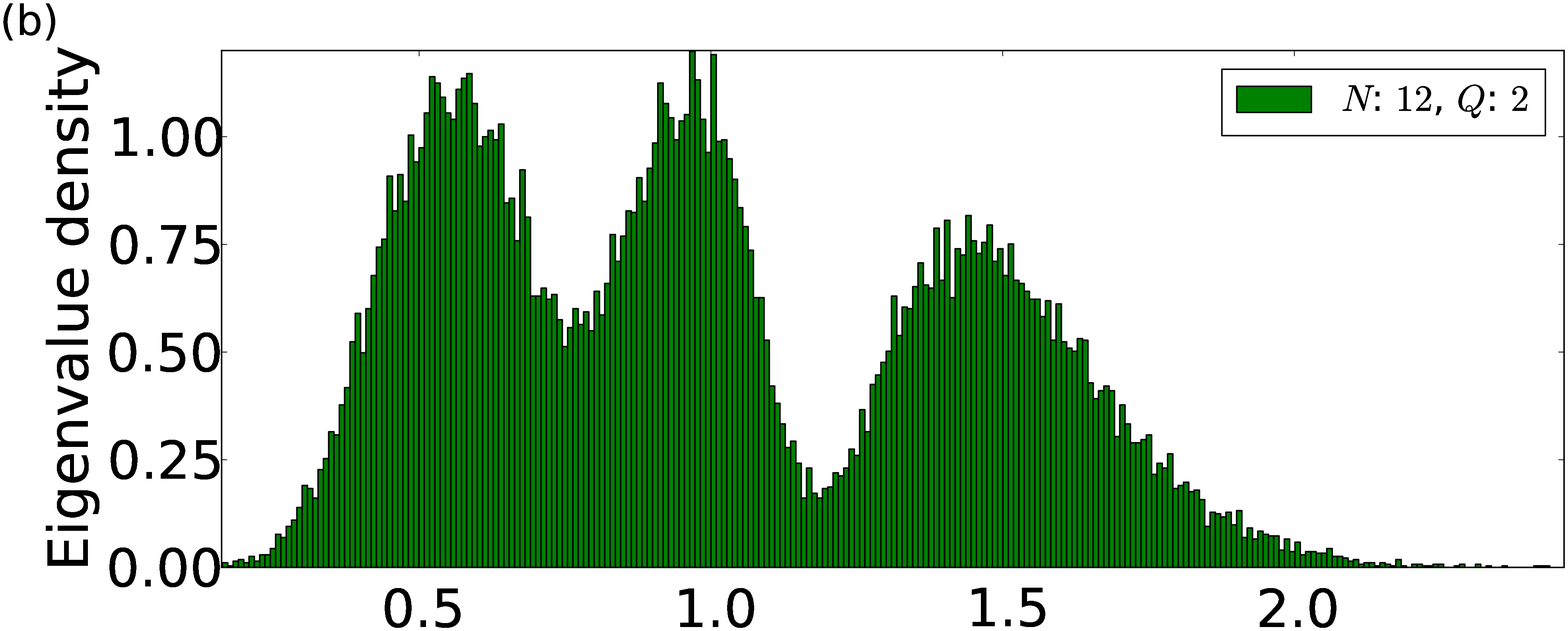}
    \\
    \includegraphics[width=0.48\textwidth]{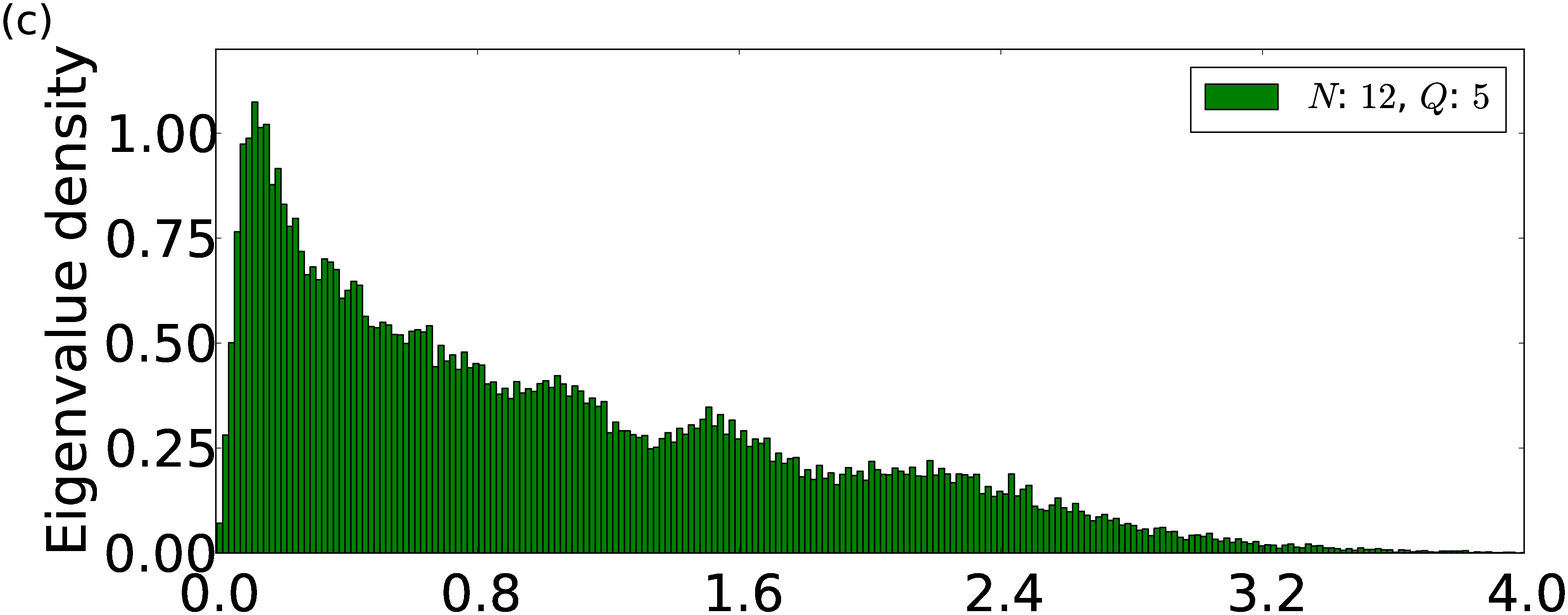}
    \\
    \includegraphics[width=0.48\textwidth]{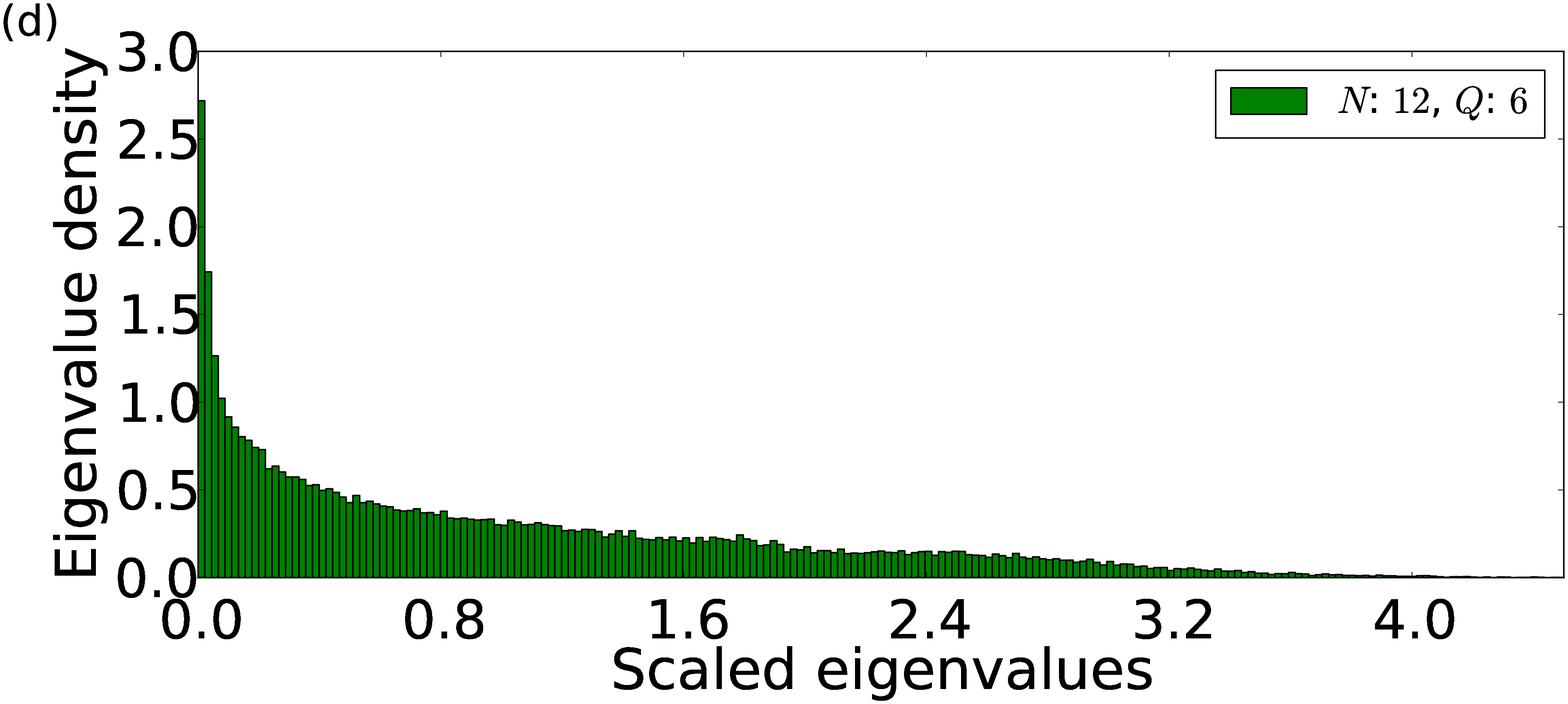}
    
    \caption{(Color online) The histogram of eigenvalues of the $Q$-qubit reduced density matrix of an $N$-qubit permutation symmetric system. The eigenvalues have been obtained from subsystems of $10000$ randomly generated $N$-qubit permutation symmetric states. The histogram has been scaled such that the area under the curve is one, so that the y-axis is representative of eigenvalue density.}
    \label{fig:eigval_density_single}
\end{figure}

\begin{figure}[!htbp]
    \includegraphics[width=0.49\textwidth]{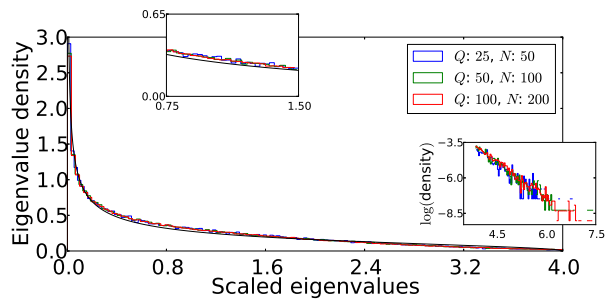}
    \caption{(Color online) The histogram of eigenvalues of the block subsystem for a permutation symmetric system containing (a) $N = 50$, (b) $N = 100$ and (c) $N = 200$ qubits. The eigenvalues have been obtained from $N/2$-qubit subsystems of $2500$ randomly generated $N$-qubit permutation symmetric states, and have been scaled by a factor of $N/2 + 1$. The histogram has been scaled such that the area under the curve is one, so that the y-axis is representative of eigenvalue density. The Marchenko-Pastur distribution (black line) has been plotted for comparison. The rightmost inset shows the logarithm of the density in the tail region, indicating an exponentially decaying tail.}
    \label{fig:eigval_density_multiple}
\end{figure}

\subsection{Average purity and linear entropy}

The superscript label in $\rho^{PS}_Q$ is now dropped for brevity. The easiest nontrivial quantity that maybe found for the ensemble is its average purity $\br \tr \rho_Q^2 \kt_{PS}$ and hence the average linear entropy $\br S_Q^l \kt_{PS}=1-\br \tr \rho_Q^2 \kt_{PS}$. Here $\rho_Q=A_Q A_Q^{\dg}$, where $A_Q$ is from the random ensemble as described in Eq.~(\ref{eq:PermSymmRandDenMats}) and Eq.~(\ref{eq:jpdfHaar}), while $\br \, \cdot \, \kt_{PS}$ indicates averaging with respect to this ``permutation symmetric" ensemble.

\begin{widetext}
\beq
\br \tr \rho_Q^2 \kt _{PS} = \left \br \sum_{i=1}^{Q+1} \lambda_i^2 \right \kt _{PS} =\sum_{k,j=0}^{Q} \binom{Q}{k}\binom{Q}{j} \sum_{m,n=0}^{N-Q} \binom{N-Q}{m}\binom{N-Q}{n} \dfrac{\br a_{k+m}a_{j+n}a^*_{k+n}a^*_{j+m}\kt }{\sqrt{\binom{N}{k+m}\binom{N}{j+n}\binom{N}{k+n}\binom{N}{j+m}}}.
\eeq
\end{widetext}

If the $a_i$ are drawn from the distribution Eq.~(\ref{eq:jpdfHaar}) it is easy to see that $\br |a_{i}|^2 |a_{j}|^2\kt =1/[(N+1)(N+2)]$ if $i \neq j$ and is $2/[(N+1)(N+2)]$ if $i=j$. These are the only non-zero average terms that are needed to show that
\beq
\br a_{k+m}a_{j+n}a^*_{k+n}a^*_{j+m}\kt=\df{1}{(N+1)(N+2)} \left( \delta_{mn}+\delta_{kj}\right).
\eeq
Using this, the average purity becomes,
\beq
\begin{split}
\br \tr \rho_Q^2 \kt _{PS} &= \df{1}{(N+1)(N+2)} \left[ 
\sum_{k,j,m} \dfrac{\binom{Q}{k}\binom{Q}{j} \binom{N-Q}{m}^2}{\binom{N}{k+m}\binom{N}{j+m}} \right. \\  & + \left. \sum_{k,m,n} \dfrac{\binom{Q}{k}^2\binom{N-Q}{m} \binom{N-Q}{n}}{\binom{N}{k+m}\binom{N}{k+n}}\right].
\end{split}
\eeq
Apparently intimidating, the combinatorial sums are in fact benign, the first is
\beq
\label{eq:AvgPurityCombIdentity}
\sum_{k,j,m} \dfrac{\binom{Q}{k}\binom{Q}{j} \binom{N-Q}{m}^2}{\binom{N}{k+m}\binom{N}{j+m}}=\df{(N+1)^2}{(N-Q+1)},
\eeq
and the second follows on replacing $Q$ with $N-Q$. 
Hence finally 
\beq
\begin{split}
\label{eq:AvgPurityLinEnt}
\br \tr  \rho_Q^2 \kt _{PS} = \df{N+1}{(Q+1)(N-Q+1)}, \\ \br S^l_Q \kt _{PS}=\df{Q(N-Q)}{(Q+1)(N-Q+1)}.
\end{split}
\eeq
To reiterate, the above is the average $Q$ qubit purity and linear entropy of random $N$ qubit permutation symmetric pure states. As expected, it is symmetric under the replacement $Q \rarrow N-Q$ .

Comparing with the Wishart ensembles, it is known \cite{Lubkin} that the average purity of a $M$ dimensional system in a $MN$ dimensional random pure state is $(M+N)/(1+MN)$. Using this we get 
\beq
\label{eq:avg:HS}
\begin{split}
 \br S_Q^l \kt _{W, 2^Q}=\df{(2^Q-1)(2^{N-Q}-1)}{2^N+1}, \\
 \br S_Q^l \kt _{W, Q+1}=\df{Q(N-Q)}{1+(Q+1)(N-Q+1)}.
\end{split}
\eeq
It is understandable that the second ensemble $W,Q+1$ is close in entropy to that of permutation symmetric states. That it is slightly smaller than that of permutation symmetric states is consistent with the behaviour of the density of eigenvalues. While the one for the $W,Q+1$ ensemble is the Marchenko-Pastur one that diverges at the origin and is sharply cut-off at 
$x=\lambda (Q+1) =4$, the PS ensemble seems to not diverge at the origin and extends to infinity with an exponential tail.

One quantity frequently used in previous studies is $S^l_1$, that is the linear entropy of a single qubit \cite{Wang2004, madhok2015signatures}. In this case:
\beq
\begin{split}
\br S_1^l \kt _{PS} =\f{1}{2}\left(1-\f{1}{N}\right), \\ \br S^l_1\kt _{W,2^Q} =\f{1}{2}\left(1-\f{3}{2^N+1}\right), \\ \br S^l_1\kt_{W,Q+1} =\f{1}{2}\left(1-\f{3}{2N+1}\right).
\end{split}
\eeq
Thus while both for permutation symmetric as well as asymmetric random states the average of the linear entropy single qubit density matrix approaches the thermalized value of a most mixed state ($=1/2$), the approach is algebraic in the number of qubits for the permutation symmetric states while it is exponential for the asymmetric states if it is 
assumed to be a $Q$ qubit subsystem.
 
For the case of $N = 2Q$ we have for large $N$:
\beq
\br S_{N/2}^l \kt _{PS} \approx 1-\f{4}{N},\, \br S_{N/2}^l \kt _{W,2^Q} \approx 1-\f{2}{2^{N/2}},
\eeq
again illustrating the algebraic rate of scaling to maximum entropy as compared to an exponential one for $Q$ qubits. The case of a $Q + 1$ dimensional subsystems of a random state ($W, Q+1$) has the same leading order deviation from $1$ as the PS states. At this level for large $N$ and $N = 2Q$ the linear entropies are indistinguishable. The differences however can be seen in quantities like the von Neumann entropy to which we turn to now.

\subsection{Average von Neumann entropy}

While the linear entropy is easy to calculate and find the average of, it is not additive, and in fact the von Neumann entropy has the special place as the measure of entanglement. In the absence of a 
j.p.d.f of the eigenvalues of the reduced density matrix or other features of the ensemble, we resort to numerical methods for this quantity. One thing that can be definitely said is that $S^{vN}_Q\leq \log(Q+1)$, as the rank of the reduced density matrix is at most $Q+1$ (The logarithms used in the paper are in base 2.)
In sharp contrast, for random non-symmetric states, the von Neumann entropy is bounded by $Q$ and the average differs from this by utmost an order 1 number.

Extensive numerical calculations support the following simple formula
\beq
\label{eq:avgvNPS}
\begin{split}
\br S^{vN}_Q \kt_{PS} = - \left \br \sum_{i=1}^{Q+1} \lambda_i \log \lambda_i \right \kt_{PS} \\ \approx \log(Q+1)-\alpha \f{Q+1}{N-Q+1}
\end{split}
\eeq
where $1 \leq Q \leq N/2$,  $1 \ll N$ and $\alpha \approx 2/3$ is a constant. The leading order correction to this seems to be $1/(N + 1)$ for the case of $Q = N/2$, and Fig.~ \ref{fig:vncomparion} gives the numerical evidence supporting this claim.  
Numerically computed average von Neumann entropy and the approximate formula given in Eq.~(\ref{eq:avgvNPS}) for a $Q$-qubit subsystem of an $N$-qubit permutation symmetric system match quite well.

It is of interest to see how close is the entropy for permutation symmetric states to that of random states in a $(Q+1)\times(N-Q+1)$ dimensional space, namely the $W,Q+1$ ensemble. Using the large $N$ approximation again with $N = 2Q$ yields, using a formula first conjectured by Page \cite{page1993average} and proved thereafter \cite{sen1996average, foong1994proof, sanchez1995simple}
\beq
\label{eqn:avgvNHSQp1}
\br S^{vN}_Q \kt_{W,Q+1} \approx \log(Q+1)-\f{1}{2\ln 2} \approx \log(Q+1)-0.721.
\eeq
Comparing with Eq.~(\ref{eq:avgvNPS}) implies that $\br S^{vN}_Q \kt_{W,Q+1}$ is marginally smaller than that the random permutation symmetric case, whose average is $\approx 0.66$ smaller that the maximum entanglement. This is seen in Fig.~ \ref{fig:vncomparion} as the systematic lower entropy for the $W, Q+1$ case, which is consistent with the lower linear entropy. 

In contrast for random states of $N$ qubits without any symmetry, the $W, 2^Q$ case,
\beq
\label{eq:avgvNHS}
\br S^{vN}_Q \kt_{W,2^Q} \approx Q-\f{1}{\ln 2} \f{2^{Q}}{2^{N-Q+1}}.
\eeq
Thus permutation symmetric states have marginal entanglement that scales with system size in a logarithmic manner. This has been known for some time using different approaches and in Dicke states \cite{popkov2005logarithmic}, but in the context of random states this has not been studied before. It is interesting to note that the asymmetric random states follow a ``volume" law while permutation symmetric states are marginal and increase logarithmically with the subsystem size, similar to critical spin chains and integrable CFTs \cite{Vidal2003, latorre2005entanglement, popkov2005entangling, popkov2005logarithmic, barthel2006entanglement, vidal2007entanglement}, and subsequently this affects the behaviour of mutual information \cite{wilms2012finite}. This has implications for the sign of the tripartite mutual information as we now discuss.

\begin{figure}
    \includegraphics[width=0.5\textwidth]{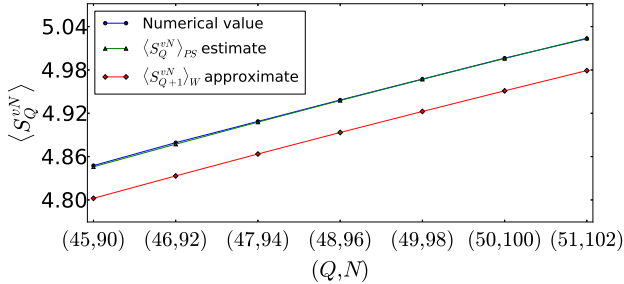}
    \caption{(Color online) A comparison between numerically computed average von Neumann entropy and the approximate formula given in Eq.~(\ref{eq:avgvNPS}) (along with $1/(N + 1)$ correction) for a $Q$-qubit subsystem of an $N$-qubit permutation symmetric system. It can be seen that the approximate formula (green line) matches reasonably well with the numerically computed average (blue line), in contrast with the approximation (Eq. \ref{eqn:avgvNHSQp1}) obtained from Page formula (red line). The numerical values have been obtained by averaging over $10000$ random permutation symmetric states.}
    \label{fig:vncomparion}

\end{figure}

\subsection{Tripartite mutual information}
With the above, it is possible to estimate the behaviour of TMI in random permutation symmetric states as well as the Wishart ensemble
of $2^N$ qubits. Using the definition of the TMI in Eq.~(\ref{eqn:tmi_from_mi}) it follows that 
\beq
\label{eq:avTMI-PS}
\begin{split}
&\left \br  I_3  (Q,Q,Q)\right \kt_{PS} \approx 3 \log(Q+1)-3\log(2Q+1)\\&+\log(3Q+1)=\log\left[ \frac{(3Q+1)(Q+1)^3}{(2Q+1)^3} \right]>0,
\end{split}
\eeq
where we have simply used $S_A \approx \log(Q+1)$ and ignored corrections. This is sufficient to show that the TMI average value for random permutation symmetric states is positive. The above will be a good approximation only for large $Q$ and $N$.

For the case of the ensemble of all random states on the full $2^N$ dimensional space, using the leading order term in Eq.~(\ref{eq:avgvNHS}) gives for the TMI $3Q - 6Q + 3Q = 0$ and therefore it is necessary to use the deviation from maximally entangled states in that formula and this yields
\beq
\label{eq:avTMI-HS}
\left \br  I_3  (Q,Q,Q)\right \kt_{W, 2^Q} \approx - \frac{2^{2Q-N-1}}{\ln 2}\left(2^{4Q}-3 \, 2^{2Q}+3\right).
\eeq 
This is negative as $x^2 - 3x + 3 \ge 3/4$. This is again valid for large $Q$ and $N$. Therefore typical states are not only entangled but they also have a negative TMI implying that information is distributed in multipartite ways. In contrast for permutation symmetric states, the entanglement is small, being only logarithmic in the number of qubits and hence the TMI is typically positive and the information is stored more in bipartite partitions and is not spread out. We can observe this numerically in Fig.~\ref{fig:tmi_HS_PS_randomstates}, where we compute TMI for subsystems of random states with and without permutation symmetry.

Often it is easier for calculations to use the linear entropy and hence we define, even if a little dubious, a linear entropic TMI by using for $H(\cdot)$ in Eq.~(\ref{eqn:tmi_from_mi}) the linear entropies. Hence using Eq.~(\ref{eq:AvgPurityLinEnt}) we can get exact average  linear entropic TMI, for example
\beq
\label{eq:avglinTMI-PS1}
\left \br  I_3^l  (1,1,1)\right \kt_{PS}=\frac{(N-3)(N^2-N+4)}{4(N-2)N(N-1)} \approx \frac{1}{4}-\frac{1}{4N},
\eeq
and 
\beq
\label{eq:avglinTMI-PSm}
\left \br  I_3^l  (m,m,m)\right \kt_{PS} \approx \frac{6 m^3}{(m+1)(2m+1)(3m+1)}.
\eeq
These demonstrate again that the TMI of permutation symmetric states, now with the linear entropy, is also positive on the average.
However this is not quite useful as it is easy to check that it is also positive for random states that are not symmetric, that is when the reduced density matrix is from the Wishart ensemble. Additivity of the entropy is an important property and as the linear entropy is not additive, it does not distinguish the ensembles. We have verified for example that the R\'enyi entropy, which is additive at any order, does in fact have the capability.

\begin{figure}[!htpb]
    \includegraphics[width=0.5\textwidth]{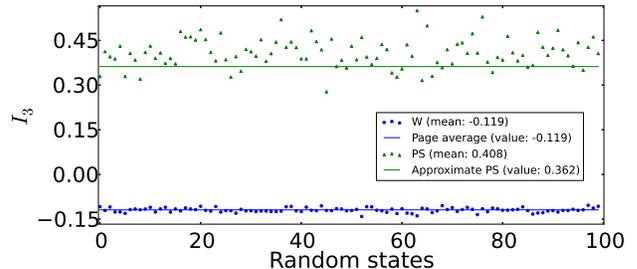}
    \caption{(Color online) TMI ($vN$) between $1$ qubit, $2$ qubit, $2$ qubit subsystem of a $12$ qubit system. The blue points show the TMI for different realizations of asymmetric random states, and it can be seen that the average obtained from Page formula (blue line) matches with the numerical average. The green triangles, on the other hand, show TMI for different realizations of random PS states, and the average TMI obtained from Eq.~(\ref{eq:avgvNPS}) is moderately close to the numerical average.}
    \label{fig:tmi_HS_PS_randomstates}
\end{figure}

In the preceding discussion, we noted that a positive value for average TMI is observed for the PS ensemble. This statement can in fact be extended to cover individual (random) states by the use of L\'evy's lemma. In Appendix \ref{apdx:permsymmlevy}, we show how to apply L\'evy's lemma to permutation symmetric systems so that the average of the quantities we are interested in is taken only over random PS states and not asymmetric states. Using this, we can easily see that for a large number of qubits in the total system, the TMI (based on either von Neumann or linear entropy) for most random PS states is nearly equal to the average TMI. Thus if the average TMI is positive, we can expect most states to have a positive TMI as well, given that we are working in large enough dimensions.

Now, from Eq.~(\ref{eqn:tmi_from_mi}), we can see that a positive TMI implies that $I_2(A: B) + I_2(A: C) > I_2(A: BC)$, which would mean that the (bipartite) mutual information is not monogamous. Monogamy of an information measure refers to the ability of a party to share correlations, as defined by this measure, with other parties. Our results, therefore, suggest connections between mutual information in permutation symmetric systems and monogamy. More precisely, if we start with a pure permutation symmetric state $\ket{\psi_{ABCD}}$ corresponding to some partition $ABCD$, the mutual information that $A$ shares individually with the subsystems $B$ and $C$ is greater than what it shares with the joint system $BC$. Thus, the subsystems can share correlations with other parties. Since the mutual information is a measure of total correlations, the correlations that can be shared may either be classical or quantum.
 
\section{Mutual informations and entanglement in the kicked top}

We now turn to the details of the kicked top propagator, evolve states using this and explore how in the chaotic regime the 
random state values of the previous section holds. In the process we study the short time evolution and growth of 
several interesting measures, the mutual information, the tripartite mutual information, the entanglement and the 
out-of-time-ordered correlator.

\subsection{The kicked top as a many-body spin model}
The kicked top consists of a single large rotor whose quantum evolution over one time period $\tau$ of the kicking is given by the propagator \cite{Haake}
\begin{equation}
    U = \exp\left(-i \frac{k}{2j} J_z^2\right) \exp\left(-i p J_y\right) \label{report1eqn:kicedtopunitary}
\end{equation}
Since $[\bm{J}^2, J_i] = 0$, we are restricted to a $(2j + 1)$-dimensional Hilbert space, and we can use the standard angular momentum basis $\{\ket{j, m} | -j \leq m \leq j\}$, which are the simultaneous eigenstates of $\bm{J}^2$ and $J_z$. Here $k$ and $p$ are parameters. The classical map \cite{Haake, kus1987symmetry, kus88, Zyczkowski90} is from the surface of the sphere $(J_x^2 + J_y^2 + J_z^2)/j^2 = 1$ into itself. It can have regular as well as chaotic dynamics, and this behaviour is controlled by $k$. We use $p = \pi/2$ below, for which $k = 0$ is integrable, being simply a rotation about the $y$ - axis, while around $k = 3$ the phase space is a mixed one, with a measure of chaotic and regular trajectories, while for $k = 6$ is almost fully chaotic.

The mapping between the kicked top model and the dynamics of qubits allows us to study the kicked top as a many-body system \cite{Wang2004, mgtg15, Neill16, Madhok2018_corr}. Since $\bm{J^2}$ is conserved, the state of the system can by mapped to a $2j$-qubit (or equivalently, spin-half) system, with the additional constraint that the qubits are always permutation symmetric. The existence of such a mapping can be understood from the fact that permutation symmetry effectively reduces the dimension of the $2j$-qubit system from $2^{2j}$ to $2j + 1$, and thus a linear isomorphism exists from the angular momentum $j$ system to the permutation symmetric $2j$-qubit system. The ``natural" basis states for the permutation symmetric $2j$-qubit system are therefore the Dicke states.
To be explicit, replacing $J^{x,y,z}$ with $\sum_{l=1}^{2j} \sigma^{x,y,z}/2$, the unitary or Floquet operator of the resultant $2j$ spin system is 
\beq
U = \exp\left(-i \frac{k}{8j}  \sum_{ l\ne l'=1}^{2j} \sigma^z_{l} \sigma^z_{l'}\right)
      \exp\left( -i \frac{\pi}{4} \sum_{l=1}^{2j}\sigma^y_l \right),
\eeq
where the $\sigma^{x,y,z}_l$ are the standard Pauli matrices, and an overall phase is neglected. 

Thus the spin model can be regarded as a kicked long-range transverse field Ising model which has identical interactions between all pairs of spins. While the nearest neighbor transverse field Ising model, kicked or otherwise, is integrable, (for example see \cite{prosen2000exact}) the long-range model can be non-integrable in the thermodynamic limit, which is also the classical limit $j \rarrow \infty$. Floquet models of many-body spin systems are being actively explored in the literature from many perspectives including many-body localization. Thus the kicked top in this many-body avatar presents an opportunity to study entanglement sharing and other typical questions that have been addressed thus far using many other spin models \cite{VedralRMPmanybody}, besides giving us a simple, if potentially chaotic, thermodynamic limit.

If the state vector of the system after $n$ kicks is $\ket{\psi_n}$, the state after $(n+1)$th kick is given by $\ket{\psi_{n+1}} = U \ket{\psi_n}$. In order to study the quantum-classical correspondence, one typically takes spin coherent states \cite{Glauber, Puri} as initial states. These states are parameterized by $\theta$ and $\phi$, and are minimum uncertainty states and therefore closest analogs to points on the classical phase space, aiding a quantum-classical comparison. Noting that $\ket{j, j}$ has minimum $\bm{J}^2$ uncertainty (the variance of $\bm{J}^2$ goes as $j$), the rest of the minimum uncertainty states are generated by rotating this state using $R(\theta, \phi) = \exp{\left(i \theta (J_x \sin(\phi) - J_y \cos(\phi))\right)}$, i.e., $\ket{\theta, \phi} = R(\theta, \phi) \ket{j, j}$. As a state of qubits this is simply the tensor product of qubit states whose Bloch sphere position is uniformly $(\theta, \phi)$: $|\theta,\phi\kt =\otimes^{2j} (\cos(\theta/2) |0 \kt +e^{i \phi} \sin(\theta/2) |1\kt)$.

Starting from such an initial spin coherent state, the time evolved state $U^n|\theta,\phi\kt$ resides in the permutationally invariant subspace of the complete Hilbert space of the $2j$ qubits. If the classical kicked top is in a completely chaotic regime, the quantum one generates pseudo-random states for sufficiently large time $n$ \cite{Haake}. Viewed as a multi-qubit state, it is in the permutation symmetric subspace and therefore we can expect in this case that the results of the previous section on random permutation symmetric states hold good.

There are two aspects of studying the behaviour of TMI and bipartite correlations like the mutual information and entanglement in permutation symmetric states like the ones generated when the dynamics is governed by the kicked top Hamiltonian. One is the investigation of their temporal behaviour that governs their growth and the other is exploring long-time averages or saturation values that would correspond to ensemble averages such as those calculated from the random permutation symmetric states.

\subsection{Growth of information measures and entanglement}

It is of particular interest to probe the dynamical behaviour of these correlations as a function of the initial coherent state localized in the regular or chaotic regions of the kicked top phase space. Unlike initial states localized in regular regions, the long-time evolution of the state under global chaos can be expected to thermalize or equilibrate, and this leads us to the study of ensemble averages. To the best of our knowledge, a study of dynamics of these quantum correlations across arbitrary subsystem decompositions is unexplored and earlier works have largely focused on subsystems consisting of a single or at most two qubits of a multi-qubit kicked top \cite{Wang2004, madhok2015signatures}. Moreover, the dynamical behaviour of quantities like entanglement and discord were compared rather qualitatively in the chaotic and regular regimes. We find that the Ehrenfest time plays a crucial role for saturation of many of the measures.
 
\begin{figure}[!htpb]
    \includegraphics[width=0.49\textwidth]{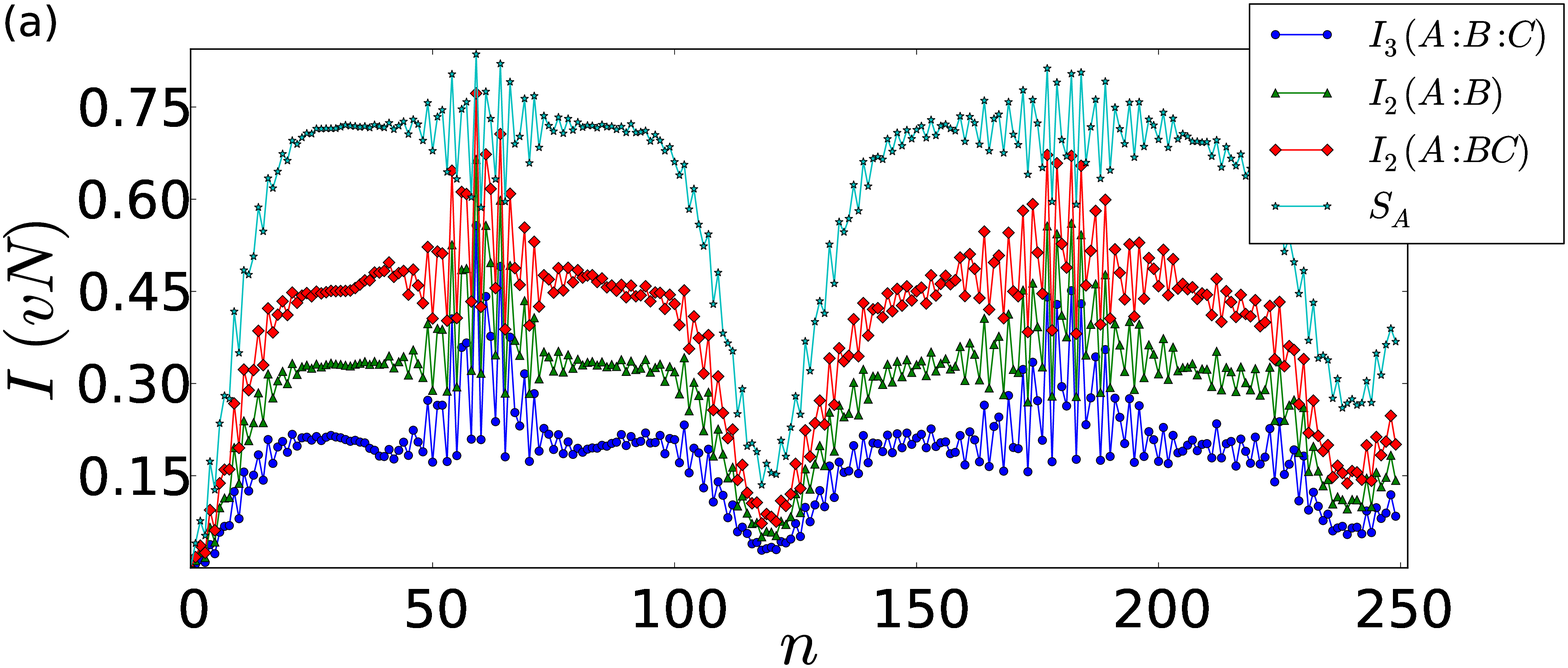}
    \\
    \includegraphics[width=0.49\textwidth]{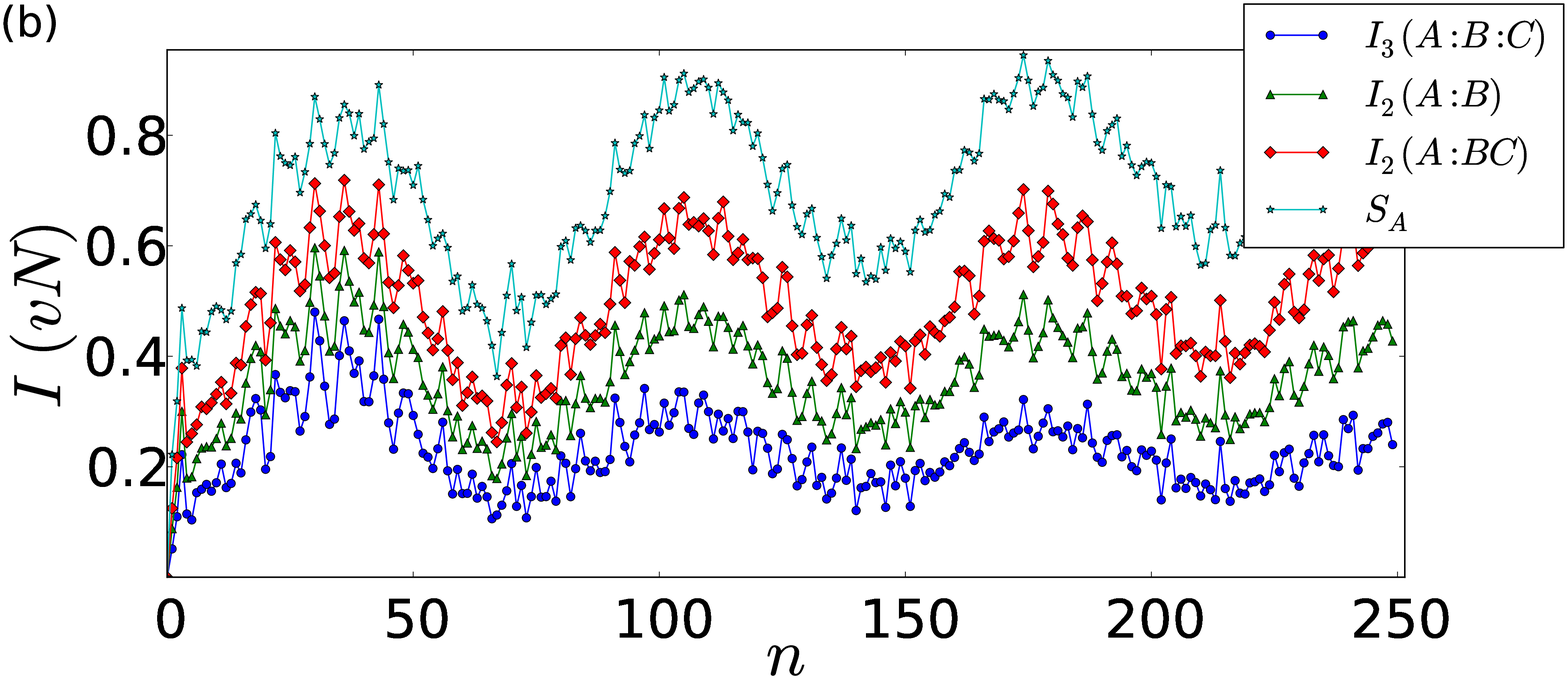}
    \\
    \includegraphics[width=0.49\textwidth]{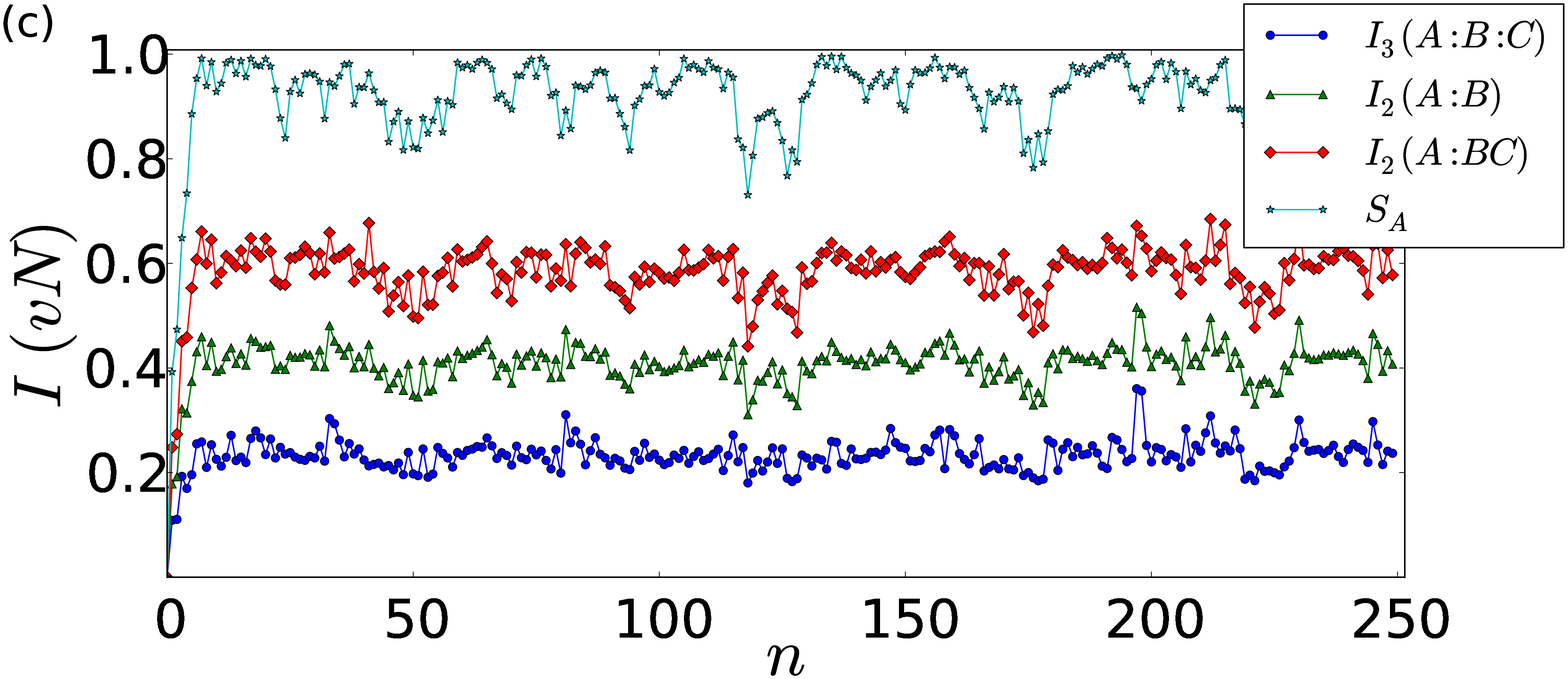}
    \\
    \includegraphics[width=0.49\textwidth]{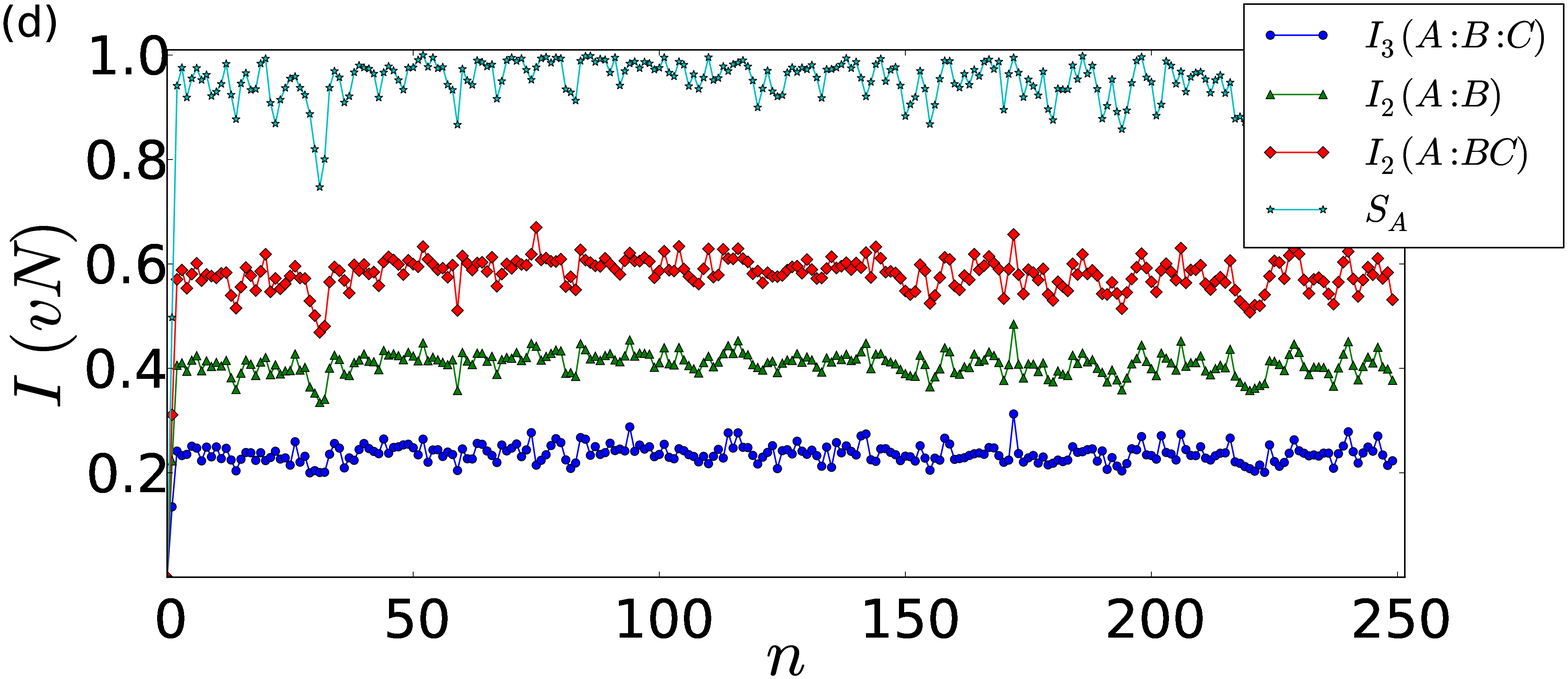}
    \caption{(Color online) TMI, MI and entanglement (computed using von Neumann entropy) between three $1$ qubit subsystems for $j = 10$ (i.e., a $20$ qubit system) in the kicked top. The TMI has been obtained in (a) the regular region with $k = 1$, $\phi = 0.63$, (b) a point in the regular island of the mixed phase space for $k = 3$, $\phi = 0.63$, (c) a point in the chaotic sea of the mixed phase space for $k = 3$, $\phi = 2$, and (d) the globally chaotic region with $k = 6$, $\phi = 0.63$ (in each case, $p = \pi/2$ and $\theta = 2.25$). It can be seen that the TMI reflects the regular nature of the underlying classical dynamics.}
    \label{fig:tmi_with_time}
\end{figure}

In accordance with this idea, we plot the behaviour of TMI with time in the kicked top in regular, mixed and chaotic regimes (see figure \ref{fig:tmi_with_time}). One would notice right away that the extent of regularity of the classical dynamics bears its signatures on the TMI. 
Regular dynamics leads to oscillatory TMI with large time variations and large values. In particular $I_3(1,1,1)$ is shown for an initial coherent state in a 20 qubit system. When $k=1$ the dynamics is regular and we see large positive values of the TMI indicating that information is shared more in a bipartite manner than collectively. When $k=3$ the dynamics is that of a mixed phase space and starting from a regular region leads to large oscillations of the TMI that seems to be damping over very long time scales. At the same value of $k$, starting from an initial state localized in a region in which the  classical limit is chaotic, leads to a rapid growth and saturation of the TMI with small oscillations around $0.234$. Comparing this to the corresponding average TMI ($vN$) for random permutation symmetric states, which is  $\approx 0.245$, corroborates our expectation that chaos in the kicked top thermalizes the initially coherent state, making it similar to a random PS state. 

A similar situation arises when $k = 6$ when the classical map is essentially globally chaotic, only the rise is even faster and the fluctuations smaller. It is this growth rate in a regime of global chaos that is examined further. Figure \ref{fig:tmi_initial_time_growth} compares the growth of TMI with that of mutual information and entanglement computed using von Neumann entropy as well as linear entropy for $j = 750$ (i.e, $1500$ qubit system), with each of the subsystems containing $100$ qubits each. To put the growth of TMI seen in this figure in perspective, we estimate the Ehrenfest time which is given as 
\begin{equation}
    t_{\text{Ehrf}} \sim \ln\left(h_{\text{eff}}^{-1}\right)/ \lambda_{cl} =\ln(2j+1)/\lambda_{cl},
\end{equation}
where the effective Planck's constant $h_{\text{eff}}$ goes as inverse dimension and $\lambda_{cl}$ is the classical Lyapunov exponent of the underlying system. 

\begin{figure}[!htpb]
    \includegraphics[width=0.47\textwidth]{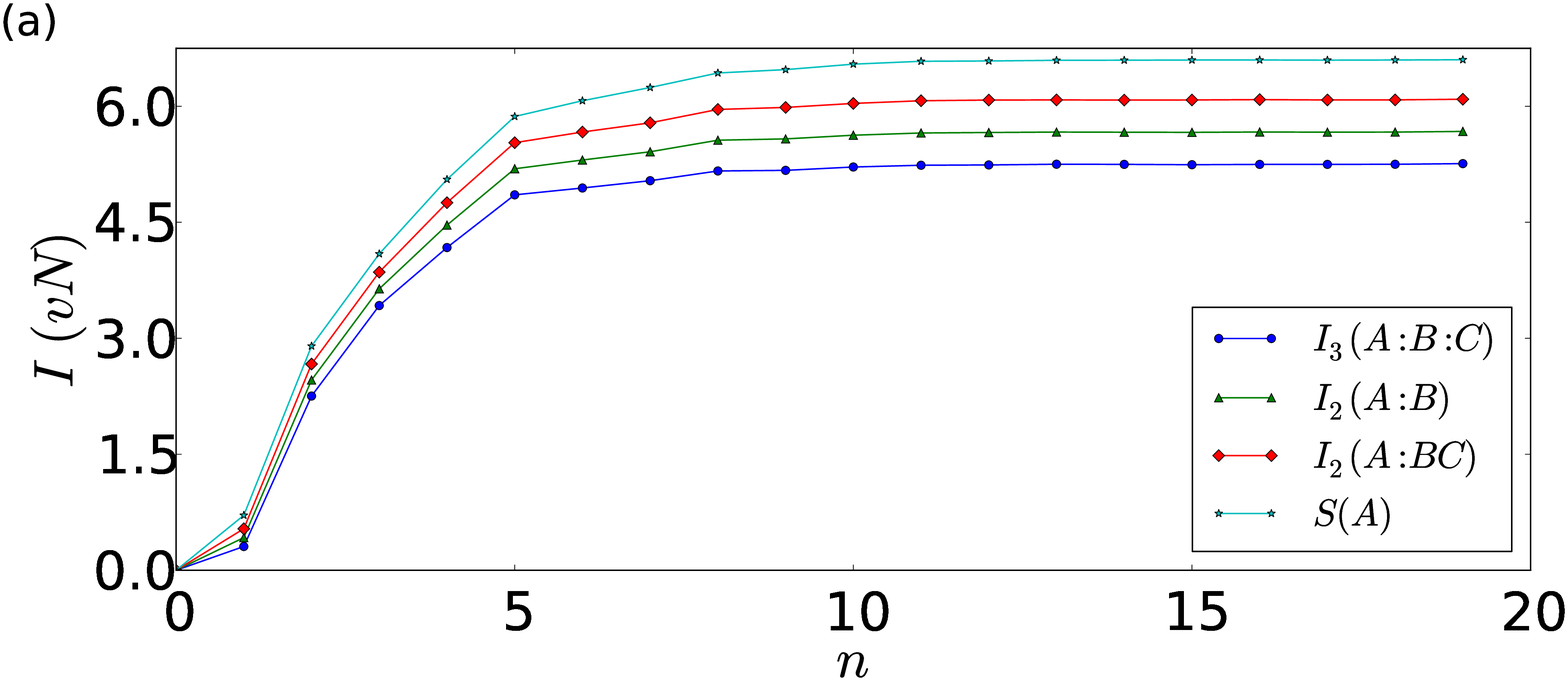}
    \\
    \includegraphics[width=0.47\textwidth]{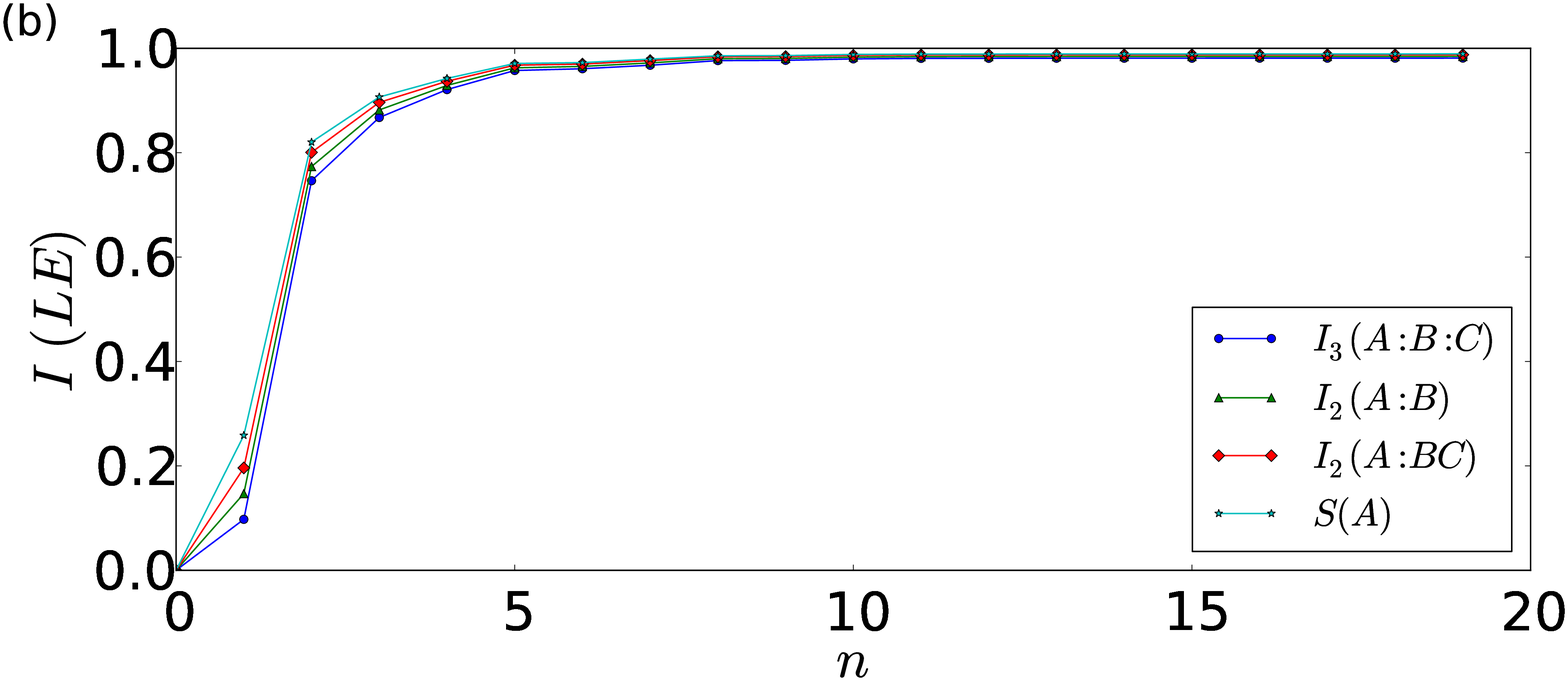}
    
    \caption{(Color online) TMI, MI and entanglement between three $100$ qubit subsystems for $j = 750$ (i.e., a $1500$ qubit system) in the kicked top using (a) von Neumann entropy and (b) linear entropy. It can be seen that the bipartite mutual information and entanglement qualitatively behave like the TMI.}
    \label{fig:tmi_initial_time_growth}
\end{figure}

The Lyapunov exponent for the classical kicked top map for $k = 6$ is approximately $0.97$, and for $j = 750$, we get an Ehrenfest time of approximately $7.54$ time steps. It is evident from figure \ref{fig:tmi_initial_time_growth} that the mutual information and TMI grow over the Ehrenfest time. 
We also see that behaviour of the bipartite mutual informations $I_2(A: B)$ and $I_2(A: BC)$, entanglement and the TMI are qualitatively alike. While we have shown this for subsystems containing equal number of qubits, this seems to be true even otherwise.

This implies, therefore, that the information between the subsystems $A$ and $B$ (or between $A$ and the joint subsystem $BC$) is similar to the information shared jointly by $A$, $B$ and $C$. Also, the entanglement grows in the same way. One possible reason for such a behaviour could be that the total state of the kicked top grows during the Ehrenfest time due to the effect of underlying classical dynamics, and this total state is what is mapped to a PS system of qubits. This growth, therefore, is reflected in the reduced density matrices obtained from this PS state and the measures computed from these. For such time varying scenarios in PS systems, one therefore needs to re-evaluate the significance of the TMI, as the information provided by TMI is also given by mutual information or even entanglement.

The TMI is defined with the entropies being the additive von Neumann entropy. It appears that the TMI (or mutual information) approaches the saturation value with an exponential rate. We note this by plotting the $\log$ behaviour of the TMI with the saturation value subtracted (not shown here). However, a more careful analysis is required for ascertaining the behaviour of TMI during the growth, and for quantifying the corresponding growth rates.

\subsection{Out-of-time-ordered correlators}

The scrambling of localized information in many-body systems is being studied and out-of-time-ordered correlators are used to characterize these \cite{Maldacena2016, shenker2014black, hayden2007black, hosur2016chaos}. It is convenient to define these via growth of commutators. As we wish to remain in the permutation symmetric subspace, the operators we choose are also permutation invariant and not local.
However, it is interesting to study this as it is a simple model of quantum chaos with a well-understood classical limit. For example this was studied in the kicked rotor in \cite{Rozenbaum17}, and there also has been a proposal to investigate scrambling experimentally using the kicked top \cite{swingle2016measuring}.

While postponing a detailed discussion to another work, we wish to contrast and compare this growth with that of the measures discussed above. Define, 
\beq
\label{eq:OTOC}
\begin{split}
&F(n)=-\mbox{Tr}[J_x(0), J_x(n)]^2=2\left( C_2(n)-C_4(n) \right), \\
&C_2(n)= \mbox{Tr}(J^2_x(n)J^2_x(0)),\\
&C_4(n)= \mbox{Tr}(J_x(n)J_x(0)J_x(n)J_x(0)),
\end{split}
\eeq
where $J_x(n)=U^{-n}J_x(0) U^n$ and $C_2(n)$ is a two-point correlator while $C_4(n)$ is the 4-point OTOC and its decrease with time essentially contributes to the growth of the commutator. The $C_2(n)$ behaviour is one of fast relaxation within a ``diffusion time" \cite{Maldacena2016}.
In figure \ref{fig:OTOC_initial_time_growth} we plot these three quantities (usual plot and semi-log plot) for $j = 750$, $k = 6$, $p = \pi/2$ and see exponential growth in $F(t)$. It is interesting that there are oscillations in $C_2(n)$ and $C_4(n)$ initially that compensate and lead to an exponential increase in the commutator.

The rate of growth of $F(t)$ seems to be slightly higher ($\approx 2.5$) than the estimate from simply replacing the commutator by classical Poisson brackets and estimating their growth rate at twice the Lyapunov exponent ($\approx 1.94$). On comparing with the other measures such as entanglement of a block of qubits, or TMI in Fig.~\ref{fig:tmi_initial_time_growth} we see that the Ehrenfest time is the log-time during which the commutator increases exponentially, and then saturates rather quickly.

Thus while there is an exponential growth of the commutators and hence ``scrambling", it seems to be at variance with the positive TMI observed above. Of course there is chaos, and mixing in the classical limit of the kicked top. The resolution maybe in the fact that there is scrambling in phase-space but not in qubit-space. Initial coherent state, localized in phase-space, spreads out exponentially and scrambles in the sense of becoming nonlocal in phase- space, however when viewed as scrambling within the qubits that comprise the effective spin model there is none. This dichotomy is also a reflection of the strong statistical properties, such as spectral fluctuations and eigenvector statistics,  of the kicked-top when viewed as a map on the sphere and a single large spin, compared to the marginal entanglement present in subsystems when viewed as collections of qubits.

There is however an effect or reflection of the scrambling in the qubit space. Not in the sign of the TMI but in terms of equilibration of the measures such as entanglement and tripartite mutual information to values that are given by random symmetric ensembles. More importantly, for sufficiently large subsystems, these grow over a log-time $ \sim \log(N)$. If there is no chaos we see a much slower growth of these quantities and while more detailed studies are needed, the time-scale seems to be as large as the Heisenberg time $\sim N$, see Fig.~(\ref{fig:tmi_with_time}).

\begin{figure}[!htpb]
    \includegraphics[width=0.48\textwidth]{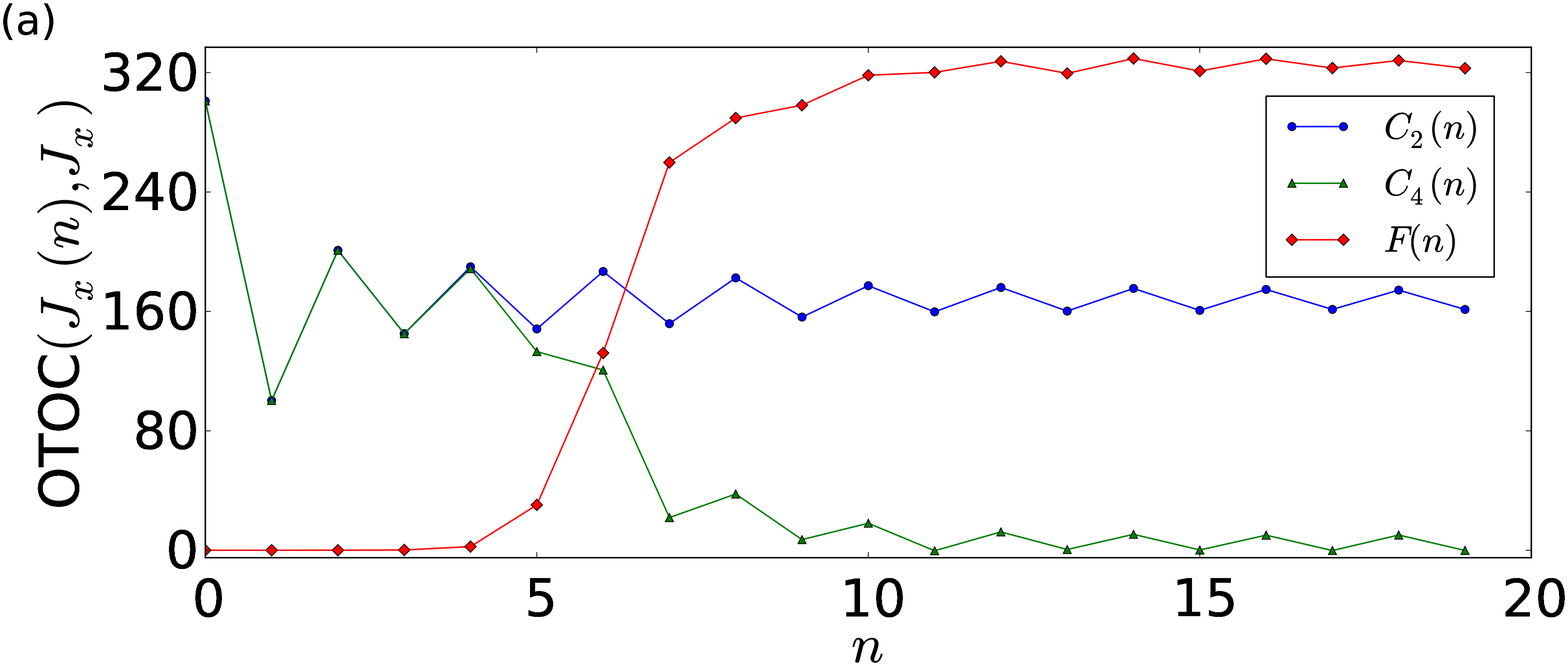}
    \\
    \includegraphics[width=0.48\textwidth]{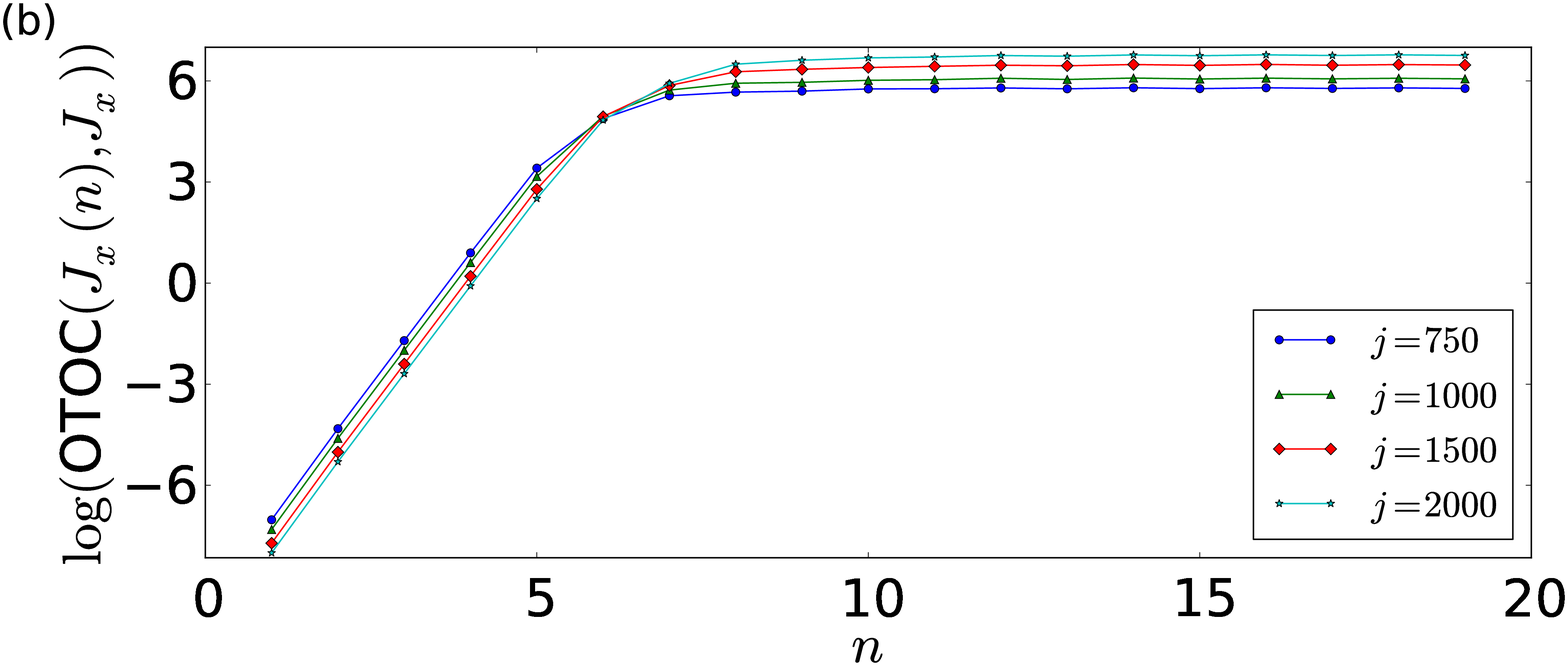}

    \caption{(Color online) (a) $F(n)$, $C_2(n)$ and $C_4(n)$ for $j = 750$, $k = 6$, $p = \pi/2$. The OTOC value has been divided by $j^4$ to normalize (or equivalently, the operators are being normalized by $j$). (b) Semi-log plot of OTOC for $j = 750, 1000, 1500, 2000$ and $k = 6$.}
    \label{fig:OTOC_initial_time_growth}
\end{figure}

\subsection{Saturation and long-time averages}
In addition to the dynamical behaviour, we also study the time-averaged TMI in the kicked top. We have already noted that the signature of classical dynamics is present in the time-varying TMI. Figure \ref{fig:timeaveraged_tmi} shows that such a behaviour is also present when we take time averages of TMI. In particular, comparing this figure with figure \ref{fig:classicalphaseportrait}, we can see that classical structures leave their mark on the quantum system, with regular islands being visible for integrable regimes, and low-period periodic orbits being present for the chaotic regimes. The latter is also a result of eigenfunction scarring \cite{heller1984bound}.

\begin{figure}[!htpb]
    \includegraphics[width=0.49\textwidth]{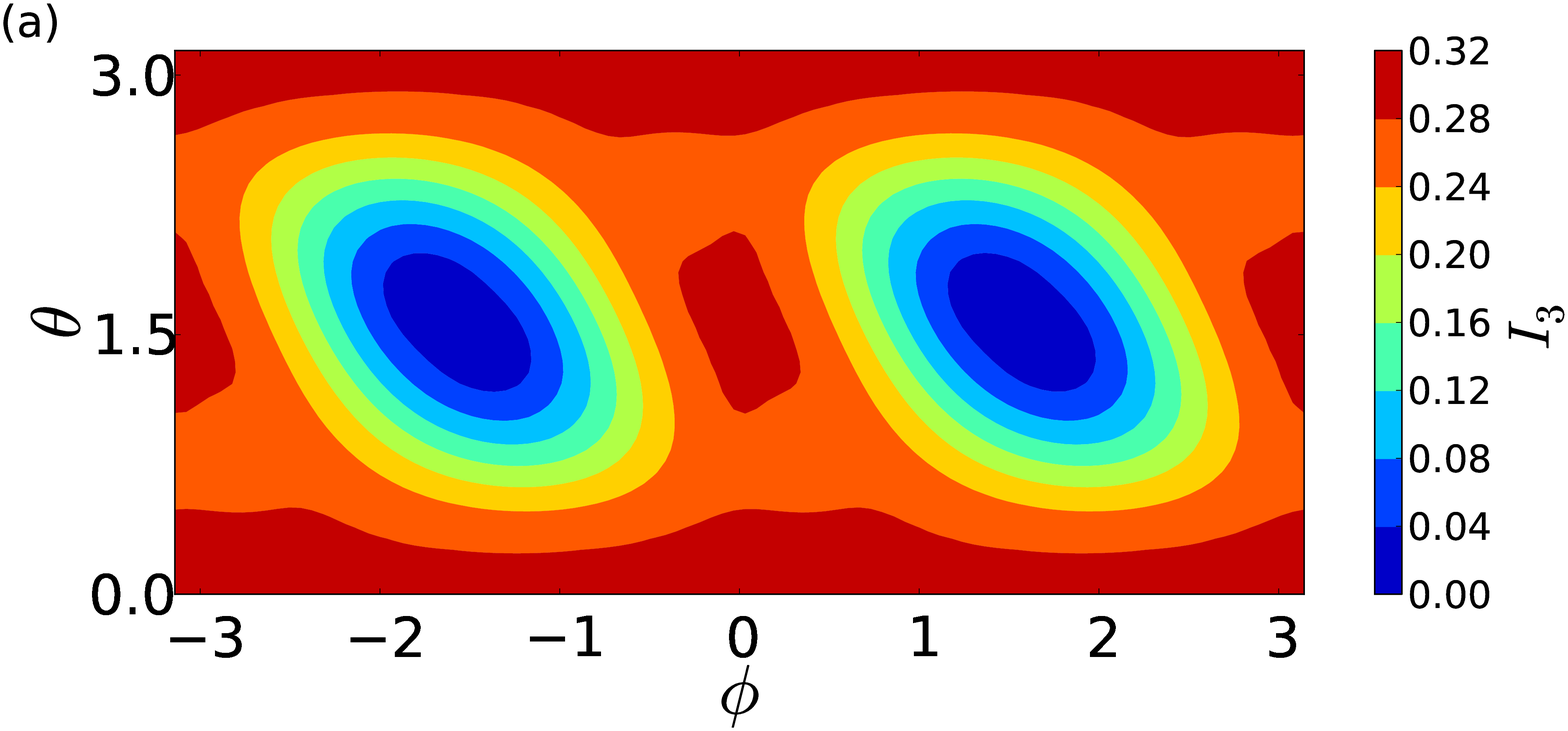}
    \\
    \includegraphics[width=0.49\textwidth]{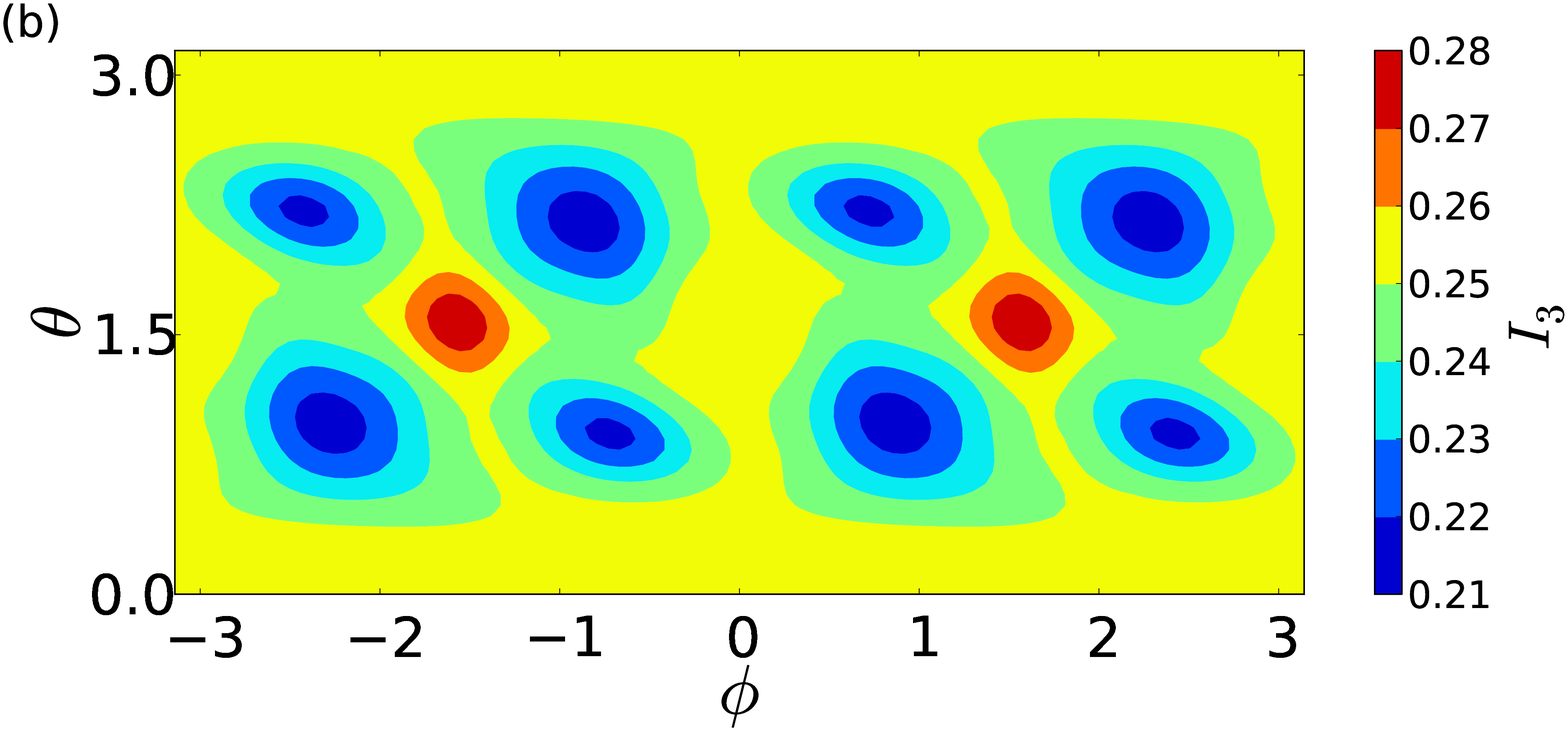}
    \\
    \includegraphics[width=0.49\textwidth]{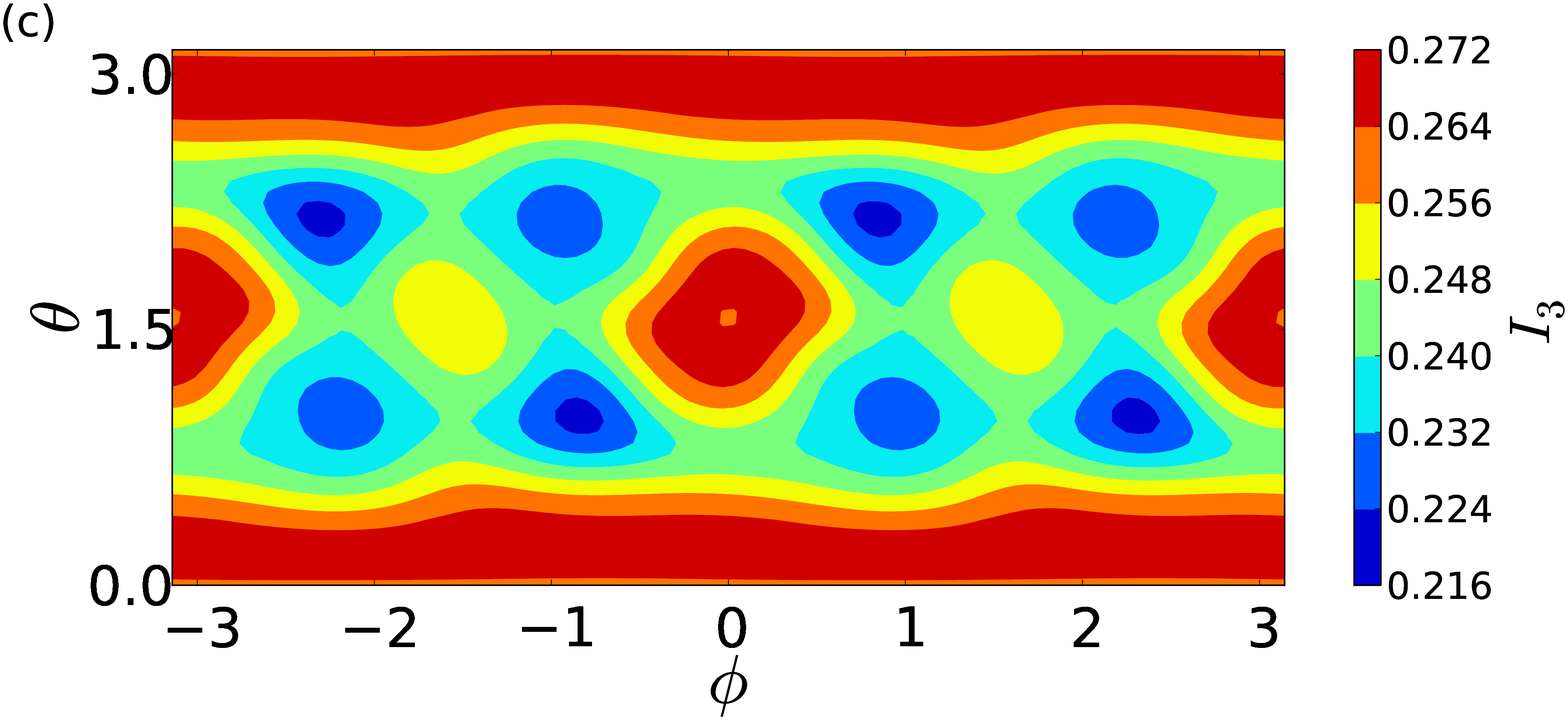}

    \caption{(Color online) Time-averaged TMI between three one qubit subsystems for $j = 6$ ($12$ qubits) and $p = \pi/2$, with the initial state corresponding to a coherent stated associated with a $50 \times 100$ grid discretizing $\theta \in [0,\ \pi)$ and $\phi \in [0,\ 2\pi)$. The time-averaged TMI is obtained for (a) $k = 1$ corresponding to a classically regular region, (b) $k = 3$, a mixed regime and (c) $k = 6$, a globally chaotic region, and is computed by averaging over $1000$ iterations starting from the initial coherent state. It can be that the classical structures are reproduced to a certain extent.}
    \label{fig:timeaveraged_tmi}
\end{figure}

\begin{figure}[!htpb]
    \includegraphics[width=0.49\textwidth]{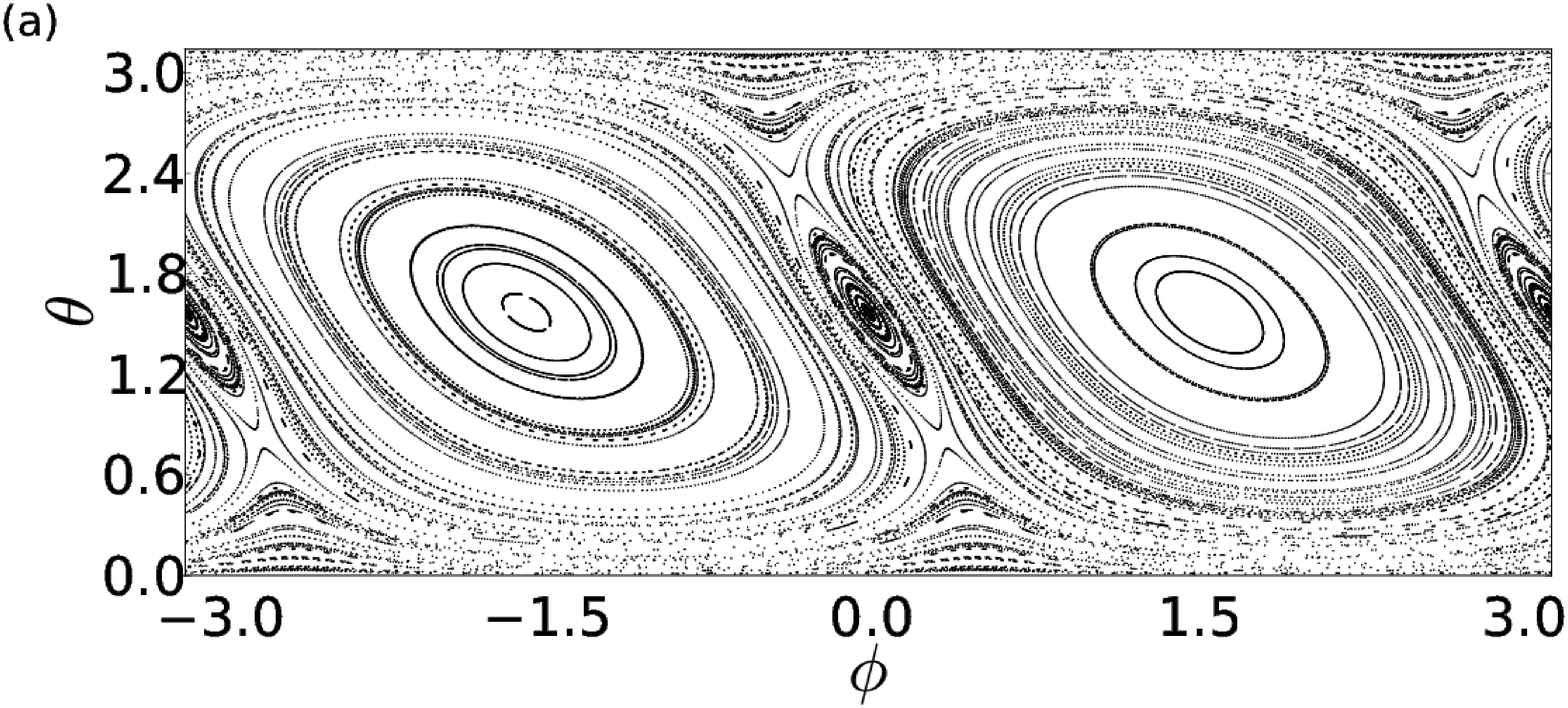}
    \\
    \includegraphics[width=0.49\textwidth]{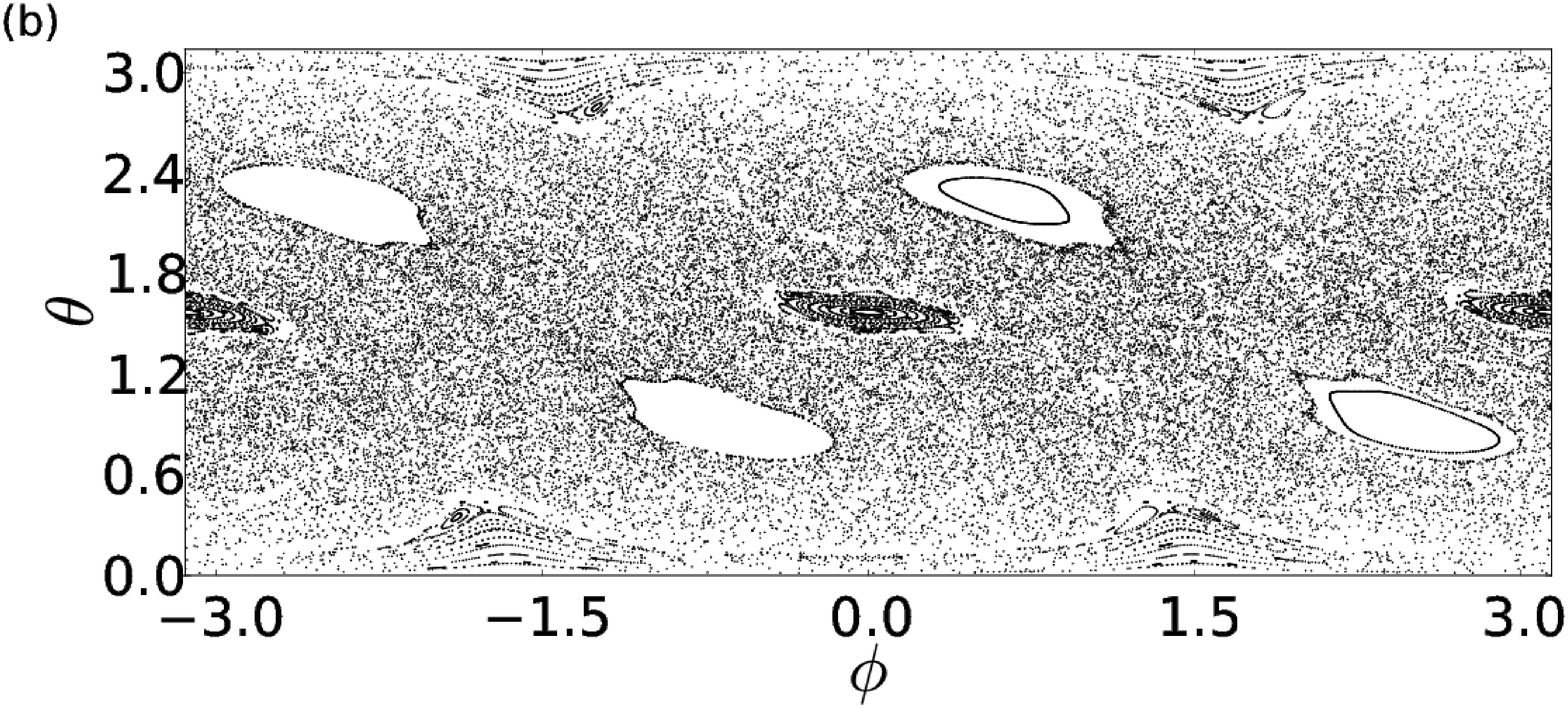}
    \\
    \includegraphics[width=0.49\textwidth]{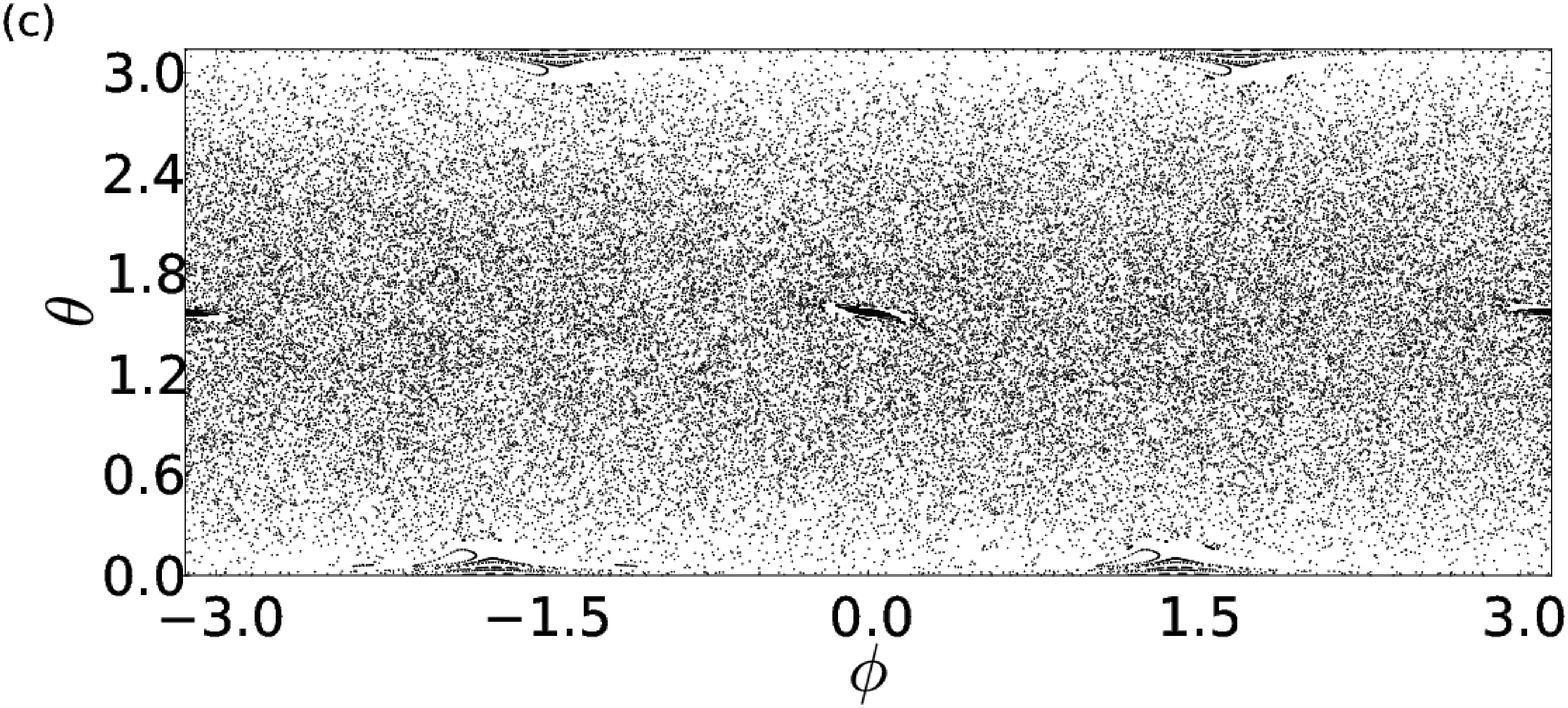}

    \caption{Phase portrait of the classical kicked top map, showing (a) globally regular dynamics for $k = 1$, (b) a mixed dynamics with coexisting regions of chaoticity and regularity for $k = 3$ and (c) (almost) globally chaotic dynamics for $k = 6$. $p = \pi/2$ is used in all the cases.}
    \label{fig:classicalphaseportrait}
\end{figure}

One can note, in particular, that the average TMI for a large number of points in the phase space is positive. This is true not only in the globally chaotic region, but also in the mixed and regular regimes. In the globally chaotic case, we can see that the average TMI varies approximately from $0.22$ to $0.27$, while the average TMI obtained from random PS states for the corresponding case is $\approx 0.24$. As before, this indicates that chaotic evolution nearly transforms the initially coherent states to random PS states.

\section{Summary and discussion}

In this paper, multi-qubit permutation symmetric states have been studied in general. An algorithm for getting the reduced density matrix of any number of qubits is given. This lead to the definition of a new random matrix ensemble that is of relevance in such cases, the random permutation symmetric states. It has been shown that the distribution of the  eigenvalues of the reduced density matrices of this ensemble, while similar to the Marchenko-Pastur is also different in crucial features and defines a new universality class.

For such a random ensemble we have derived analytical expressions for the average purity and linear entropy of arbitrary subsystems and approximate expressions for the average von Neumann entropy. In particular random permutation symmetric states have an ``area-law" kind of entanglement scaling only as $\log$(number of qubits in the subsystem). This small multipartite entanglement effectively results in positive tripartite mutual information, for which we derive analytical estimates. In comparison we point out that random states with no symmetry would have a typical TMI value that is negative.

We applied these general statements to a qubit model of the kicked top which can show a range of dynamical behaviour from the regular to the fully chaotic. An interesting aspect that has been recently investigated in certain quantum chaotic systems is the notion of scrambling of information, which refers to the delocalization with time of initially localized information across the system. Both OTOC and TMI have been proposed as a measure to detect scrambling of information, wherein a growth/decay of the OTOC or a negative value of TMI are considered as signatures of scrambling. In our analysis of the kicked top, we observe that OTOCs defined with observables on phase space grow exponentially fast with time (within the Ehrenfest time), while as noted above, permutation symmetry in the system is responsible for most states having a positive value of TMI. Thus the quantum chaos in this case results in scrambling in phase-space while remaining unscrambled in qubit-space.

While Iyoda and Sagawa have also pointed out that TMI for certain initial states are positive \cite{IyodaSagawa}, they argue that the effective dimension of the set of such states is small. If we consider the qubit model of the kicked top, the permutation symmetric subspace that we are restricted to is only $N + 1$ of the total $2^N$ dimensional one, and is therefore consistent with their observations. It will be interesting to explore the model in the non-permutation symmetric subspace when it may not have an interpretation as a kicked top. Nevertheless it is likely that the same parameter regimes lead to strongly chaotic behaviour, negative TMI and exponentially growing OTOC. We are also hopeful that our observations about permutation symmetric states will be of broad applicability.

\textit{Note added}: Along with the present study, another work appeared \cite{pappalardi2018scrambling} which investigates scrambling in long-range spin chains, including the case of kicked top.

\bibliography{permut_references}

\begin{thebibliography}{94}%
\makeatletter
\providecommand \@ifxundefined [1]{%
 \@ifx{#1\undefined}
}%
\providecommand \@ifnum [1]{%
 \ifnum #1\expandafter \@firstoftwo
 \else \expandafter \@secondoftwo
 \fi
}%
\providecommand \@ifx [1]{%
 \ifx #1\expandafter \@firstoftwo
 \else \expandafter \@secondoftwo
 \fi
}%
\providecommand \natexlab [1]{#1}%
\providecommand \enquote  [1]{``#1''}%
\providecommand \bibnamefont  [1]{#1}%
\providecommand \bibfnamefont [1]{#1}%
\providecommand \citenamefont [1]{#1}%
\providecommand \href@noop [0]{\@secondoftwo}%
\providecommand \href [0]{\begingroup \@sanitize@url \@href}%
\providecommand \@href[1]{\@@startlink{#1}\@@href}%
\providecommand \@@href[1]{\endgroup#1\@@endlink}%
\providecommand \@sanitize@url [0]{\catcode `\\12\catcode `\$12\catcode
  `\&12\catcode `\#12\catcode `\^12\catcode `\_12\catcode `\%12\relax}%
\providecommand \@@startlink[1]{}%
\providecommand \@@endlink[0]{}%
\providecommand \url  [0]{\begingroup\@sanitize@url \@url }%
\providecommand \@url [1]{\endgroup\@href {#1}{\urlprefix }}%
\providecommand \urlprefix  [0]{URL }%
\providecommand \Eprint [0]{\href }%
\providecommand \doibase [0]{http://dx.doi.org/}%
\providecommand \selectlanguage [0]{\@gobble}%
\providecommand \bibinfo  [0]{\@secondoftwo}%
\providecommand \bibfield  [0]{\@secondoftwo}%
\providecommand \translation [1]{[#1]}%
\providecommand \BibitemOpen [0]{}%
\providecommand \bibitemStop [0]{}%
\providecommand \bibitemNoStop [0]{.\EOS\space}%
\providecommand \EOS [0]{\spacefactor3000\relax}%
\providecommand \BibitemShut  [1]{\csname bibitem#1\endcsname}%
\let\auto@bib@innerbib\@empty
\bibitem [{\citenamefont {Arnol'd}\ and\ \citenamefont
  {Avez}(1968)}]{arnol1968ergodic}%
  \BibitemOpen
  \bibfield  {author} {\bibinfo {author} {\bibfnamefont {V.~I.}\ \bibnamefont
  {Arnol'd}}\ and\ \bibinfo {author} {\bibfnamefont {A.}~\bibnamefont {Avez}},\
  }\href@noop {} {\emph {\bibinfo {title} {Ergodic problems of classical
  mechanics}}}\ (\bibinfo  {publisher} {WA Benjamin},\ \bibinfo {year}
  {1968})\BibitemShut {NoStop}%
\bibitem [{\citenamefont {Gutzwiller}(1990)}]{Gutzwiller90}%
  \BibitemOpen
  \bibfield  {author} {\bibinfo {author} {\bibfnamefont {M.~C.}\ \bibnamefont
  {Gutzwiller}},\ }\href@noop {} {\emph {\bibinfo {title} {Chaos in Classical
  and Quantum Mechanics}}}\ (\bibinfo  {publisher} {Springer-Verlag},\ \bibinfo
  {address} {New York},\ \bibinfo {year} {1990})\BibitemShut {NoStop}%
\bibitem [{\citenamefont {Haake}(1991)}]{Haake}%
  \BibitemOpen
  \bibfield  {author} {\bibinfo {author} {\bibfnamefont {F.}~\bibnamefont
  {Haake}},\ }\href@noop {} {\emph {\bibinfo {title} {Quantum Signatures of
  Chaos}}}\ (\bibinfo  {publisher} {Spring-Verlag, Berlin},\ \bibinfo {year}
  {1991})\BibitemShut {NoStop}%
\bibitem [{\citenamefont {Gutzwiller}(1971)}]{gutzwiller1971periodic}%
  \BibitemOpen
  \bibfield  {author} {\bibinfo {author} {\bibfnamefont {M.~C.}\ \bibnamefont
  {Gutzwiller}},\ }\href@noop {} {\bibfield  {journal} {\bibinfo  {journal}
  {Journal of Mathematical Physics}\ }\textbf {\bibinfo {volume} {12}},\
  \bibinfo {pages} {343} (\bibinfo {year} {1971})}\BibitemShut {NoStop}%
\bibitem [{\citenamefont {Berry}\ and\ \citenamefont {Tabor}(1977)}]{Berry375}%
  \BibitemOpen
  \bibfield  {author} {\bibinfo {author} {\bibfnamefont {M.}~\bibnamefont
  {Berry}}\ and\ \bibinfo {author} {\bibfnamefont {M.}~\bibnamefont {Tabor}},\
  }\href@noop {} {\bibfield  {journal} {\bibinfo  {journal} {Proceedings of the
  Royal Society of London A: Mathematical, Physical and Engineering Sciences}\
  }\textbf {\bibinfo {volume} {356}},\ \bibinfo {pages} {375} (\bibinfo {year}
  {1977})}\BibitemShut {NoStop}%
\bibitem [{\citenamefont {Bohigas}\ \emph {et~al.}(1984)\citenamefont
  {Bohigas}, \citenamefont {Giannoni},\ and\ \citenamefont
  {Schmit}}]{bohigas1984}%
  \BibitemOpen
  \bibfield  {author} {\bibinfo {author} {\bibfnamefont {O.}~\bibnamefont
  {Bohigas}}, \bibinfo {author} {\bibfnamefont {M.-J.}\ \bibnamefont
  {Giannoni}}, \ and\ \bibinfo {author} {\bibfnamefont {C.}~\bibnamefont
  {Schmit}},\ }\href@noop {} {\bibfield  {journal} {\bibinfo  {journal}
  {Physical Review Letters}\ }\textbf {\bibinfo {volume} {52}},\ \bibinfo
  {pages} {1} (\bibinfo {year} {1984})}\BibitemShut {NoStop}%
\bibitem [{\citenamefont {Robnik}\ and\ \citenamefont
  {Berry}(1985)}]{robnik1985classical}%
  \BibitemOpen
  \bibfield  {author} {\bibinfo {author} {\bibfnamefont {M.}~\bibnamefont
  {Robnik}}\ and\ \bibinfo {author} {\bibfnamefont {M.~V.}\ \bibnamefont
  {Berry}},\ }\href@noop {} {\bibfield  {journal} {\bibinfo  {journal} {Journal
  of Physics A: Mathematical and General}\ }\textbf {\bibinfo {volume} {18}},\
  \bibinfo {pages} {1361} (\bibinfo {year} {1985})}\BibitemShut {NoStop}%
\bibitem [{\citenamefont {Berry}(1977)}]{Berry77a}%
  \BibitemOpen
  \bibfield  {author} {\bibinfo {author} {\bibfnamefont {M.~V.}\ \bibnamefont
  {Berry}},\ }\href@noop {} {\bibfield  {journal} {\bibinfo  {journal} {Journal
  of Physics A: Mathematical and General}\ }\textbf {\bibinfo {volume} {10}},\
  \bibinfo {pages} {2083} (\bibinfo {year} {1977})}\BibitemShut {NoStop}%
\bibitem [{\citenamefont {Voros}(1976)}]{voros1976semi}%
  \BibitemOpen
  \bibfield  {author} {\bibinfo {author} {\bibfnamefont {A.}~\bibnamefont
  {Voros}},\ }\href@noop {} {\bibfield  {journal} {\bibinfo  {journal} {Ann.
  Inst. Henri Poincar{\'e} A}\ }\textbf {\bibinfo {volume} {24}},\ \bibinfo
  {pages} {31} (\bibinfo {year} {1976})}\BibitemShut {NoStop}%
\bibitem [{\citenamefont {McDonald}\ and\ \citenamefont
  {Kaufman}(1979)}]{McDonald}%
  \BibitemOpen
  \bibfield  {author} {\bibinfo {author} {\bibfnamefont {S.~W.}\ \bibnamefont
  {McDonald}}\ and\ \bibinfo {author} {\bibfnamefont {A.~N.}\ \bibnamefont
  {Kaufman}},\ }\href {\doibase 10.1103/PhysRevLett.42.1189} {\bibfield
  {journal} {\bibinfo  {journal} {Physical Review Letters}\ }\textbf {\bibinfo
  {volume} {42}},\ \bibinfo {pages} {1189} (\bibinfo {year}
  {1979})}\BibitemShut {NoStop}%
\bibitem [{\citenamefont {Heller}(1984)}]{heller1984bound}%
  \BibitemOpen
  \bibfield  {author} {\bibinfo {author} {\bibfnamefont {E.~J.}\ \bibnamefont
  {Heller}},\ }\href@noop {} {\bibfield  {journal} {\bibinfo  {journal}
  {Physical Review Letters}\ }\textbf {\bibinfo {volume} {53}},\ \bibinfo
  {pages} {1515} (\bibinfo {year} {1984})}\BibitemShut {NoStop}%
\bibitem [{\citenamefont {Berry}\ \emph {et~al.}(1979)\citenamefont {Berry},
  \citenamefont {Balazs}, \citenamefont {Tabor},\ and\ \citenamefont
  {Voros}}]{BERRY197926}%
  \BibitemOpen
  \bibfield  {author} {\bibinfo {author} {\bibfnamefont {M.}~\bibnamefont
  {Berry}}, \bibinfo {author} {\bibfnamefont {N.}~\bibnamefont {Balazs}},
  \bibinfo {author} {\bibfnamefont {M.}~\bibnamefont {Tabor}}, \ and\ \bibinfo
  {author} {\bibfnamefont {A.}~\bibnamefont {Voros}},\ }\href {\doibase
  https://doi.org/10.1016/0003-4916(79)90296-3} {\bibfield  {journal} {\bibinfo
   {journal} {Annals of Physics}\ }\textbf {\bibinfo {volume} {122}},\ \bibinfo
  {pages} {26 } (\bibinfo {year} {1979})}\BibitemShut {NoStop}%
\bibitem [{\citenamefont {{\AA}berg}(1990)}]{aberg1990onset}%
  \BibitemOpen
  \bibfield  {author} {\bibinfo {author} {\bibfnamefont {S.}~\bibnamefont
  {{\AA}berg}},\ }\href@noop {} {\bibfield  {journal} {\bibinfo  {journal}
  {Physical Review Letters}\ }\textbf {\bibinfo {volume} {64}},\ \bibinfo
  {pages} {3119} (\bibinfo {year} {1990})}\BibitemShut {NoStop}%
\bibitem [{\citenamefont {Montambaux}\ \emph {et~al.}(1993)\citenamefont
  {Montambaux}, \citenamefont {Poilblanc}, \citenamefont {Bellissard},\ and\
  \citenamefont {Sire}}]{montambaux1993quantum}%
  \BibitemOpen
  \bibfield  {author} {\bibinfo {author} {\bibfnamefont {G.}~\bibnamefont
  {Montambaux}}, \bibinfo {author} {\bibfnamefont {D.}~\bibnamefont
  {Poilblanc}}, \bibinfo {author} {\bibfnamefont {J.}~\bibnamefont
  {Bellissard}}, \ and\ \bibinfo {author} {\bibfnamefont {C.}~\bibnamefont
  {Sire}},\ }\href {\doibase 10.1103/PhysRevLett.70.497} {\bibfield  {journal}
  {\bibinfo  {journal} {Physical Review Letters}\ }\textbf {\bibinfo {volume}
  {70}},\ \bibinfo {pages} {497} (\bibinfo {year} {1993})}\BibitemShut
  {NoStop}%
\bibitem [{\citenamefont {Hsu}\ and\ \citenamefont
  {Angles~d’Auriac}(1993)}]{hsu1993level}%
  \BibitemOpen
  \bibfield  {author} {\bibinfo {author} {\bibfnamefont {T.~C.}\ \bibnamefont
  {Hsu}}\ and\ \bibinfo {author} {\bibfnamefont {J.}~\bibnamefont
  {Angles~d’Auriac}},\ }\href@noop {} {\bibfield  {journal} {\bibinfo
  {journal} {Physical Review B}\ }\textbf {\bibinfo {volume} {47}},\ \bibinfo
  {pages} {14291} (\bibinfo {year} {1993})}\BibitemShut {NoStop}%
\bibitem [{\citenamefont {van~der Pals}\ and\ \citenamefont
  {Gaspard}(1994)}]{van1994two}%
  \BibitemOpen
  \bibfield  {author} {\bibinfo {author} {\bibfnamefont {P.~v.~E.}\
  \bibnamefont {van~der Pals}}\ and\ \bibinfo {author} {\bibfnamefont
  {P.}~\bibnamefont {Gaspard}},\ }\href@noop {} {\bibfield  {journal} {\bibinfo
   {journal} {Physical Review E}\ }\textbf {\bibinfo {volume} {49}},\ \bibinfo
  {pages} {79} (\bibinfo {year} {1994})}\BibitemShut {NoStop}%
\bibitem [{\citenamefont {Kudo}\ and\ \citenamefont
  {Deguchi}(2003)}]{kudo2003unexpected}%
  \BibitemOpen
  \bibfield  {author} {\bibinfo {author} {\bibfnamefont {K.}~\bibnamefont
  {Kudo}}\ and\ \bibinfo {author} {\bibfnamefont {T.}~\bibnamefont {Deguchi}},\
  }\href@noop {} {\bibfield  {journal} {\bibinfo  {journal} {Physical Review
  B}\ }\textbf {\bibinfo {volume} {68}},\ \bibinfo {pages} {052510} (\bibinfo
  {year} {2003})}\BibitemShut {NoStop}%
\bibitem [{\citenamefont {Rabson}\ \emph {et~al.}(2004)\citenamefont {Rabson},
  \citenamefont {Narozhny},\ and\ \citenamefont
  {Millis}}]{rabson2004crossover}%
  \BibitemOpen
  \bibfield  {author} {\bibinfo {author} {\bibfnamefont {D.~A.}\ \bibnamefont
  {Rabson}}, \bibinfo {author} {\bibfnamefont {B.}~\bibnamefont {Narozhny}}, \
  and\ \bibinfo {author} {\bibfnamefont {A.}~\bibnamefont {Millis}},\
  }\href@noop {} {\bibfield  {journal} {\bibinfo  {journal} {Physical Review
  B}\ }\textbf {\bibinfo {volume} {69}},\ \bibinfo {pages} {054403} (\bibinfo
  {year} {2004})}\BibitemShut {NoStop}%
\bibitem [{\citenamefont {Jacquod}\ and\ \citenamefont
  {Shepelyansky}(1997)}]{jacquod1997emergence}%
  \BibitemOpen
  \bibfield  {author} {\bibinfo {author} {\bibfnamefont {P.}~\bibnamefont
  {Jacquod}}\ and\ \bibinfo {author} {\bibfnamefont {D.}~\bibnamefont
  {Shepelyansky}},\ }\href@noop {} {\bibfield  {journal} {\bibinfo  {journal}
  {Physical Review Letters}\ }\textbf {\bibinfo {volume} {79}},\ \bibinfo
  {pages} {1837} (\bibinfo {year} {1997})}\BibitemShut {NoStop}%
\bibitem [{\citenamefont {Georgeot}\ and\ \citenamefont
  {Shepelyansky}(1998)}]{georgeot1998integrability}%
  \BibitemOpen
  \bibfield  {author} {\bibinfo {author} {\bibfnamefont {B.}~\bibnamefont
  {Georgeot}}\ and\ \bibinfo {author} {\bibfnamefont {D.}~\bibnamefont
  {Shepelyansky}},\ }\href@noop {} {\bibfield  {journal} {\bibinfo  {journal}
  {Physical Review Letters}\ }\textbf {\bibinfo {volume} {81}},\ \bibinfo
  {pages} {5129} (\bibinfo {year} {1998})}\BibitemShut {NoStop}%
\bibitem [{\citenamefont {Lakshminarayan}\ and\ \citenamefont
  {Subrahmanyam}(2003)}]{LakSub2003}%
  \BibitemOpen
  \bibfield  {author} {\bibinfo {author} {\bibfnamefont {A.}~\bibnamefont
  {Lakshminarayan}}\ and\ \bibinfo {author} {\bibfnamefont {V.}~\bibnamefont
  {Subrahmanyam}},\ }\href@noop {} {\bibfield  {journal} {\bibinfo  {journal}
  {Physical Review A}\ }\textbf {\bibinfo {volume} {67}},\ \bibinfo {pages}
  {052304} (\bibinfo {year} {2003})}\BibitemShut {NoStop}%
\bibitem [{\citenamefont {Avishai}\ \emph {et~al.}(2002)\citenamefont
  {Avishai}, \citenamefont {Richert},\ and\ \citenamefont
  {Berkovits}}]{avishai2002level}%
  \BibitemOpen
  \bibfield  {author} {\bibinfo {author} {\bibfnamefont {Y.}~\bibnamefont
  {Avishai}}, \bibinfo {author} {\bibfnamefont {J.}~\bibnamefont {Richert}}, \
  and\ \bibinfo {author} {\bibfnamefont {R.}~\bibnamefont {Berkovits}},\
  }\href@noop {} {\bibfield  {journal} {\bibinfo  {journal} {Physical Review
  B}\ }\textbf {\bibinfo {volume} {66}},\ \bibinfo {pages} {052416} (\bibinfo
  {year} {2002})}\BibitemShut {NoStop}%
\bibitem [{\citenamefont {Santos}(2004)}]{santos2004integrability}%
  \BibitemOpen
  \bibfield  {author} {\bibinfo {author} {\bibfnamefont {L.}~\bibnamefont
  {Santos}},\ }\href@noop {} {\bibfield  {journal} {\bibinfo  {journal}
  {Journal of Physics A: Mathematical and General}\ }\textbf {\bibinfo {volume}
  {37}},\ \bibinfo {pages} {4723} (\bibinfo {year} {2004})}\BibitemShut
  {NoStop}%
\bibitem [{\citenamefont {Kudo}\ and\ \citenamefont
  {Deguchi}(2004)}]{kudo2004level}%
  \BibitemOpen
  \bibfield  {author} {\bibinfo {author} {\bibfnamefont {K.}~\bibnamefont
  {Kudo}}\ and\ \bibinfo {author} {\bibfnamefont {T.}~\bibnamefont {Deguchi}},\
  }\href@noop {} {\bibfield  {journal} {\bibinfo  {journal} {Physical Review
  B}\ }\textbf {\bibinfo {volume} {69}},\ \bibinfo {pages} {132404} (\bibinfo
  {year} {2004})}\BibitemShut {NoStop}%
\bibitem [{\citenamefont {Lakshminarayan}\ and\ \citenamefont
  {Subrahmanyam}(2005)}]{LakSub2005}%
  \BibitemOpen
  \bibfield  {author} {\bibinfo {author} {\bibfnamefont {A.}~\bibnamefont
  {Lakshminarayan}}\ and\ \bibinfo {author} {\bibfnamefont {V.}~\bibnamefont
  {Subrahmanyam}},\ }\href {\doibase 10.1103/PhysRevA.71.062334} {\bibfield
  {journal} {\bibinfo  {journal} {Physical Review A}\ }\textbf {\bibinfo
  {volume} {71}},\ \bibinfo {pages} {062334} (\bibinfo {year}
  {2005})}\BibitemShut {NoStop}%
\bibitem [{\citenamefont {Karthik}\ \emph {et~al.}(2007)\citenamefont
  {Karthik}, \citenamefont {Sharma},\ and\ \citenamefont
  {Lakshminarayan}}]{KarthikSharmaLak2007}%
  \BibitemOpen
  \bibfield  {author} {\bibinfo {author} {\bibfnamefont {J.}~\bibnamefont
  {Karthik}}, \bibinfo {author} {\bibfnamefont {A.}~\bibnamefont {Sharma}}, \
  and\ \bibinfo {author} {\bibfnamefont {A.}~\bibnamefont {Lakshminarayan}},\
  }\href {\doibase 10.1103/PhysRevA.75.022304} {\bibfield  {journal} {\bibinfo
  {journal} {Physical Review A}\ }\textbf {\bibinfo {volume} {75}},\ \bibinfo
  {pages} {022304} (\bibinfo {year} {2007})}\BibitemShut {NoStop}%
\bibitem [{\citenamefont {D'Alessio}\ \emph {et~al.}(2016)\citenamefont
  {D'Alessio}, \citenamefont {Kafri}, \citenamefont {Polkovnikov},\ and\
  \citenamefont {Rigol}}]{Rigol2016}%
  \BibitemOpen
  \bibfield  {author} {\bibinfo {author} {\bibfnamefont {L.}~\bibnamefont
  {D'Alessio}}, \bibinfo {author} {\bibfnamefont {Y.}~\bibnamefont {Kafri}},
  \bibinfo {author} {\bibfnamefont {A.}~\bibnamefont {Polkovnikov}}, \ and\
  \bibinfo {author} {\bibfnamefont {M.}~\bibnamefont {Rigol}},\ }\href
  {\doibase 10.1080/00018732.2016.1198134} {\bibfield  {journal} {\bibinfo
  {journal} {Advances in Physics}\ }\textbf {\bibinfo {volume} {65}},\ \bibinfo
  {pages} {239} (\bibinfo {year} {2016})},\ \Eprint
  {http://arxiv.org/abs/https://doi.org/10.1080/00018732.2016.1198134}
  {https://doi.org/10.1080/00018732.2016.1198134} \BibitemShut {NoStop}%
\bibitem [{\citenamefont {Maldacena}\ and\ \citenamefont
  {Stanford}(2016)}]{MaldacenaSYK15}%
  \BibitemOpen
  \bibfield  {author} {\bibinfo {author} {\bibfnamefont {J.}~\bibnamefont
  {Maldacena}}\ and\ \bibinfo {author} {\bibfnamefont {D.}~\bibnamefont
  {Stanford}},\ }\href {\doibase 10.1103/PhysRevD.94.106002} {\bibfield
  {journal} {\bibinfo  {journal} {Physical Review D}\ }\textbf {\bibinfo
  {volume} {94}},\ \bibinfo {pages} {106002} (\bibinfo {year}
  {2016})}\BibitemShut {NoStop}%
\bibitem [{\citenamefont {Bohigas}\ and\ \citenamefont
  {Flores}(1971)}]{bohigas1971spacing}%
  \BibitemOpen
  \bibfield  {author} {\bibinfo {author} {\bibfnamefont {O.}~\bibnamefont
  {Bohigas}}\ and\ \bibinfo {author} {\bibfnamefont {J.}~\bibnamefont
  {Flores}},\ }\href@noop {} {\bibfield  {journal} {\bibinfo  {journal}
  {Physics Letters B}\ }\textbf {\bibinfo {volume} {35}},\ \bibinfo {pages}
  {383} (\bibinfo {year} {1971})}\BibitemShut {NoStop}%
\bibitem [{\citenamefont {Kota}(2001)}]{Kota}%
  \BibitemOpen
  \bibfield  {author} {\bibinfo {author} {\bibfnamefont {V.}~\bibnamefont
  {Kota}},\ }\href@noop {} {\bibfield  {journal} {\bibinfo  {journal} {Physics
  Reports}\ }\textbf {\bibinfo {volume} {347}},\ \bibinfo {pages} {223}
  (\bibinfo {year} {2001})}\BibitemShut {NoStop}%
\bibitem [{\citenamefont {Mehta}(2004)}]{mehta2004random}%
  \BibitemOpen
  \bibfield  {author} {\bibinfo {author} {\bibfnamefont {M.~L.}\ \bibnamefont
  {Mehta}},\ }\href@noop {} {\emph {\bibinfo {title} {Random matrices}}},\
  Vol.\ \bibinfo {volume} {142}\ (\bibinfo  {publisher} {Elsevier},\ \bibinfo
  {year} {2004})\BibitemShut {NoStop}%
\bibitem [{\citenamefont {Neill}\ \emph {et~al.}(2016)\citenamefont {Neill},
  \citenamefont {Roushan}, \citenamefont {Fang}, \citenamefont {Chen},
  \citenamefont {Kolodrubetz}, \citenamefont {Chen}, \citenamefont {Megrant},
  \citenamefont {Barends}, \citenamefont {Campbell}, \citenamefont {Chiaro},
  \citenamefont {Dunsworth}, \citenamefont {Jeffrey}, \citenamefont {Kelly},
  \citenamefont {Mutus}, \citenamefont {O’Malley}, \citenamefont {Quintana},
  \citenamefont {Sank}, \citenamefont {Vainsencher}, \citenamefont {Wenner},
  \citenamefont {White}, \citenamefont {Polkovnikov},\ and\ \citenamefont
  {Martinis}}]{Neill16}%
  \BibitemOpen
  \bibfield  {author} {\bibinfo {author} {\bibfnamefont {C.}~\bibnamefont
  {Neill}}, \bibinfo {author} {\bibfnamefont {P.}~\bibnamefont {Roushan}},
  \bibinfo {author} {\bibfnamefont {M.}~\bibnamefont {Fang}}, \bibinfo {author}
  {\bibfnamefont {Y.}~\bibnamefont {Chen}}, \bibinfo {author} {\bibfnamefont
  {M.}~\bibnamefont {Kolodrubetz}}, \bibinfo {author} {\bibfnamefont
  {Z.}~\bibnamefont {Chen}}, \bibinfo {author} {\bibfnamefont {A.}~\bibnamefont
  {Megrant}}, \bibinfo {author} {\bibfnamefont {R.}~\bibnamefont {Barends}},
  \bibinfo {author} {\bibfnamefont {B.}~\bibnamefont {Campbell}}, \bibinfo
  {author} {\bibfnamefont {B.}~\bibnamefont {Chiaro}}, \bibinfo {author}
  {\bibfnamefont {A.}~\bibnamefont {Dunsworth}}, \bibinfo {author}
  {\bibfnamefont {E.}~\bibnamefont {Jeffrey}}, \bibinfo {author} {\bibfnamefont
  {J.}~\bibnamefont {Kelly}}, \bibinfo {author} {\bibfnamefont
  {J.}~\bibnamefont {Mutus}}, \bibinfo {author} {\bibfnamefont {P.~J.~J.}\
  \bibnamefont {O’Malley}}, \bibinfo {author} {\bibfnamefont
  {C.}~\bibnamefont {Quintana}}, \bibinfo {author} {\bibfnamefont
  {D.}~\bibnamefont {Sank}}, \bibinfo {author} {\bibfnamefont {A.}~\bibnamefont
  {Vainsencher}}, \bibinfo {author} {\bibfnamefont {J.}~\bibnamefont {Wenner}},
  \bibinfo {author} {\bibfnamefont {T.~C.}\ \bibnamefont {White}}, \bibinfo
  {author} {\bibfnamefont {A.}~\bibnamefont {Polkovnikov}}, \ and\ \bibinfo
  {author} {\bibfnamefont {J.~M.}\ \bibnamefont {Martinis}},\ }\href@noop {}
  {\bibfield  {journal} {\bibinfo  {journal} {Nature Physics}\ }\textbf
  {\bibinfo {volume} {12}},\ \bibinfo {pages} {1037} (\bibinfo {year}
  {2016})}\BibitemShut {NoStop}%
\bibitem [{\citenamefont {Kaufman}\ \emph {et~al.}(2016)\citenamefont
  {Kaufman}, \citenamefont {Tai}, \citenamefont {Lukin}, \citenamefont
  {Rispoli}, \citenamefont {Schittko}, \citenamefont {Preiss},\ and\
  \citenamefont {Greiner}}]{kaufman2016quantum}%
  \BibitemOpen
  \bibfield  {author} {\bibinfo {author} {\bibfnamefont {A.~M.}\ \bibnamefont
  {Kaufman}}, \bibinfo {author} {\bibfnamefont {M.~E.}\ \bibnamefont {Tai}},
  \bibinfo {author} {\bibfnamefont {A.}~\bibnamefont {Lukin}}, \bibinfo
  {author} {\bibfnamefont {M.}~\bibnamefont {Rispoli}}, \bibinfo {author}
  {\bibfnamefont {R.}~\bibnamefont {Schittko}}, \bibinfo {author}
  {\bibfnamefont {P.~M.}\ \bibnamefont {Preiss}}, \ and\ \bibinfo {author}
  {\bibfnamefont {M.}~\bibnamefont {Greiner}},\ }\href@noop {} {\bibfield
  {journal} {\bibinfo  {journal} {Science}\ }\textbf {\bibinfo {volume}
  {353}},\ \bibinfo {pages} {794} (\bibinfo {year} {2016})}\BibitemShut
  {NoStop}%
\bibitem [{\citenamefont {Miller}\ and\ \citenamefont
  {Sarkar}(1999)}]{miller1999signatures}%
  \BibitemOpen
  \bibfield  {author} {\bibinfo {author} {\bibfnamefont {P.~A.}\ \bibnamefont
  {Miller}}\ and\ \bibinfo {author} {\bibfnamefont {S.}~\bibnamefont
  {Sarkar}},\ }\href@noop {} {\bibfield  {journal} {\bibinfo  {journal}
  {Physical Review E}\ }\textbf {\bibinfo {volume} {60}},\ \bibinfo {pages}
  {1542} (\bibinfo {year} {1999})}\BibitemShut {NoStop}%
\bibitem [{\citenamefont
  {Lakshminarayan}(2001)}]{lakshminarayan2001entangling}%
  \BibitemOpen
  \bibfield  {author} {\bibinfo {author} {\bibfnamefont {A.}~\bibnamefont
  {Lakshminarayan}},\ }\href@noop {} {\bibfield  {journal} {\bibinfo  {journal}
  {Physical Review E}\ }\textbf {\bibinfo {volume} {64}},\ \bibinfo {pages}
  {036207} (\bibinfo {year} {2001})}\BibitemShut {NoStop}%
\bibitem [{\citenamefont {Bandyopadhyay}\ and\ \citenamefont
  {Lakshminarayan}(2002)}]{bandyopadhyay2002testing}%
  \BibitemOpen
  \bibfield  {author} {\bibinfo {author} {\bibfnamefont {J.~N.}\ \bibnamefont
  {Bandyopadhyay}}\ and\ \bibinfo {author} {\bibfnamefont {A.}~\bibnamefont
  {Lakshminarayan}},\ }\href@noop {} {\bibfield  {journal} {\bibinfo  {journal}
  {Physical Review Letters}\ }\textbf {\bibinfo {volume} {89}},\ \bibinfo
  {pages} {060402} (\bibinfo {year} {2002})}\BibitemShut {NoStop}%
\bibitem [{\citenamefont {Scott}\ and\ \citenamefont
  {Caves}(2003)}]{ScottCaves2003}%
  \BibitemOpen
  \bibfield  {author} {\bibinfo {author} {\bibfnamefont {A.~J.}\ \bibnamefont
  {Scott}}\ and\ \bibinfo {author} {\bibfnamefont {C.~M.}\ \bibnamefont
  {Caves}},\ }\href@noop {} {\bibfield  {journal} {\bibinfo  {journal} {Journal
  of Physics A: Mathematical and General}\ }\textbf {\bibinfo {volume} {36}},\
  \bibinfo {pages} {9553} (\bibinfo {year} {2003})}\BibitemShut {NoStop}%
\bibitem [{\citenamefont
  {Jacquod}(2004{\natexlab{a}})}]{jacquod2004semiclassical}%
  \BibitemOpen
  \bibfield  {author} {\bibinfo {author} {\bibfnamefont {P.}~\bibnamefont
  {Jacquod}},\ }\href@noop {} {\bibfield  {journal} {\bibinfo  {journal}
  {Physical Review Letters}\ }\textbf {\bibinfo {volume} {92}},\ \bibinfo
  {pages} {150403} (\bibinfo {year} {2004}{\natexlab{a}})}\BibitemShut
  {NoStop}%
\bibitem [{\citenamefont {Jacquod}(2004{\natexlab{b}})}]{Jacq2004Erratum}%
  \BibitemOpen
  \bibfield  {author} {\bibinfo {author} {\bibfnamefont {P.}~\bibnamefont
  {Jacquod}},\ }\href {\doibase 10.1103/PhysRevLett.93.219903} {\bibfield
  {journal} {\bibinfo  {journal} {Physical Review Letters}\ }\textbf {\bibinfo
  {volume} {93}},\ \bibinfo {pages} {219903} (\bibinfo {year}
  {2004}{\natexlab{b}})}\BibitemShut {NoStop}%
\bibitem [{\citenamefont {Petitjean}\ and\ \citenamefont
  {Jacquod}(2006)}]{PetitJacq2006}%
  \BibitemOpen
  \bibfield  {author} {\bibinfo {author} {\bibfnamefont {C.}~\bibnamefont
  {Petitjean}}\ and\ \bibinfo {author} {\bibfnamefont {P.}~\bibnamefont
  {Jacquod}},\ }\href {\doibase 10.1103/PhysRevLett.97.194103} {\bibfield
  {journal} {\bibinfo  {journal} {Physical Review Letters}\ }\textbf {\bibinfo
  {volume} {97}},\ \bibinfo {pages} {194103} (\bibinfo {year}
  {2006})}\BibitemShut {NoStop}%
\bibitem [{\citenamefont {Ghose}\ and\ \citenamefont {Sanders}(2004)}]{Ghose}%
  \BibitemOpen
  \bibfield  {author} {\bibinfo {author} {\bibfnamefont {S.}~\bibnamefont
  {Ghose}}\ and\ \bibinfo {author} {\bibfnamefont {B.~C.}\ \bibnamefont
  {Sanders}},\ }\href@noop {} {\bibfield  {journal} {\bibinfo  {journal}
  {Physical Review A}\ }\textbf {\bibinfo {volume} {70}},\ \bibinfo {pages}
  {062315} (\bibinfo {year} {2004})}\BibitemShut {NoStop}%
\bibitem [{\citenamefont {Wang}\ \emph {et~al.}(2004)\citenamefont {Wang},
  \citenamefont {Ghose}, \citenamefont {Sanders},\ and\ \citenamefont
  {Hu}}]{Wang2004}%
  \BibitemOpen
  \bibfield  {author} {\bibinfo {author} {\bibfnamefont {X.}~\bibnamefont
  {Wang}}, \bibinfo {author} {\bibfnamefont {S.}~\bibnamefont {Ghose}},
  \bibinfo {author} {\bibfnamefont {B.~C.}\ \bibnamefont {Sanders}}, \ and\
  \bibinfo {author} {\bibfnamefont {B.}~\bibnamefont {Hu}},\ }\href@noop {}
  {\bibfield  {journal} {\bibinfo  {journal} {Physical Review E}\ }\textbf
  {\bibinfo {volume} {70}},\ \bibinfo {pages} {016217} (\bibinfo {year}
  {2004})}\BibitemShut {NoStop}%
\bibitem [{\citenamefont {Trail}\ \emph {et~al.}(2008)\citenamefont {Trail},
  \citenamefont {Madhok},\ and\ \citenamefont
  {Deutsch}}]{trail2008entanglement}%
  \BibitemOpen
  \bibfield  {author} {\bibinfo {author} {\bibfnamefont {C.~M.}\ \bibnamefont
  {Trail}}, \bibinfo {author} {\bibfnamefont {V.}~\bibnamefont {Madhok}}, \
  and\ \bibinfo {author} {\bibfnamefont {I.~H.}\ \bibnamefont {Deutsch}},\
  }\href@noop {} {\bibfield  {journal} {\bibinfo  {journal} {Physical Review
  E}\ }\textbf {\bibinfo {volume} {78}},\ \bibinfo {pages} {046211} (\bibinfo
  {year} {2008})}\BibitemShut {NoStop}%
\bibitem [{\citenamefont {Madhok}\ \emph
  {et~al.}(2015{\natexlab{a}})\citenamefont {Madhok}, \citenamefont {Gupta},
  \citenamefont {Trottier},\ and\ \citenamefont
  {Ghose}}]{madhok2015signatures}%
  \BibitemOpen
  \bibfield  {author} {\bibinfo {author} {\bibfnamefont {V.}~\bibnamefont
  {Madhok}}, \bibinfo {author} {\bibfnamefont {V.}~\bibnamefont {Gupta}},
  \bibinfo {author} {\bibfnamefont {D.-A.}\ \bibnamefont {Trottier}}, \ and\
  \bibinfo {author} {\bibfnamefont {S.}~\bibnamefont {Ghose}},\ }\href@noop {}
  {\bibfield  {journal} {\bibinfo  {journal} {Physical Review E}\ }\textbf
  {\bibinfo {volume} {91}},\ \bibinfo {pages} {032906} (\bibinfo {year}
  {2015}{\natexlab{a}})}\BibitemShut {NoStop}%
\bibitem [{\citenamefont {Madhok}\ \emph {et~al.}(2014)\citenamefont {Madhok},
  \citenamefont {Riofr{\'\i}o}, \citenamefont {Ghose},\ and\ \citenamefont
  {Deutsch}}]{madhok2014information}%
  \BibitemOpen
  \bibfield  {author} {\bibinfo {author} {\bibfnamefont {V.}~\bibnamefont
  {Madhok}}, \bibinfo {author} {\bibfnamefont {C.~A.}\ \bibnamefont
  {Riofr{\'\i}o}}, \bibinfo {author} {\bibfnamefont {S.}~\bibnamefont {Ghose}},
  \ and\ \bibinfo {author} {\bibfnamefont {I.~H.}\ \bibnamefont {Deutsch}},\
  }\href@noop {} {\bibfield  {journal} {\bibinfo  {journal} {Physical Review
  Letters}\ }\textbf {\bibinfo {volume} {112}},\ \bibinfo {pages} {014102}
  (\bibinfo {year} {2014})}\BibitemShut {NoStop}%
\bibitem [{\citenamefont {Braunstein}\ \emph {et~al.}(2013)\citenamefont
  {Braunstein}, \citenamefont {Pirandola},\ and\ \citenamefont
  {{\.Z}yczkowski}}]{braunstein2013better}%
  \BibitemOpen
  \bibfield  {author} {\bibinfo {author} {\bibfnamefont {S.~L.}\ \bibnamefont
  {Braunstein}}, \bibinfo {author} {\bibfnamefont {S.}~\bibnamefont
  {Pirandola}}, \ and\ \bibinfo {author} {\bibfnamefont {K.}~\bibnamefont
  {{\.Z}yczkowski}},\ }\href@noop {} {\bibfield  {journal} {\bibinfo  {journal}
  {Physical Review Letters}\ }\textbf {\bibinfo {volume} {110}},\ \bibinfo
  {pages} {101301} (\bibinfo {year} {2013})}\BibitemShut {NoStop}%
\bibitem [{\citenamefont {Hartman}\ and\ \citenamefont
  {Maldacena}(2013)}]{hartman2013time}%
  \BibitemOpen
  \bibfield  {author} {\bibinfo {author} {\bibfnamefont {T.}~\bibnamefont
  {Hartman}}\ and\ \bibinfo {author} {\bibfnamefont {J.}~\bibnamefont
  {Maldacena}},\ }\href@noop {} {\bibfield  {journal} {\bibinfo  {journal}
  {Journal of High Energy Physics}\ }\textbf {\bibinfo {volume} {2013}},\
  \bibinfo {pages} {14} (\bibinfo {year} {2013})}\BibitemShut {NoStop}%
\bibitem [{\citenamefont {Shenker}\ and\ \citenamefont
  {Stanford}(2014)}]{shenker2014black}%
  \BibitemOpen
  \bibfield  {author} {\bibinfo {author} {\bibfnamefont {S.~H.}\ \bibnamefont
  {Shenker}}\ and\ \bibinfo {author} {\bibfnamefont {D.}~\bibnamefont
  {Stanford}},\ }\href@noop {} {\bibfield  {journal} {\bibinfo  {journal}
  {Journal of High Energy Physics}\ }\textbf {\bibinfo {volume} {2014}},\
  \bibinfo {pages} {67} (\bibinfo {year} {2014})}\BibitemShut {NoStop}%
\bibitem [{\citenamefont {Sekino}\ and\ \citenamefont
  {Susskind}(2008)}]{sekino2008fast}%
  \BibitemOpen
  \bibfield  {author} {\bibinfo {author} {\bibfnamefont {Y.}~\bibnamefont
  {Sekino}}\ and\ \bibinfo {author} {\bibfnamefont {L.}~\bibnamefont
  {Susskind}},\ }\href@noop {} {\bibfield  {journal} {\bibinfo  {journal}
  {Journal of High Energy Physics}\ }\textbf {\bibinfo {volume} {2008}},\
  \bibinfo {pages} {065} (\bibinfo {year} {2008})}\BibitemShut {NoStop}%
\bibitem [{\citenamefont {Lashkari}\ \emph {et~al.}(2013)\citenamefont
  {Lashkari}, \citenamefont {Stanford}, \citenamefont {Hastings}, \citenamefont
  {Osborne},\ and\ \citenamefont {Hayden}}]{Lashkari2013}%
  \BibitemOpen
  \bibfield  {author} {\bibinfo {author} {\bibfnamefont {N.}~\bibnamefont
  {Lashkari}}, \bibinfo {author} {\bibfnamefont {D.}~\bibnamefont {Stanford}},
  \bibinfo {author} {\bibfnamefont {M.}~\bibnamefont {Hastings}}, \bibinfo
  {author} {\bibfnamefont {T.}~\bibnamefont {Osborne}}, \ and\ \bibinfo
  {author} {\bibfnamefont {P.}~\bibnamefont {Hayden}},\ }\href {\doibase
  10.1007/JHEP04(2013)022} {\bibfield  {journal} {\bibinfo  {journal} {Journal
  of High Energy Physics}\ }\textbf {\bibinfo {volume} {2013}},\ \bibinfo
  {pages} {22} (\bibinfo {year} {2013})}\BibitemShut {NoStop}%
\bibitem [{\citenamefont {Maldacena}\ \emph {et~al.}(2016)\citenamefont
  {Maldacena}, \citenamefont {Shenker},\ and\ \citenamefont
  {Stanford}}]{Maldacena2016}%
  \BibitemOpen
  \bibfield  {author} {\bibinfo {author} {\bibfnamefont {J.}~\bibnamefont
  {Maldacena}}, \bibinfo {author} {\bibfnamefont {S.~H.}\ \bibnamefont
  {Shenker}}, \ and\ \bibinfo {author} {\bibfnamefont {D.}~\bibnamefont
  {Stanford}},\ }\href {\doibase 10.1007/JHEP08(2016)106} {\bibfield  {journal}
  {\bibinfo  {journal} {Journal of High Energy Physics}\ }\textbf {\bibinfo
  {volume} {2016}},\ \bibinfo {pages} {106} (\bibinfo {year}
  {2016})}\BibitemShut {NoStop}%
\bibitem [{\citenamefont {Hosur}\ \emph {et~al.}(2016)\citenamefont {Hosur},
  \citenamefont {Qi}, \citenamefont {Roberts},\ and\ \citenamefont
  {Yoshida}}]{hosur2016chaos}%
  \BibitemOpen
  \bibfield  {author} {\bibinfo {author} {\bibfnamefont {P.}~\bibnamefont
  {Hosur}}, \bibinfo {author} {\bibfnamefont {X.-L.}\ \bibnamefont {Qi}},
  \bibinfo {author} {\bibfnamefont {D.~A.}\ \bibnamefont {Roberts}}, \ and\
  \bibinfo {author} {\bibfnamefont {B.}~\bibnamefont {Yoshida}},\ }\href@noop
  {} {\bibfield  {journal} {\bibinfo  {journal} {Journal of High Energy
  Physics}\ }\textbf {\bibinfo {volume} {2016}},\ \bibinfo {pages} {4}
  (\bibinfo {year} {2016})}\BibitemShut {NoStop}%
\bibitem [{\citenamefont {Iyoda}\ and\ \citenamefont
  {Sagawa}(2017)}]{IyodaSagawa}%
  \BibitemOpen
  \bibfield  {author} {\bibinfo {author} {\bibfnamefont {E.}~\bibnamefont
  {Iyoda}}\ and\ \bibinfo {author} {\bibfnamefont {T.}~\bibnamefont {Sagawa}},\
  }\href@noop {} {\bibfield  {journal} {\bibinfo  {journal} {arXiv:1704.04850}\
  } (\bibinfo {year} {2017})}\BibitemShut {NoStop}%
\bibitem [{\citenamefont {Ku{\'s}}\ \emph {et~al.}(1987)\citenamefont
  {Ku{\'s}}, \citenamefont {Scharf},\ and\ \citenamefont
  {Haake}}]{kus1987symmetry}%
  \BibitemOpen
  \bibfield  {author} {\bibinfo {author} {\bibfnamefont {M.}~\bibnamefont
  {Ku{\'s}}}, \bibinfo {author} {\bibfnamefont {R.}~\bibnamefont {Scharf}}, \
  and\ \bibinfo {author} {\bibfnamefont {F.}~\bibnamefont {Haake}},\
  }\href@noop {} {\bibfield  {journal} {\bibinfo  {journal} {Zeitschrift
  f{\"u}r Physik B Condensed Matter}\ }\textbf {\bibinfo {volume} {66}},\
  \bibinfo {pages} {129} (\bibinfo {year} {1987})}\BibitemShut {NoStop}%
\bibitem [{\citenamefont {Kus}\ \emph {et~al.}(1988)\citenamefont {Kus},
  \citenamefont {Mostowski},\ and\ \citenamefont {Haake}}]{kus88}%
  \BibitemOpen
  \bibfield  {author} {\bibinfo {author} {\bibfnamefont {M.}~\bibnamefont
  {Kus}}, \bibinfo {author} {\bibfnamefont {J.}~\bibnamefont {Mostowski}}, \
  and\ \bibinfo {author} {\bibfnamefont {F.}~\bibnamefont {Haake}},\
  }\href@noop {} {\bibfield  {journal} {\bibinfo  {journal} {Journal of Physics
  A: Mathematical and General}\ }\textbf {\bibinfo {volume} {21}},\ \bibinfo
  {pages} {L1073} (\bibinfo {year} {1988})}\BibitemShut {NoStop}%
\bibitem [{\citenamefont {Zyczkowski}(1990)}]{Zyczkowski90}%
  \BibitemOpen
  \bibfield  {author} {\bibinfo {author} {\bibfnamefont {K.}~\bibnamefont
  {Zyczkowski}},\ }\href@noop {} {\bibfield  {journal} {\bibinfo  {journal}
  {Journal of Physics A: Mathematical and General}\ }\textbf {\bibinfo {volume}
  {23}},\ \bibinfo {pages} {4427} (\bibinfo {year} {1990})}\BibitemShut
  {NoStop}%
\bibitem [{\citenamefont {Emerson}\ \emph {et~al.}(2003)\citenamefont
  {Emerson}, \citenamefont {Weinstein}, \citenamefont {Saraceno}, \citenamefont
  {Lloyd},\ and\ \citenamefont {Cory}}]{emerson2003pseudo}%
  \BibitemOpen
  \bibfield  {author} {\bibinfo {author} {\bibfnamefont {J.}~\bibnamefont
  {Emerson}}, \bibinfo {author} {\bibfnamefont {Y.~S.}\ \bibnamefont
  {Weinstein}}, \bibinfo {author} {\bibfnamefont {M.}~\bibnamefont {Saraceno}},
  \bibinfo {author} {\bibfnamefont {S.}~\bibnamefont {Lloyd}}, \ and\ \bibinfo
  {author} {\bibfnamefont {D.~G.}\ \bibnamefont {Cory}},\ }\href@noop {}
  {\bibfield  {journal} {\bibinfo  {journal} {science}\ }\textbf {\bibinfo
  {volume} {302}},\ \bibinfo {pages} {2098} (\bibinfo {year}
  {2003})}\BibitemShut {NoStop}%
\bibitem [{\citenamefont {Vidal}\ \emph {et~al.}(2003)\citenamefont {Vidal},
  \citenamefont {Latorre}, \citenamefont {Rico},\ and\ \citenamefont
  {Kitaev}}]{Vidal2003}%
  \BibitemOpen
  \bibfield  {author} {\bibinfo {author} {\bibfnamefont {G.}~\bibnamefont
  {Vidal}}, \bibinfo {author} {\bibfnamefont {J.~I.}\ \bibnamefont {Latorre}},
  \bibinfo {author} {\bibfnamefont {E.}~\bibnamefont {Rico}}, \ and\ \bibinfo
  {author} {\bibfnamefont {A.}~\bibnamefont {Kitaev}},\ }\href {\doibase
  10.1103/PhysRevLett.90.227902} {\bibfield  {journal} {\bibinfo  {journal}
  {Physical Review Letters}\ }\textbf {\bibinfo {volume} {90}},\ \bibinfo
  {pages} {227902} (\bibinfo {year} {2003})}\BibitemShut {NoStop}%
\bibitem [{\citenamefont {Latorre}\ \emph {et~al.}(2005)\citenamefont
  {Latorre}, \citenamefont {Or{\'u}s}, \citenamefont {Rico},\ and\
  \citenamefont {Vidal}}]{latorre2005entanglement}%
  \BibitemOpen
  \bibfield  {author} {\bibinfo {author} {\bibfnamefont {J.~I.}\ \bibnamefont
  {Latorre}}, \bibinfo {author} {\bibfnamefont {R.}~\bibnamefont {Or{\'u}s}},
  \bibinfo {author} {\bibfnamefont {E.}~\bibnamefont {Rico}}, \ and\ \bibinfo
  {author} {\bibfnamefont {J.}~\bibnamefont {Vidal}},\ }\href@noop {}
  {\bibfield  {journal} {\bibinfo  {journal} {Physical Review A}\ }\textbf
  {\bibinfo {volume} {71}},\ \bibinfo {pages} {064101} (\bibinfo {year}
  {2005})}\BibitemShut {NoStop}%
\bibitem [{\citenamefont {Popkov}\ \emph {et~al.}(2005)\citenamefont {Popkov},
  \citenamefont {Salerno},\ and\ \citenamefont
  {Sch{\"u}tz}}]{popkov2005entangling}%
  \BibitemOpen
  \bibfield  {author} {\bibinfo {author} {\bibfnamefont {V.}~\bibnamefont
  {Popkov}}, \bibinfo {author} {\bibfnamefont {M.}~\bibnamefont {Salerno}}, \
  and\ \bibinfo {author} {\bibfnamefont {G.}~\bibnamefont {Sch{\"u}tz}},\
  }\href@noop {} {\bibfield  {journal} {\bibinfo  {journal} {Physical Review
  A}\ }\textbf {\bibinfo {volume} {72}},\ \bibinfo {pages} {032327} (\bibinfo
  {year} {2005})}\BibitemShut {NoStop}%
\bibitem [{\citenamefont {Popkov}\ and\ \citenamefont
  {Salerno}(2005)}]{popkov2005logarithmic}%
  \BibitemOpen
  \bibfield  {author} {\bibinfo {author} {\bibfnamefont {V.}~\bibnamefont
  {Popkov}}\ and\ \bibinfo {author} {\bibfnamefont {M.}~\bibnamefont
  {Salerno}},\ }\href@noop {} {\bibfield  {journal} {\bibinfo  {journal}
  {Physical Review A}\ }\textbf {\bibinfo {volume} {71}},\ \bibinfo {pages}
  {012301} (\bibinfo {year} {2005})}\BibitemShut {NoStop}%
\bibitem [{\citenamefont {Barthel}\ \emph {et~al.}(2006)\citenamefont
  {Barthel}, \citenamefont {Dusuel},\ and\ \citenamefont
  {Vidal}}]{barthel2006entanglement}%
  \BibitemOpen
  \bibfield  {author} {\bibinfo {author} {\bibfnamefont {T.}~\bibnamefont
  {Barthel}}, \bibinfo {author} {\bibfnamefont {S.}~\bibnamefont {Dusuel}}, \
  and\ \bibinfo {author} {\bibfnamefont {J.}~\bibnamefont {Vidal}},\
  }\href@noop {} {\bibfield  {journal} {\bibinfo  {journal} {Physical Review
  Letters}\ }\textbf {\bibinfo {volume} {97}},\ \bibinfo {pages} {220402}
  (\bibinfo {year} {2006})}\BibitemShut {NoStop}%
\bibitem [{\citenamefont {Vidal}\ \emph {et~al.}(2007)\citenamefont {Vidal},
  \citenamefont {Dusuel},\ and\ \citenamefont
  {Barthel}}]{vidal2007entanglement}%
  \BibitemOpen
  \bibfield  {author} {\bibinfo {author} {\bibfnamefont {J.}~\bibnamefont
  {Vidal}}, \bibinfo {author} {\bibfnamefont {S.}~\bibnamefont {Dusuel}}, \
  and\ \bibinfo {author} {\bibfnamefont {T.}~\bibnamefont {Barthel}},\
  }\href@noop {} {\bibfield  {journal} {\bibinfo  {journal} {Journal of
  Statistical Mechanics: Theory and Experiment}\ }\textbf {\bibinfo {volume}
  {2007}},\ \bibinfo {pages} {P01015} (\bibinfo {year} {2007})}\BibitemShut
  {NoStop}%
\bibitem [{\citenamefont {Mar{\v{c}}enko}\ and\ \citenamefont
  {Pastur}(1967)}]{marvcenko1967distribution}%
  \BibitemOpen
  \bibfield  {author} {\bibinfo {author} {\bibfnamefont {V.~A.}\ \bibnamefont
  {Mar{\v{c}}enko}}\ and\ \bibinfo {author} {\bibfnamefont {L.~A.}\
  \bibnamefont {Pastur}},\ }\href@noop {} {\bibfield  {journal} {\bibinfo
  {journal} {Mathematics of the USSR-Sbornik}\ }\textbf {\bibinfo {volume}
  {1}},\ \bibinfo {pages} {457} (\bibinfo {year} {1967})}\BibitemShut {NoStop}%
\bibitem [{\citenamefont {Madhok}\ \emph {et~al.}(2018)\citenamefont {Madhok},
  \citenamefont {Dogra},\ and\ \citenamefont
  {Lakshminarayan}}]{Madhok2018_corr}%
  \BibitemOpen
  \bibfield  {author} {\bibinfo {author} {\bibfnamefont {V.}~\bibnamefont
  {Madhok}}, \bibinfo {author} {\bibfnamefont {S.}~\bibnamefont {Dogra}}, \
  and\ \bibinfo {author} {\bibfnamefont {A.}~\bibnamefont {Lakshminarayan}},\
  }\href {\doibase https://doi.org/10.1016/j.optcom.2018.03.069} {\bibfield
  {journal} {\bibinfo  {journal} {Optics Communications}\ }\textbf {\bibinfo
  {volume} {420}},\ \bibinfo {pages} {189 } (\bibinfo {year}
  {2018})}\BibitemShut {NoStop}%
\bibitem [{\citenamefont {Rangamani}\ and\ \citenamefont
  {Rota}(2015)}]{rangamani2015entanglement}%
  \BibitemOpen
  \bibfield  {author} {\bibinfo {author} {\bibfnamefont {M.}~\bibnamefont
  {Rangamani}}\ and\ \bibinfo {author} {\bibfnamefont {M.}~\bibnamefont
  {Rota}},\ }\href@noop {} {\bibfield  {journal} {\bibinfo  {journal} {Journal
  of Physics A: Mathematical and Theoretical}\ }\textbf {\bibinfo {volume}
  {48}},\ \bibinfo {pages} {385301} (\bibinfo {year} {2015})}\BibitemShut
  {NoStop}%
\bibitem [{\citenamefont {Rota}(2016)}]{rota2016tripartite}%
  \BibitemOpen
  \bibfield  {author} {\bibinfo {author} {\bibfnamefont {M.}~\bibnamefont
  {Rota}},\ }\href@noop {} {\bibfield  {journal} {\bibinfo  {journal} {Journal
  of High Energy Physics}\ }\textbf {\bibinfo {volume} {2016}},\ \bibinfo
  {pages} {75} (\bibinfo {year} {2016})}\BibitemShut {NoStop}%
\bibitem [{\citenamefont {Moreno}\ and\ \citenamefont
  {Parisio}(2018)}]{moreno2018all}%
  \BibitemOpen
  \bibfield  {author} {\bibinfo {author} {\bibfnamefont {M.}~\bibnamefont
  {Moreno}}\ and\ \bibinfo {author} {\bibfnamefont {F.}~\bibnamefont
  {Parisio}},\ }\href@noop {} {\bibfield  {journal} {\bibinfo  {journal} {arXiv
  preprint arXiv:1801.00762}\ } (\bibinfo {year} {2018})}\BibitemShut {NoStop}%
\bibitem [{\citenamefont {Popkov}\ and\ \citenamefont
  {Salerno}(2012)}]{popkov2012reduced}%
  \BibitemOpen
  \bibfield  {author} {\bibinfo {author} {\bibfnamefont {V.}~\bibnamefont
  {Popkov}}\ and\ \bibinfo {author} {\bibfnamefont {M.}~\bibnamefont
  {Salerno}},\ }\href@noop {} {\bibfield  {journal} {\bibinfo  {journal}
  {International Journal of Modern Physics B}\ }\textbf {\bibinfo {volume}
  {26}},\ \bibinfo {pages} {1243009} (\bibinfo {year} {2012})}\BibitemShut
  {NoStop}%
\bibitem [{\citenamefont {Markham}(2011)}]{markham2011entanglement}%
  \BibitemOpen
  \bibfield  {author} {\bibinfo {author} {\bibfnamefont {D.~J.}\ \bibnamefont
  {Markham}},\ }\href@noop {} {\bibfield  {journal} {\bibinfo  {journal}
  {Physical Review A}\ }\textbf {\bibinfo {volume} {83}},\ \bibinfo {pages}
  {042332} (\bibinfo {year} {2011})}\BibitemShut {NoStop}%
\bibitem [{\citenamefont {Devi}\ \emph {et~al.}(2012)\citenamefont {Devi},
  \citenamefont {Rajagopal} \emph {et~al.}}]{devi2012majorana}%
  \BibitemOpen
  \bibfield  {author} {\bibinfo {author} {\bibfnamefont {A.~U.}\ \bibnamefont
  {Devi}}, \bibinfo {author} {\bibfnamefont {A.}~\bibnamefont {Rajagopal}},
  \emph {et~al.},\ }\href@noop {} {\bibfield  {journal} {\bibinfo  {journal}
  {Quantum Information Processing}\ }\textbf {\bibinfo {volume} {11}},\
  \bibinfo {pages} {685} (\bibinfo {year} {2012})}\BibitemShut {NoStop}%
\bibitem [{\citenamefont {Bohnet-Waldraff}\ \emph {et~al.}(2016)\citenamefont
  {Bohnet-Waldraff}, \citenamefont {Braun},\ and\ \citenamefont
  {Giraud}}]{bohnet2016partial}%
  \BibitemOpen
  \bibfield  {author} {\bibinfo {author} {\bibfnamefont {F.}~\bibnamefont
  {Bohnet-Waldraff}}, \bibinfo {author} {\bibfnamefont {D.}~\bibnamefont
  {Braun}}, \ and\ \bibinfo {author} {\bibfnamefont {O.}~\bibnamefont
  {Giraud}},\ }\href@noop {} {\bibfield  {journal} {\bibinfo  {journal}
  {Physical Review A}\ }\textbf {\bibinfo {volume} {94}},\ \bibinfo {pages}
  {042343} (\bibinfo {year} {2016})}\BibitemShut {NoStop}%
\bibitem [{\citenamefont {Wang}\ and\ \citenamefont
  {M{\o}lmer}(2002)}]{wang2002pairwise}%
  \BibitemOpen
  \bibfield  {author} {\bibinfo {author} {\bibfnamefont {X.}~\bibnamefont
  {Wang}}\ and\ \bibinfo {author} {\bibfnamefont {K.}~\bibnamefont
  {M{\o}lmer}},\ }\href@noop {} {\bibfield  {journal} {\bibinfo  {journal} {The
  European Physical Journal D-Atomic, Molecular, Optical and Plasma Physics}\
  }\textbf {\bibinfo {volume} {18}},\ \bibinfo {pages} {385} (\bibinfo {year}
  {2002})}\BibitemShut {NoStop}%
\bibitem [{\citenamefont {Baguette}\ \emph {et~al.}(2014)\citenamefont
  {Baguette}, \citenamefont {Bastin},\ and\ \citenamefont
  {Martin}}]{baguette2014multiqubit}%
  \BibitemOpen
  \bibfield  {author} {\bibinfo {author} {\bibfnamefont {D.}~\bibnamefont
  {Baguette}}, \bibinfo {author} {\bibfnamefont {T.}~\bibnamefont {Bastin}}, \
  and\ \bibinfo {author} {\bibfnamefont {J.}~\bibnamefont {Martin}},\
  }\href@noop {} {\bibfield  {journal} {\bibinfo  {journal} {Physical Review
  A}\ }\textbf {\bibinfo {volume} {90}},\ \bibinfo {pages} {032314} (\bibinfo
  {year} {2014})}\BibitemShut {NoStop}%
\bibitem [{\citenamefont {Dicke}(1954)}]{Dicke}%
  \BibitemOpen
  \bibfield  {author} {\bibinfo {author} {\bibfnamefont {R.~H.}\ \bibnamefont
  {Dicke}},\ }\href {\doibase 10.1103/PhysRev.93.99} {\bibfield  {journal}
  {\bibinfo  {journal} {Physical Review}\ }\textbf {\bibinfo {volume} {93}},\
  \bibinfo {pages} {99} (\bibinfo {year} {1954})}\BibitemShut {NoStop}%
\bibitem [{\citenamefont {Lakshminarayan}\ \emph {et~al.}(2008)\citenamefont
  {Lakshminarayan}, \citenamefont {Tomsovic}, \citenamefont {Bohigas},\ and\
  \citenamefont {Majumdar}}]{lakshminarayan2008}%
  \BibitemOpen
  \bibfield  {author} {\bibinfo {author} {\bibfnamefont {A.}~\bibnamefont
  {Lakshminarayan}}, \bibinfo {author} {\bibfnamefont {S.}~\bibnamefont
  {Tomsovic}}, \bibinfo {author} {\bibfnamefont {O.}~\bibnamefont {Bohigas}}, \
  and\ \bibinfo {author} {\bibfnamefont {S.~N.}\ \bibnamefont {Majumdar}},\
  }\href {\doibase 10.1103/PhysRevLett.100.044103} {\bibfield  {journal}
  {\bibinfo  {journal} {Physical Review Letters}\ }\textbf {\bibinfo {volume}
  {100}},\ \bibinfo {pages} {044103} (\bibinfo {year} {2008})}\BibitemShut
  {NoStop}%
\bibitem [{\citenamefont {Wishart}(1928)}]{wishart1928generalised}%
  \BibitemOpen
  \bibfield  {author} {\bibinfo {author} {\bibfnamefont {J.}~\bibnamefont
  {Wishart}},\ }\href@noop {} {\bibfield  {journal} {\bibinfo  {journal}
  {Biometrika}\ ,\ \bibinfo {pages} {32}} (\bibinfo {year} {1928})}\BibitemShut
  {NoStop}%
\bibitem [{\citenamefont {Bryc}\ \emph {et~al.}(2006)\citenamefont {Bryc},
  \citenamefont {Dembo},\ and\ \citenamefont {Jiang}}]{bryc2006}%
  \BibitemOpen
  \bibfield  {author} {\bibinfo {author} {\bibfnamefont {W.}~\bibnamefont
  {Bryc}}, \bibinfo {author} {\bibfnamefont {A.}~\bibnamefont {Dembo}}, \ and\
  \bibinfo {author} {\bibfnamefont {T.}~\bibnamefont {Jiang}},\ }\href
  {\doibase 10.1214/009117905000000495} {\bibfield  {journal} {\bibinfo
  {journal} {Ann. Probab.}\ }\textbf {\bibinfo {volume} {34}},\ \bibinfo
  {pages} {1} (\bibinfo {year} {2006})}\BibitemShut {NoStop}%
\bibitem [{\citenamefont {Lubkin}(1978)}]{Lubkin}%
  \BibitemOpen
  \bibfield  {author} {\bibinfo {author} {\bibfnamefont {E.}~\bibnamefont
  {Lubkin}},\ }\href@noop {} {\bibfield  {journal} {\bibinfo  {journal}
  {Journal of Mathematical Physics}\ }\textbf {\bibinfo {volume} {19}}
  (\bibinfo {year} {1978})}\BibitemShut {NoStop}%
\bibitem [{\citenamefont {Page}(1993)}]{page1993average}%
  \BibitemOpen
  \bibfield  {author} {\bibinfo {author} {\bibfnamefont {D.~N.}\ \bibnamefont
  {Page}},\ }\href@noop {} {\bibfield  {journal} {\bibinfo  {journal} {Physical
  Review Letters}\ }\textbf {\bibinfo {volume} {71}},\ \bibinfo {pages} {1291}
  (\bibinfo {year} {1993})}\BibitemShut {NoStop}%
\bibitem [{\citenamefont {Sen}(1996)}]{sen1996average}%
  \BibitemOpen
  \bibfield  {author} {\bibinfo {author} {\bibfnamefont {S.}~\bibnamefont
  {Sen}},\ }\href@noop {} {\bibfield  {journal} {\bibinfo  {journal} {Physical
  Review Letters}\ }\textbf {\bibinfo {volume} {77}},\ \bibinfo {pages} {1}
  (\bibinfo {year} {1996})}\BibitemShut {NoStop}%
\bibitem [{\citenamefont {Foong}\ and\ \citenamefont
  {Kanno}(1994)}]{foong1994proof}%
  \BibitemOpen
  \bibfield  {author} {\bibinfo {author} {\bibfnamefont {S.}~\bibnamefont
  {Foong}}\ and\ \bibinfo {author} {\bibfnamefont {S.}~\bibnamefont {Kanno}},\
  }\href@noop {} {\bibfield  {journal} {\bibinfo  {journal} {Physical Review
  Letters}\ }\textbf {\bibinfo {volume} {72}},\ \bibinfo {pages} {1148}
  (\bibinfo {year} {1994})}\BibitemShut {NoStop}%
\bibitem [{\citenamefont {S{\'a}nchez-Ruiz}(1995)}]{sanchez1995simple}%
  \BibitemOpen
  \bibfield  {author} {\bibinfo {author} {\bibfnamefont {J.}~\bibnamefont
  {S{\'a}nchez-Ruiz}},\ }\href@noop {} {\bibfield  {journal} {\bibinfo
  {journal} {Physical Review E}\ }\textbf {\bibinfo {volume} {52}},\ \bibinfo
  {pages} {5653} (\bibinfo {year} {1995})}\BibitemShut {NoStop}%
\bibitem [{\citenamefont {Wilms}\ \emph {et~al.}(2012)\citenamefont {Wilms},
  \citenamefont {Vidal}, \citenamefont {Verstraete},\ and\ \citenamefont
  {Dusuel}}]{wilms2012finite}%
  \BibitemOpen
  \bibfield  {author} {\bibinfo {author} {\bibfnamefont {J.}~\bibnamefont
  {Wilms}}, \bibinfo {author} {\bibfnamefont {J.}~\bibnamefont {Vidal}},
  \bibinfo {author} {\bibfnamefont {F.}~\bibnamefont {Verstraete}}, \ and\
  \bibinfo {author} {\bibfnamefont {S.}~\bibnamefont {Dusuel}},\ }\href@noop {}
  {\bibfield  {journal} {\bibinfo  {journal} {Journal of Statistical Mechanics:
  Theory and Experiment}\ }\textbf {\bibinfo {volume} {2012}},\ \bibinfo
  {pages} {P01023} (\bibinfo {year} {2012})}\BibitemShut {NoStop}%
\bibitem [{\citenamefont {Madhok}\ \emph
  {et~al.}(2015{\natexlab{b}})\citenamefont {Madhok}, \citenamefont {Gupta},
  \citenamefont {Trottier},\ and\ \citenamefont {Ghose}}]{mgtg15}%
  \BibitemOpen
  \bibfield  {author} {\bibinfo {author} {\bibfnamefont {V.}~\bibnamefont
  {Madhok}}, \bibinfo {author} {\bibfnamefont {V.}~\bibnamefont {Gupta}},
  \bibinfo {author} {\bibfnamefont {D.-A.}\ \bibnamefont {Trottier}}, \ and\
  \bibinfo {author} {\bibfnamefont {S.}~\bibnamefont {Ghose}},\ }\href@noop {}
  {\bibfield  {journal} {\bibinfo  {journal} {Physical Review E}\ }\textbf
  {\bibinfo {volume} {91}},\ \bibinfo {pages} {032906} (\bibinfo {year}
  {2015}{\natexlab{b}})}\BibitemShut {NoStop}%
\bibitem [{\citenamefont {Prosen}(2000)}]{prosen2000exact}%
  \BibitemOpen
  \bibfield  {author} {\bibinfo {author} {\bibfnamefont {T.}~\bibnamefont
  {Prosen}},\ }\href@noop {} {\bibfield  {journal} {\bibinfo  {journal}
  {Progress of Theoretical Physics Supplement}\ }\textbf {\bibinfo {volume}
  {139}},\ \bibinfo {pages} {191} (\bibinfo {year} {2000})}\BibitemShut
  {NoStop}%
\bibitem [{\citenamefont {Amico}\ \emph {et~al.}(2008)\citenamefont {Amico},
  \citenamefont {Fazio}, \citenamefont {Osterloh},\ and\ \citenamefont
  {Vedral}}]{VedralRMPmanybody}%
  \BibitemOpen
  \bibfield  {author} {\bibinfo {author} {\bibfnamefont {L.}~\bibnamefont
  {Amico}}, \bibinfo {author} {\bibfnamefont {R.}~\bibnamefont {Fazio}},
  \bibinfo {author} {\bibfnamefont {A.}~\bibnamefont {Osterloh}}, \ and\
  \bibinfo {author} {\bibfnamefont {V.}~\bibnamefont {Vedral}},\ }\href
  {\doibase 10.1103/RevModPhys.80.517} {\bibfield  {journal} {\bibinfo
  {journal} {Reviews of Modern Physics}\ }\textbf {\bibinfo {volume} {80}},\
  \bibinfo {pages} {517} (\bibinfo {year} {2008})}\BibitemShut {NoStop}%
\bibitem [{\citenamefont {Glauber}\ and\ \citenamefont
  {Haake}(1976)}]{Glauber}%
  \BibitemOpen
  \bibfield  {author} {\bibinfo {author} {\bibfnamefont {R.~J.}\ \bibnamefont
  {Glauber}}\ and\ \bibinfo {author} {\bibfnamefont {F.}~\bibnamefont
  {Haake}},\ }\href@noop {} {\bibfield  {journal} {\bibinfo  {journal}
  {Physical Review A}\ }\textbf {\bibinfo {volume} {13}},\ \bibinfo {pages}
  {357} (\bibinfo {year} {1976})}\BibitemShut {NoStop}%
\bibitem [{\citenamefont {Puri}(2001)}]{Puri}%
  \BibitemOpen
  \bibfield  {author} {\bibinfo {author} {\bibfnamefont {R.~R.}\ \bibnamefont
  {Puri}},\ }\href@noop {} {\emph {\bibinfo {title} {Mathematical Methods of
  Quantum Optics}}}\ (\bibinfo  {publisher} {Springer},\ \bibinfo {address}
  {Berlin},\ \bibinfo {year} {2001})\BibitemShut {NoStop}%
\bibitem [{\citenamefont {Hayden}\ and\ \citenamefont
  {Preskill}(2007)}]{hayden2007black}%
  \BibitemOpen
  \bibfield  {author} {\bibinfo {author} {\bibfnamefont {P.}~\bibnamefont
  {Hayden}}\ and\ \bibinfo {author} {\bibfnamefont {J.}~\bibnamefont
  {Preskill}},\ }\href@noop {} {\bibfield  {journal} {\bibinfo  {journal}
  {Journal of High Energy Physics}\ }\textbf {\bibinfo {volume} {2007}},\
  \bibinfo {pages} {120} (\bibinfo {year} {2007})}\BibitemShut {NoStop}%
\bibitem [{\citenamefont {Rozenbaum}\ \emph {et~al.}(2017)\citenamefont
  {Rozenbaum}, \citenamefont {Ganeshan},\ and\ \citenamefont
  {Galitski}}]{Rozenbaum17}%
  \BibitemOpen
  \bibfield  {author} {\bibinfo {author} {\bibfnamefont {E.~B.}\ \bibnamefont
  {Rozenbaum}}, \bibinfo {author} {\bibfnamefont {S.}~\bibnamefont {Ganeshan}},
  \ and\ \bibinfo {author} {\bibfnamefont {V.}~\bibnamefont {Galitski}},\
  }\href {\doibase 10.1103/PhysRevLett.118.086801} {\bibfield  {journal}
  {\bibinfo  {journal} {Physical Review Letters}\ }\textbf {\bibinfo {volume}
  {118}},\ \bibinfo {pages} {086801} (\bibinfo {year} {2017})}\BibitemShut
  {NoStop}%
\bibitem [{\citenamefont {Swingle}\ \emph {et~al.}(2016)\citenamefont
  {Swingle}, \citenamefont {Bentsen}, \citenamefont {Schleier-Smith},\ and\
  \citenamefont {Hayden}}]{swingle2016measuring}%
  \BibitemOpen
  \bibfield  {author} {\bibinfo {author} {\bibfnamefont {B.}~\bibnamefont
  {Swingle}}, \bibinfo {author} {\bibfnamefont {G.}~\bibnamefont {Bentsen}},
  \bibinfo {author} {\bibfnamefont {M.}~\bibnamefont {Schleier-Smith}}, \ and\
  \bibinfo {author} {\bibfnamefont {P.}~\bibnamefont {Hayden}},\ }\href@noop {}
  {\bibfield  {journal} {\bibinfo  {journal} {Physical Review A}\ }\textbf
  {\bibinfo {volume} {94}},\ \bibinfo {pages} {040302} (\bibinfo {year}
  {2016})}\BibitemShut {NoStop}%
\bibitem [{\citenamefont {Pappalardi}\ \emph {et~al.}(2018)\citenamefont
  {Pappalardi}, \citenamefont {Russomanno}, \citenamefont
  {{\v{Z}}unkovi{\v{c}}}, \citenamefont {Iemini}, \citenamefont {Silva},\ and\
  \citenamefont {Fazio}}]{pappalardi2018scrambling}%
  \BibitemOpen
  \bibfield  {author} {\bibinfo {author} {\bibfnamefont {S.}~\bibnamefont
  {Pappalardi}}, \bibinfo {author} {\bibfnamefont {A.}~\bibnamefont
  {Russomanno}}, \bibinfo {author} {\bibfnamefont {B.}~\bibnamefont
  {{\v{Z}}unkovi{\v{c}}}}, \bibinfo {author} {\bibfnamefont {F.}~\bibnamefont
  {Iemini}}, \bibinfo {author} {\bibfnamefont {A.}~\bibnamefont {Silva}}, \
  and\ \bibinfo {author} {\bibfnamefont {R.}~\bibnamefont {Fazio}},\
  }\href@noop {} {\bibfield  {journal} {\bibinfo  {journal} {arXiv preprint
  arXiv:1806.00022}\ } (\bibinfo {year} {2018})}\BibitemShut {NoStop}%
\bibitem [{\citenamefont {Hayden}\ \emph {et~al.}(2006)\citenamefont {Hayden},
  \citenamefont {Leung},\ and\ \citenamefont {Winter}}]{hayden2006aspects}%
  \BibitemOpen
  \bibfield  {author} {\bibinfo {author} {\bibfnamefont {P.}~\bibnamefont
  {Hayden}}, \bibinfo {author} {\bibfnamefont {D.~W.}\ \bibnamefont {Leung}}, \
  and\ \bibinfo {author} {\bibfnamefont {A.}~\bibnamefont {Winter}},\
  }\href@noop {} {\bibfield  {journal} {\bibinfo  {journal} {Communications in
  Mathematical Physics}\ }\textbf {\bibinfo {volume} {265}},\ \bibinfo {pages}
  {95} (\bibinfo {year} {2006})}\BibitemShut {NoStop}%
\end{thebibliography}%

\appendix
\section{\label{apdx:permsymmembedding}Subsystems of Permutation Symmetric Systems}
In section II, we saw that the structure of Dicke states can be used to write permutation symmetric vectors of $N$ qubits in terms of tensor product of permutation symmetric vectors containing $Q$ and $N - Q$ qubits ($Q < N$). Equation \ref{eq:permpsi2} holds, however, only when the vectors are viewed as elements of the full $2^N$ dimensional space. We can nevertheless use the idea presented in Eq.~(\ref{eq:permpsi2}) to work with permutation symmetric systems without viewing them as subspace of a larger space. An apparent tradeoff of working exclusively with permutation symmetric spaces is that the tensor product structure would be lost, as an $N + 1$ dimensional space in general cannot be written as a tensor product of $Q + 1$ and $N - Q + 1$ dimensional spaces. Fortunately though, it is still possible to work conveniently with subsystems. We show in the following that we can embed a permutation symmetric vector of $N$ qubits into the tensor product of $Q$ qubit and $N - Q$ qubit permutation symmetric spaces such that inner products are preserved.

\begin{proposition}
    \label{propn:permsymmembedding}
    For $d_X, d_Y, d_Z \in \mathbb{N}$, let $\{\ket{i_X}\ |\ 0 \leq i \leq d_X\}$ be some fixed orthonormal basis of $X$, $\{\ket{j_Y}\ |\ 0 \leq j \leq d_Y\}$ be a fixed orthonormal basis of $Y$ and $\{\ket{k_Z}\ |\ 0 \leq k \leq d_Z\}$ be a fixed orthonormal basis of $Z$, such that $d_X = d_Y + d_Z$. Then the map
    \begin{align*}
        &\iota \colon X \to Y \otimes Z \\
        &\iota\left(\sum_{i = 0}^{d_X} a_i \ket{i_X}\right) = \sum_{j = 0}^{d_Y} \sum_{k = 0}^{d_Z} A_{jk} \ket{j_Y} \ket{k_Z} \\
        &\text{with } A_{jk} = C_{jk} a_{j + k} = \sqrt{\frac{\binom{d_Y}{j} \binom{d_Z}{k}}{\binom{d_X}{j + k}}} a_{j + k}
    \end{align*}
    is linear and injective. Further, the map is isometric \textnormal{(}i.e., preserves inner products\textnormal{)}.
\end{proposition}
\begin{proof}
    That the map is linear is clear. To see that it is injective, note that if $\iota(\ket{a_X}) = 0$ for some $\ket{a_X} \in X$, then each of the coefficients $A_{jk}$ are zero. Thus $a_i = 0$ implying $\ket{a_X} = 0$ (i.e., the kernel of $\iota$ is trivial), establishing the injectivity of $\iota$. Showing isometry of $\iota$ involves a simple computation using Vandermonde Convolution identity.

    Towards this end, say $\ket{a_X} = \sum_i a_i \ket{i_X}$ and $\ket{b_X} = \sum_i b_i \ket{i_X}$, so that $\ip{\ket{a_X}, \ket{b_X}} = \sum_i a_i^* b_i$. On the other hand, $\ip{\iota(\ket{a_X}), \iota(\ket{b_X})} = \sum_j \sum_k A_{jk}^* B_{jk} = \sum_j \sum_k C_{jk}^2 a_{j + k}^* b_{j + k}$. Now the sum over $0 \leq j \leq d_Y$ and $0 \leq k \leq d_Z$ is done along the ``cross-diagonals" $j + k = r$. That is fix an $r$ ($0 \leq r \leq d_X = d_Y + d_Z$) and sum over $\max(0, r - d_Z) \leq j \leq \min(r, d_Y)$, and do this for all $0 \leq r \leq d_X$.
    \begin{align*}
        \ip{\iota(\ket{a_X}), \iota(\ket{b_X})} &= \sum_{r = 0}^{d_X} \frac{a_r^* b_r}{\binom{d_X}{r}}
                                                        \sum_{j = \max(0, r - d_Z)}^{\min(r, d_Y)} \binom{d_Y}{j} \binom{d_Z}{r - j} \\
                                                &= \sum_{r = 0}^{d_Z} \frac{a_r^* b_r}{\binom{d_X}{r}}
                                                        \sum_{j = 0}^{\min(r, d_Y)} \binom{d_Y}{j} \binom{d_Z}{r - j} \\
                                                   &+ \sum_{r = d_Z + 1}^{d_X} \frac{a_r^* b_r}{\binom{d_X}{r}}
                                                            \sum_{j = 0}^{d_X - r} \binom{d_Y}{j} \binom{d_Z}{d_X - r - j} \\
                                                &= \sum_{r = 0}^{d_Z} \frac{a_r^* b_r}{\binom{d_X}{r}} \binom{d_X}{r} +
                                                   \sum_{r = d_Z + 1}^{d_X} \frac{a_r^* b_r}{\binom{d_X}{r}} \binom{d_X}{r} \\
                                                   &\hspace{0.5cm}\text{ (using Vandermonde Convolution) } \\
                                                &= \sum_{r = 0}^{d_X} a_r^* b_r \\
                                                &= \ip{\ket{a_X}, \ket{b_X}}
    \end{align*}
    The above assumes that $d_Y \leq d_Z$, if not, sum over $k$ instead, following the same steps. Thus, $\iota$ is an isometry.
\end{proof}

Hence, given an $N$ qubit permutation symmetric state, we can work as if it has a tensor product decomposition in terms of permutation symmetric states of $Q$ and $N - Q$ qubits.

Furthermore, if we were to embed these vectors in the full $2^N$ dimensional space (by mapping each basis vector to the appropriate Dicke state), then these are actually the same vector. To elaborate on this point, say given a vector space $V$ of dimension $d_V + 1$, we define the linear map $\mathcal{E}_V$ as one that takes some fixed basis vectors of the space $V$ to the corresponding Dicke states in the $2^{d_X}$ dimensional space ($d_V \leq d_X$).
\begin{equation}
    \mathcal{E}_V(\ket{i_V}) = \frac{1}{\sqrt{\binom{d_V}{i}}} \sum_{\substack{0 \leq l \leq 2^{d_X} - 1\\ w(l) = i}} \ket{l},\
                                                                                                                            0 \leq i \leq d_V
\end{equation}
The action on the fixed basis vectors $\ket{i_V}$ can be extended linearly to define $\mathcal{E}_V$ on every vector of $V$ (and such an extension is unique). With such a linear map, we can see that the vectors related by the embedding $\iota$ are actually equal in the full $2^{d_X}$ dimensional space. In other words, we have $\mathcal{E}_X = (\mathcal{E}_Y \otimes \mathcal{E}_Z) \circ \iota$, which is essentially the content of Eq.~(\ref{eq:permpsi2}) mentioned in section II.

Now, since we always work in dimensions linear in the number of qubits, one can study permutation symmetric systems with very large number of qubits. In such a scenario, it may be helpful to obtain the combinatorial coefficients appearing in the expression for reduced density matrix recursively. So if $C_{kl}$ are the combinatorial coefficients as defined in proposition \ref{propn:permsymmembedding}, we have the following recursion relations for the same.
\begin{align}
    C_{k+1,l}   &= \sqrt{\left(\frac{Q - k}{N - k - l}\right)} \sqrt{\left(\frac{k + l + 1}{k + 1}\right)} C_{kl} \\
    C_{k,l+1}   &= \sqrt{\left(\frac{N - Q - l}{N - k - l}\right)} \sqrt{\left(\frac{k + l + 1}{l + 1}\right)} C_{kl} \\
    C_{k+1,l+1} &= \sqrt{\left(\frac{Q - k}{N - k - l - 1}\right)} \sqrt{\left(\frac{k + l + 2}{k + 1}\right)} \\
                    &\hspace{0.8cm}\sqrt{\left(\frac{N - Q - l}{N - k - l}\right)} \sqrt{\left(\frac{k + l + 1}{l + 1}\right)} C_{kl}
\end{align}
Here, in place of $d_X$ we have used $N$ and $Q$ in place of $d_Y$, which denote the number of qubits in the system and the subsystem respectively. These relations can be used to compute $C_{kl}$ when working with a large number of qubits.

\section{\label{apdx:permsymmlevy}L\'evy's Lemma for Permutation Symmetric Systems}
L\'evy's lemma is a statement describing the relation between the values taken by a Lipschitz continuous function defined on a sphere and its average with respect to the uniform measure on the sphere. The normalization requirement of quantum states allows us to consider the state as a point on a unit sphere. The work in \cite{hayden2006aspects} describes the use of L\'evy's lemma in the context of quantum systems. In particular, their study covers the cases of von Neumann and linear entropy for a system of qubits. But if we directly apply their results to a PS system of qubits, the average would correspond to the Wishart ensemble, while we are interested in averages over the PS ensemble. Thus we need to appropriately adapt the relevant proofs presented in Hayden \textit{et al.} \cite{hayden2006aspects} to cover the case of PS system of qubits, so that the averages indeed correspond to those over the PS ensemble.

Our starting point is the property noted in Appendix \ref{apdx:permsymmembedding}, that given any $N$ qubit permutation symmetric vector, it has a tensor product decomposition in terms of $Q$ qubit and $N - Q$ qubit permutation symmetric vectors. So, given any permutation symmetric vector $\ket{a_X} \in X$ of dimension $d_X + 1$ and subsystems $Y$ and $Z$ of dimensions $d_Y + 1$ and $d_Z + 1$ respectively (so that $d_X = d_Y + d_Z$), let the map $G \colon X \to \mathcal{L}(Y)$ take this vector to the reduced density matrix of subsystem $Y$. That is,
\begin{equation}
    G(\ket{a_X}) = \tr_Z \op{\iota(a_X)}{\iota(a_X)}
\end{equation}
where $\iota$ is the embedding defined in Appendix \ref{apdx:permsymmembedding}. We will combine this with the proof of Lipschitz continuity of von Neumann entropy and purity \cite{hayden2006aspects} to adapt them to the permutation symmetric case. In the following discussion, the Lipschitz continuity of functions is considered with respect to the Euclidean norm.
\begin{proposition}
    \begin{enumerate}[i.]
        \item The von Neumann entropy is a Lipschitz continuous function for permutation symmetric states. That is, the function $f\colon X \to \mathbb{R}$ defined as
        \begin{equation*}
            f(\ket{a_X}) = S^{vN}(G(\ket{a_X}))
        \end{equation*}
        is Lipschitz continuous with respect to the Euclidean norm for $d_Y \geq 2$. 
        \item The linear entropy is a Lipschitz continuous function for permutation symmetric states.
        \item The tripartite mutual information between any three subsystems of a permutation symmetric system is Lipschitz continuous, where the TMI is calculated either using von Neumann or linear entropy. For the case of von Neumann entropy, each of the subsystems must have at least two qubits, and so the total system should have at least six qubits.
    \end{enumerate}
\end{proposition}
\begin{proof}
    (i) In Ref.~\cite{hayden2006aspects} it has been shown that given any $\ket{\phi_{Y \otimes Z}} \in Y \otimes Z$, the function $f' \colon Y \otimes Z \to \mathbb{R}$ defined as $f'(\ket{\phi_{Y \otimes Z}}) = S^{vN}(\rho_Y)$ (where $\rho_Y$ is the reduced density matrix of subsystem $Y$) is Lipschitz continuous for $d_Y \geq 2$, with the Lipschitz constant bounded above by $\sqrt{8} \log(d_Y + 1)$. In other words,
    \begin{align*}
       & |f'(\ket{\phi_{Y \otimes Z}}) - f'(\ket{\psi_{Y \otimes Z}})|
                            \leq \\ & \sqrt{8} \log(d_Y + 1)) \norm{\ket{\phi_{Y \otimes Z}} - \ket{\psi_{Y \otimes Z}}}_2
    \end{align*}
    for $d_Y \geq 2$, where $\norm{\cdot}_2$ is the Euclidean norm. As this is valid for any vector in $Y \otimes Z$, it is, in particular, valid for $\iota(\ket{a_X}),\ \iota(\ket{b_X}) \in Y \otimes Z$ given any $\ket{a_X},\ \ket{b_X} \in X$. Noting that $f(\ket{a_X}) = f'(\iota(\ket{a_X})$, we get
    \begin{align*}
        |f(\ket{a_X} - f(\ket{b_X})| &\leq \sqrt{8} \log(d_Y + 1) \norm{\iota(\ket{a_X}) - \iota(\ket{b_X})}_2 \\
                                     &\leq \sqrt{8} \log(d_Y + 1) \norm{\ket{a_X} - \ket{b_X}}_2 \\
                                     &\hspace{0.75cm} \text{ (using linearity and isometry of $\iota$) }
    \end{align*}
    for $d_Y \geq 2$. Thus, $f(\ket{a_X}) = S^{vN}(G(\ket{a_X}))$ is Lipschitz continuous with the Lipschitz constant bounded above by $\sqrt{8} \log(d_Y + 1)$.
    \newline

    (ii) We start by showing that for any $\ket{a_X} \in X$, the function $f(\ket{a_X}) = \sqrt{\tr(G(\ket{a_X})^2)}$ is Lipschitz continuous. Again, we resort to the corresponding result for the case $\ket{\phi_{Y \otimes Z}} \in Y \otimes Z$ in \cite{hayden2006aspects}. In particular, for the function $f' \colon Y \otimes Z \to \mathbb{R}$ defined as $f'(\ket{\phi_{Y \otimes Z}}) = \sqrt{\tr(\rho_Y^2)}$, they showed that
    \begin{equation*}
        |f'(\ket{\phi_{Y \otimes Z}}) - f'(\ket{\psi_{Y \otimes Z}})| \leq 2 \norm{\ket{\phi_{Y \otimes Z}} - \ket{\psi_{Y \otimes Z}}}_2
    \end{equation*}
    As before, taking any $\ket{a_X},\ \ket{b_X} \in X$, we apply the above for $\iota(\ket{a_X}),\ \iota(\ket{b_X}) \in Y \otimes Z$ while noting that $f(\ket{a_X}) = f'(\iota(\ket{a_X}))$ to get
    \begin{equation*}
        |f(\ket{a_X}) - f(\ket{b_X})| \leq 2 \norm{\ket{a_X} - \ket{b_X}}_2
    \end{equation*}
    where we have utilized the linearity and isometry of $\iota$. Now observe that $f$ is bounded above by $1$ since $\tr(\rho^2) \leq 1$ for any density matrix $\rho$. Thus, $|f^2(\ket{a_X}) - f^2(\ket{b_X})| = |f(\ket{a_x}) - f(\ket{b_X})| |f(\ket{a_X}) + f(\ket{b_X})| \leq 4 \norm{\ket{a_X} - \ket{b_X}}_2$, where the triangle inequality, upper bound of $f$ and the above inequality has been used. Therefore, $f^2$ is Lipschitz continuous with the Lipschitz constant bounded above by $4$. As the linear entropy corresponding to the state $\ket{a_X}$ is given by $1 - f^2(\ket{a_X})$, it is Lipschitz continuous with the Lipschitz constant bounded above by $4$.
    \newline

    (iii) Let $A$, $B$, $C$ with dimensions $d_A + 1$, $d_B + 1$ and $d_C + 1$ respectively be subspaces of the space $X$ with dimension $d_X + 1$. We abbreviate the space of the joint permutation symmetric systems as follows: $AB$ having dimension $d_A + d_B + 1$, $BC$ having dimension $d_B + d_C + 1$, $AC$ having dimension $d_A + d_C + 1$ and $ABC$ having dimension $d_A + d_B + d_C + 1$. Further, these spaces are such that $d_A + d_B + d_C \leq d_X$. Now, say $f_Y\colon X \to \mathbb{R}$ is either the von Neumann entropy or the linear entropy (corresponding to a given subsystem $Y$). We know that $f_Y$ is Lipschitz continuous with the Lipschitz constant bounded above by $\eta_Y$, which is equal to $\sqrt{8} \log(d_Y + 1)$ for the case of von Neumann entropy and equal to $4$ for the case of linear entropy.

    Let $f^{TMI}\colon X \to \mathbb{R}$ be defined as 
    \begin{align*}
       & f^{TMI}(\ket{a_X}) = f_A(\ket{a_X}) + f_B(\ket{a_X}) + f_C(\ket{a_X}) \\
                                                  &- f_{AB}(\ket{a_X}) - f_{BC}(\ket{a_X}) - f_{AC}(\ket{a_X}) + f_{ABC}(\ket{a_X})
    \end{align*}
    Then for any vectors $\ket{a_X}, \ket{b_X} \in X$, using triangle inequality and the Lipschitz continuity of $f_Y$, we get
    \begin{align*}
        &|f^{TMI}(\ket{a_X}) - f^{TMI}(\ket{b_X})|
                               \leq \\ &(\eta_A + \eta_B + \eta_C + \eta_{AB} + \eta_{BC} 
                                     + \eta_{AC} + \eta_{ABC}) \norm{\ket{a_X} - \ket{b_X}}_2
    \end{align*}
    Therefore $f^{TMI}$ is Lipschitz continuous with the Lipschitz constant bounded above by $\eta = (\eta_A + \eta_B + \eta_C + \eta_{AB} + \eta_{BC} + \eta_{AC} + \eta_{ABC})$. In other words, the TMI defined on permutation symmetric spaces is Lipschitz continuous for either von Neumann entropy or linear entropy. Note that for Lipschitz continuity of von Neumann entropy, we require $d_Y \geq 2$ for each of the subsystems.
\end{proof}

Now that we have the von Neumann entropy, linear entropy and TMI as Lipschitz continuous functions on permutation symmetric systems, we can resort to L\'evy's lemma. To recall, L\'evy's lemma is stated as follows \cite{hayden2006aspects}.

\begin{lemma}[L\'evy]
    Let $f\colon S^{n - 1} \to \mathbb{R}$ be a Lipschitz continuous function \textnormal{(}with respect to the Euclidean norm in $\mathbb{R}^{n}$\textnormal{)} with a Lipschitz constant $\eta$, where $S^{n - 1}$ is a unit sphere in $\mathbb{R}^{n}$. If $x \in S^{n - 1}$ is picked randomly \textnormal{(}with respect to the uniform measure on the sphere\textnormal{)}, then for $\epsilon \geq 0$ we have
    \begin{equation*}
        \mathbb{P}\{|f(x) - \mathbb{E}[f]| \geq \epsilon\} \leq 2 \exp\left(- \frac{n \epsilon^2}{9 \pi^3 \ln(2) \eta^2}\right)
    \end{equation*}
    where $\mathbb{E}[f]$ is the expectation value of $f$ with respect to the uniform measure on the sphere.
\end{lemma}
Thus, for large enough dimensions, we can say that almost all random permutation symmetric states have von Neumann entropy, linear entropy and TMI close to the respective averages. In particular, a positive average TMI implies that most states also have positive TMI.

\end{document}